\documentclass[aps,pra,showpacs,twocolumn,superscriptaddress, floatfix, nofootinbib]{revtex4-2}%

\usepackage[utf8]{inputenc}
\usepackage[T1]{fontenc} % font encoding
\usepackage{times} % font family: mathpazo
\usepackage[colorlinks,citecolor=red,bookmarks=true]{hyperref}
\usepackage{microtype}
\usepackage{amsmath}

\usepackage{float}
\usepackage{amssymb}
\usepackage{amsthm}
\usepackage{mathrsfs} % required by mathscr
\usepackage{bm}
\usepackage{bbm}
\usepackage{dsfont}
\usepackage{fullpage}
\usepackage{comment}
\usepackage{tikz}
\usepackage{mathtools}
\usepackage[shortlabels]{enumitem}
\usepackage{complexity}
\usepackage[capitalize]{cleveref}
\usepackage{nag}
\usepackage{xspace}
\usepackage{physics}
\usepackage{url}
\usepackage{color}

\usepackage{graphicx}

\usepackage{algpseudocode}
\usepackage{framed}
\usepackage{mdframed}

\mdfdefinestyle{figstyle}{ %
  linecolor=black!8, %
  backgroundcolor=black!8, %
  innertopmargin=.8em, %
  innerleftmargin=.8em, %
  innerrightmargin=.8em, %
  innerbottommargin=.8em %
}

\newfloat{algorithm}{hbt}{lop} % must use \begin{algorithm}[H] in the revtex4-2 documentclass
\floatname{algorithm}{Algorithm} %
\crefname{algorithm}{Algorithm}{Algorithms} %

\newfloat{metaalgorithm}{hbt}{lop} %
\floatname{metaalgorithm}{Meta-Algorithm} %
\crefname{metaalgorithm}{Meta-Algorithm}{Meta-Algorithms} %

\newenvironment{algspec}{ %
  \begin{mdframed}[style=figstyle]}{ %
  \end{mdframed}} %

\theoremstyle{plain}
\newtheorem{theorem}{Theorem}%[section]
\newtheorem{lemma}[theorem]{Lemma}
\newtheorem{corollary}[theorem]{Corollary}
\newtheorem{proposition}[theorem]{Proposition}
\newtheorem{conjecture}[theorem]{Conjecture}
\newtheorem{problem}[theorem]{Problem}

\theoremstyle{definition}
\newtheorem{definition}[theorem]{Definition}

\theoremstyle{remark}

\let\poly\relax
\DeclareMathOperator{\poly}{poly}

\DeclareMathOperator{\range}{range}
\DeclareMathOperator{\diag}{diag}

\def\EE{\mathbb{E}}
\def\SS{\mathbb{S}}

\def\FF{\mathbb{F}}

\newcommand{\supp}{\textsf{supp}}

\def\calC{\mathcal{C}}
\def\calD{\mathcal{D}}

\def\calH{\mathcal{H}}
\def\calX{\mathcal{X}}

\def\calZ{\mathcal{Z}}
\def\calP{\mathcal{P}}

\def\QRC{\mathrm{QRC}}

\def\Ham{\mathsf{Ham}}

\DeclareMathOperator{\Stab}{Stab}
\DeclareMathOperator{\row}{row}
\DeclareMathOperator{\col}{col}

\newcommand{\s}{\vb{s}}
\newcommand{\x}{\vb{x}}

\renewcommand{\mod}[1]{\ \mathrm{mod}\ #1}

\begin{document}

\title{Instantaneous Quantum Polynomial-Time Sampling and Verifiable Quantum Advantage: Stabilizer Scheme and Classical Security}

\author{Michael J. Bremner}
\email{Michael.Bremner@uts.edu.au}
\affiliation{Centre for Quantum Software and Information, University of Technology Sydney, Sydney, New South Wales 2007, Australia }
\affiliation{Centre for Quantum Computation and Communication Technology, New South Wales 2052, Australia}

\author{Bin Cheng}
\email{bincheng@nus.edu.sg}
\affiliation{Centre for Quantum Software and Information, University of Technology Sydney, Sydney, New South Wales 2007, Australia }
\affiliation{Centre for Quantum Computation and Communication Technology, New South Wales 2052, Australia}
\affiliation{Centre for Quantum Technologies, National University of Singapore, Singapore}

\author{Zhengfeng Ji}
\email{jizhengfeng@tsinghua.edu.cn}
\affiliation{Department of Computer Science and Technology, Tsinghua University, Beijing, China}
\affiliation{Zhongguancun Laboratory}

% Include the date command, but leave its argument blank.

%\date{}

%\begin{CJK*}{GBK}{kai}


%\baselineskip24pt


\begin{abstract}
  Sampling problems demonstrating beyond classical computing power with noisy intermediate scale quantum devices have been experimentally realized.
  In those realizations, however, our trust that the quantum devices faithfully
  solve the claimed sampling problems is usually limited to simulations of
  smaller-scale instances and is, therefore, indirect.
  The problem of verifiable quantum advantage aims to resolve this critical
  issue and provides us with greater confidence in a claimed advantage.
  Instantaneous quantum polynomial-time (IQP) sampling has been proposed to
  achieve beyond classical capabilities with a verifiable scheme based on
  quadratic-residue codes (QRC).
  Unfortunately, this verification scheme was recently broken by an attack
  proposed by Kahanamoku-Meyer.
  In this work, we revive IQP-based verifiable quantum advantage by making two
  major contributions.
  Firstly, we introduce a family of IQP sampling protocols called the \emph{stabilizer scheme}, which builds on results linking IQP circuits, the stabilizer formalism, coding theory, and an efficient characterization of IQP circuit correlation functions.
  This construction extends the scope of existing IQP-based schemes while
  maintaining their simplicity and verifiability.
  Secondly, we introduce the \emph{Hidden Structured Code} (HSC) problem as a
  well-defined mathematical challenge that underlies the stabilizer scheme.
  To assess classical security, we explore a class of attacks based on secret
  extraction, including the Kahanamoku-Meyer's attack as a special case.
  We provide evidence of the security of the stabilizer scheme, assuming the
  hardness of the HSC problem.
  We also point out that the vulnerability observed in the original QRC scheme
  is primarily attributed to inappropriate parameter choices, which can be
  naturally rectified with proper parameter settings.
\end{abstract}

\maketitle

%%%%%%%%%%%%%%%%%%%%%%%%%%%%%%%%%%%%%%%%%%%%%%%%

\section{Introduction}\label{sec:introduction}

Quantum computing represents a fundamental paradigm change in the theory of
computation, and promises to achieve quantum speedup in many problems, such as
integer factorization~\cite{Shor-factorization} and database
search~\cite{grover_fast_1996}.
However, many quantum algorithms are designed to be implemented in the
fault-tolerant regime, which are too challenging for our current noisy
intermediate-scale quantum (NISQ) era~\cite{preskill_quantum_2018}.
Experimentally, we can perform random-circuit
sampling~\cite{boixo_characterizing_2018, arute_quantum_2019, zhu_quantum_2021,
  wu_strong_2021} and boson sampling~\cite{aaronson_computational_2011,
  zhong_quantum_2020} at a scale that is arguably beyond the capability of
classical simulation.
But when it comes to verifiability, although these experiments can use some
benchmarking techniques such as cross-entropy
benchmarking~\cite{arute_quantum_2019} to certify the quantum devices, they
cannot be efficiently verified in an adversarial setting without modification of
the underlying computational task.

Classical verification of quantum computation is a long-standing question, which
was first asked by Gottesman~\cite{Gottesman04}.
In the context of verifying arbitrary quantum computation, there have been a
plethora of important results~\cite{Broadbent2008, Broadbent2010-blind-original,
  aharonov_interactive_2010, Aharonov2017, Ji2016-classical, Fitzsimons2017-verifiable-blind,
  Fitzsimons2018, Reichardt2013-blind-vazirani, mahadev_classical_2018}.
The more relevant context of this work is generating a test of quantumness. 
The goal is to create a computational task that is beyond the capabilities of
classical computing, while using minimal quantum and classical computing to
generate and verify.
A motivating example is given by Shor's algorithm for integer
factorization~\cite{Shor-factorization}, which is appealing in that hard
instances can be easily generated and verified classically yet finding the
solution is beyond the capabilities of classical computers.
However, this also has the drawback that the quantum solution also seems to be
beyond the capabilities of NISQ devices.

Recently, there have been tests of quantumness that combine the power of both interactive
proofs and cryptographic primitives~\cite{brakerski_cryptographic_2018,
  brakerski_simpler_2020, kahanamoku-meyer_classically_2022}.
This class of cryptographic verification protocols usually uses a
primitive called trapdoor claw-free (TCF) functions, which has the following
properties.
First, it is a 2-to-1 function that is hard to invert, meaning that given
$y = f(\vb{x}) = f(\vb{x}')$, it is hard for an efficient classical computer to find the
preimage pair $(\vb{x}, \vb{x}')$.
Second, given a trapdoor to the function $f(\vb{x})$, the preimage pair can be
efficiently found on a classical computer.
We will refer to this class of verification protocols as the TCF-based
protocols.
The TCF-based protocols require the quantum prover to prepare the state of the
form $\sum_{\vb{x}} \ket{\vb{x}} \ket{f(\vb{x})}$.
Although a recent experiment implemented a small-scale TCF-based protocol on a
trapped-ion platform~\cite{zhu_interactive_2023}, implementing this class of
protocols is still very challenging for the current technology.

Another class of verification protocols is based on instantaneous quantum
polynomial-time (IQP) circuits initiated by Shepherd and
Bremner~\cite{shepherd_temporally_2009}.
IQP circuits are a family of quantum circuits that employ only commuting gates,
typically diagonal in the Pauli-$X$ basis.
In IQP-based verification protocols, the verifier generates a pair consisting of
an IQP circuit $U_{\mathrm{IQP}}$ and a secret key $\vb{s} \in {\{0, 1\}}^n$.
After transmitting the classical description of the IQP circuit to the prover,
the verifier requests measurement outcomes in the computational basis.
Then, the verifier uses the secret to determine whether the measurement outcomes
are from a real quantum computer.
Such a challenge seems hard for classical computers, as random IQP circuits are
believed to be computationally difficult to simulate classically with minimal
physical resources, assuming some plausible complexity-theoretic assumptions
such as the non-collapse of polynomial
hierarchy~\cite{bremner_classical_2011,bremner_average-case_2016,bremner2017achieving}.

The use of random IQP circuits for the verification protocol is however problematic due
to the anti-concentration property~\cite{yung_anti-forging_2020,
  bremner_average-case_2016}.
To address this issue, the Shepherd-Bremner scheme employs an obfuscated
quadratic-residue code (QRC) to construct the pair
$(U_{\mathrm{IQP}}, \vb{s})$~\cite{macwilliams1977book}.
While the Shepherd-Bremner scheme was experimentally attractive, it suffered
from a drawback as its cryptographic assumptions were non-standard and lacked
sufficient study compared to TCF-based protocols.
This was especially apparent when in 2019, Kahanamoku-Meyer discovered a loophole in the Shepherd-Bremner scheme, enabling a
classical prover to efficiently find the secret, which subsequently allows the
prover to generate data to spoof the test~\cite{kahanamoku-meyer_forging_2023}.
Given the potential of IQP-based protocols to achieve verifiability beyond
classical computing using fewer resources than, say, Shor's algorithm, it is
crucial to investigate the possibility of extending and rectifying the
Shepherd-Bremner construction.

In this work, we propose a new IQP-based protocol, which we refer to as the
\emph{stabilizer scheme}.
Our construction allows the verifier to efficiently generate an IQP circuit,
$U_{\mathrm{IQP}} = e^{i \pi H/8}$, and a secret, $\mathbf{s}$, so that the
correlation function relative to the secret has a magnitude equal to $2^{-g/2}$,
where $g$ is a tunable integer.
The stabilizer scheme is based on the interplay between IQP circuits, stabilizer
formalism and coding theory, and it significantly strengthens previous
constructions based on quadratic-residue codes~\cite{shepherd_temporally_2009}
or random small IQP circuits~\cite{yung_anti-forging_2020}.
Our characterization on IQP circuits builds upon and integrates several previous
results~\cite{shepherd_binary_2010,mann_simulating_2021}, which tackle this
problem from the perspective of binary matroids and Tutte polynomials.
In order to explore the classical security, we formulate the \emph{Hidden
  Structured Code} problem, which captures the hardness of classical attacks
based on secret extraction.
Then, we investigate a general class of such classical attacks, which includes
Kahanamoku-Meyer's attack as an instance.
We give positive evidence that this class of classical attacks takes exponential
time to generate the data with correct correlation relative to the secret.

Specifically, we show that a generalization of Kahanamoku-Meyer's attack, named
Linearity Attack, fails to break the stabilizer scheme if the parameters are
chosen appropriately.
Additionally, we have designed a new obfuscation technique called \emph{column
  redundancy}, which can even be used to fix the recently found weakness in the
Shepherd-Bremner construction~\cite{snoyman_comment}.
Specifically, Claim~3.1 in Ref.~\cite{kahanamoku-meyer_forging_2023} states that the attack algorithm for the Shepherd-Bremner construction succeeds in $O(n^3)$ time on average, which turns out to be true only under certain parameter choices.
This can be naturally rectified with proper parameter settings enabled by our column redundancy technique.
{We also discuss the recent new classical attacks proposed by Gross and Hangleiter~\cite{gross_secret_2023}, which impose further structural constraints on our construction. We then adjust our construction according to the new constraints and discuss its security against all known attacks.}
Our results provide positive evidence for the security of the IQP-based
verification protocols.

This paper is organized as follows.
In the rest of the Introduction, we first give the general framework of
IQP-based verification protocols.
Then, we state our main results in more detail, followed by discussing the
related works.
In \cref{sec:preliminaries}, we give the preliminaries, including
stabilizer formalism, necessary results from coding theory and the
Shepherd-Bremner construction.
In \cref{sec:stabilizer_characterization_IQP}, we give the
characterization of the state generated by IQP circuits with $\theta = \pi/4$
and the correlation function $\ev{\calZ_{\vb{s}}}$ with $\theta = \pi/8$.
Then, in \cref{sec:stab_construction}, we present the stabilizer construction
for the IQP-based protocols.
In \cref{sec:implementation}, we study the implementation of our protocol on real quantum devices, with a goal to reduce gate complexity and circuit depth.
In \cref{sec:classical_attack_security}, we analyze the classical
security of the stabilizer scheme and explore the classical attacks based
on secret extraction.
Finally, we conclude and give open problems in \cref{sec:discussion}.

\subsection{IQP-based verification protocol}\label{subsec:IQP_protocol}

Here, we focus on a specific family of IQP circuits, the $X$
program~\cite{shepherd_temporally_2009}, where all local gates are diagonal in
the Pauli-$X$ basis.
An $X$-program IQP circuit can be specified by a binary matrix $\vb{H} \in \FF_2^{m\times n}$, where $m$ is the number of gates and $n$ is the number of qubits.
The binary matrix $H$ is first transformed into an IQP Hamiltonian $H$ by
\begin{align}
  H = \Ham(\vb{H}) = \sum_{\vb{p}^T \in \row(\vb{H})} \calX_{\vb{p}} \; ,
\end{align}
where $\row(\vb{H})$ is the set of rows of $\vb{H}$, and $\calX_{\vb{p}} := X_1^{p_1} \otimes X_2^{p_2} \otimes \cdots \otimes X_n^{p_n}$.
Then, the IQP circuit is given by $U_{\vb{H}, \theta} := e^{i\theta \Ham(\vb{H})}$.
Since terms in $\Ham(\vb{H})$ commutes with each other, we have 
\begin{align}
  U_{\vb{H}, \theta} = \prod_{\vb{p}^T \in \row(\vb{H})} e^{i\theta \calX_{\vb{p}}} \ .
\end{align}
In the general case, the evolution time for each term in $\Ham(\vb{H})$ can even be different,
but we focus on the case where $\theta = \pi/8$ for all terms in this work.

\begin{figure*}[t]
  \centering %
  \includegraphics[width = 0.8\textwidth]{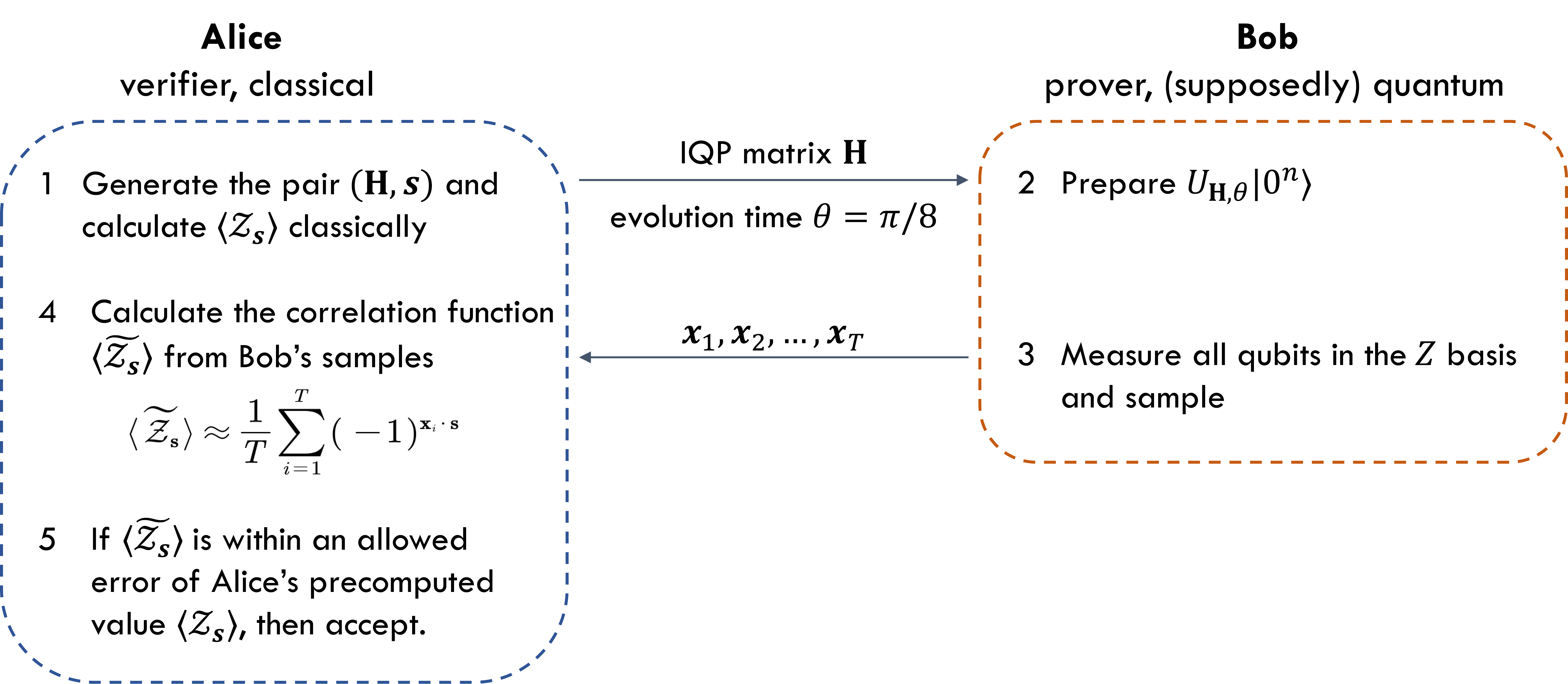}
  \caption{Schematic for IQP-based verification protocol in the case
    $\theta = \pi/8$.}\label{fig:IQP_protocol}
\end{figure*}

\paragraph{General framework.}
The general framework for the IQP-based verification protocol is shown in
Fig.~\ref{fig:IQP_protocol}.
Here, the verifier first generates the pair of an IQP matrix $\vb{H}$ and a
secret $\vb{s} \in {\{0, 1\}}^n$.
She computes the correlation function
$\ev{\calZ_{\vb{s}}} := \mel{0^n}{U_{\vb{H}, \theta}^{\dagger} \calZ_{\vb{s}} U_{\vb{H}, \theta}}{0^n}$
with respect to the secret, which can be achieved classically
efficiently~\cite{yung_anti-forging_2020,shepherd_binary_2010}.
Then, the IQP matrix $\vb{H}$ is sent to the prover,
while the secret is kept on the verifier's side.
The verifier also instructs the prover the evolution time for each term of the
Hamiltonian.
After that, the prover repeatedly prepares the state $U_{\vb{H}, \theta} \ket{0^n}$,
measures all qubits in the computational basis, and obtains a set of samples
$\vb{x}_1, \ldots, \vb{x}_T \in {\{0, 1\}}^n$, which will be sent back to the
verifier.
From the prover's measurement samples, the verifier estimates the correlation
function with respect to $\vb{s}$ by
\begin{align}\label{eq:Z_s_estimation}
  \ev*{\widetilde{\calZ_{\vb{s}}}}
  := \frac{1}{T} \sum_{i=1}^T {(-1)}^{\vb{x}_i \cdot \vb{s}} \; .
\end{align}
If the value of $\ev*{\widetilde{\calZ_{\vb{s}}}}$ is within an allowed error of
the ideal value $\ev{\calZ_{\vb{s}}}$, then the verifier accepts the result and
the prover passes the verification.

In order to ensure the effectiveness of the verification process, two key
challenges must be addressed.
The first one is to evaluate the ideal correlation function, so that the
verifier can compare it with the value obtained from the prover's measurement
outcomes.
The second one is to design a suitable pair $(\vb{H}, \vb{s})$, so that the
correlation function $\ev{\calZ_{\vb{s}}}$ is sufficiently away from zero.
Otherwise, the verifier may need to request a super-polynomial number of samples
from the prover to make the statistical error small enough, making the protocol
inefficient.

\paragraph{Evaluating the correlation function.}
To evaluate the correlation function, first note that the IQP matrix $\vb{H}$ consists of two submatrices $\vb{H}_{\vb{s}}$ and $\vb{R}_{\vb{s}}$, where $\vb{H}_{\vb{s}} \ \vb{s} = \vb{1}$ and $\vb{R}_{\vb{s}} \ \vb{s} = \vb{0}$.
Correspondingly, we have $\Ham(\vb{H}) = \Ham(\vb{H}_{\vb{s}}) + \Ham(\vb{R}_{\vb{s}})$, where $\Ham(\vb{H}_{\vb{s}})$ anti-commutes with $\calZ_{\s}$ and $\Ham(\vb{R}_{\vb{s}})$ commutes with $\calZ_{\s}$, i.e.,
\begin{align}
    \{\Ham(\vb{H}_{\vb{s}}), \calZ_{\vb{s}}\} &= 0 \ , & [\Ham(\vb{R}_{\vb{s}}), \calZ_{\vb{s}}] &= 0 \ .
\end{align}
Due to these commutation relations, the value of the correction function only
depends on the $\vb{H}_{\vb{s}}$,
i.e.~\cite{yung_anti-forging_2020,shepherd_binary_2010},
\begin{align}\label{eq:correlation_function}
  \ev{\calZ_{\s}} = \mel{0^n}{e^{i 2\theta \Ham(\vb{H}_{\vb{s}})} }{0^n} = \mel{0^n}{U_{\vb{H}_{\vb{s}}, 2\theta} }{0^n} \; .
\end{align}
Then, one can observe an intriguing point from this expression.
When $\theta = \pi/8$, the IQP circuit is non-Clifford and there is
complexity-theoretic evidence that the IQP circuits in this setting is hard to
simulate classically~\cite{bremner_average-case_2016}.
However, $U_{\vb{H}_{\vb{s}}, 2\theta} = e^{i 2\theta \Ham(\vb{H}_{\vb{s}})}$ becomes a Clifford circuit, which means that
the correlation function can be computed classically efficiently!
Indeed, $ \ev{\calZ_{\s}} = \mel{0^n}{U_{\vb{H}_{\vb{s}}, 2\theta} }{0^n}$ actually
corresponds to an amplitude of the Clifford circuit $U_{\vb{H}_{\vb{s}}, 2\theta}$.
In this way, the verifier can evaluate the correlation function efficiently
using the Gottesman-Knill algorithm~\cite{Gottesman98}.

\subsection{Main results}\label{subsec:main_results}

In this subsection, we briefly overview the main results of the paper in the
following and refer the reader to later sections for the detailed analysis.
The main objective of this work is to devise a new scheme of the IQP-based
verification protocol that strengthens its classical security and invalidates
the known attacks.
To achieve this, we start by studying the properties of the state
$U_{\vb{H}, \pi/4} \ket{0^n}$.
Given a binary matrix $\vb{H} = (\vb{c}_1, \ldots, \vb{c}_n)$, the stabilizer tableau of
$U_{\vb{H}, \pi/4} \ket{0^n}$ is given by $(\vb{G}, \vb{I}_n, \vb{r})$ (Theorem~\ref{thm:IQP_tableau}),
where the $X$ part is a Gram matrix $\vb{G} = \vb{H}^T \vb{H}$, the $Z$ part is
an identity matrix $\vb{I}_n$, and the phase column $\vb{r}$ depends on the
Hamming weight of columns in $\vb{H}$.

Next, we compute the correlation function $|\ev{\calZ_{\vb{s}}}|$ and connect it
to a property of the code $\calC_{\vb{s}}$ generated by columns of
$\vb{H}_{\vb{s}}$.
Let $\calC_{\vb{s}}^{\perp}$ be the dual code of $\calC_{\vb{s}}$,
$\calD_{\vb{s}} := \calC_{\vb{s}} \bigcap \calC_{\vb{s}}^{\perp}$ be the
self-dual intersection and consider
$g := \dim(\calC_{\vb{s}}) - \dim(\calD_{\vb{s}})$.
We then prove in Theorem~\ref{thm:cor_func_rank} that the magnitude of the
correlation function $|\ev{\calZ_{\vb{s}}}|$ is $2^{-g/2}$ if the self-dual
intersection $\calD_{\vb{s}}$ is a doubly-even code, and $0$ if it is an
unbiased even code.
Moreover, it can be proved that $\calD_{\vb{s}}$ must be in one of the two
cases, and thus the above gives a complete characterization of the magnitude of
the correlation function.
Interestingly, the $g$ number happens to be the rank of the Gram matrix
$\vb{G}_{\vb{s}} = \vb{H}_{\vb{s}}^T \vb{H}_{\vb{s}}$ associated with
$\vb{H}_{\vb{s}}$ (Proposition~\ref{prop:Gram_matrix_and_code}), which also
characterizes the overlap between $\ket{0^n}$ and
$U_{\vb{H}_{\vb{s}}, \pi/4} \ket{0^n}$ from a group-theoretic perspective
(Proposition~\ref{prop:overlap}).
\cref{thm:cor_func_rank} is an effective merging of a number of results that were first discussed by Shepherd in Ref.~\cite{shepherd_binary_2010}, with a particular focus on coding theory.
Originally, Shepherd studied IQP circuits from the perspective of binary matroids, codes, and Tutte polynomials. 

With these results established, the construction of $(\vb{H}, \vb{s})$ for the
verification protocol can be formulated as follows.
Let $\calH_{n, m, g} = \{(\vb{H}, \vb{s})\}$ be a family of pairs of an IQP
matrix $\vb{H} \in \FF_2^{m\times n}$ and a secret $\vb{s} \in \FF_2^n$ so that
the corresponding correlation function satisfies
$|\ev{\calZ_{\vb{s}}}| = 2^{-g/2}$; the precise definition is presented in
Definition~\ref{def:IQP_family}.
Here, the parameters $n$ and $m$ correspond to the size of the IQP circuits, and
$g$ corresponds to the value of the correlation function relative to the secret.
We give an efficient algorithm to sample random instances from
$\calH_{n, m, g}$, which we call the stabilizer construction
(\cref{alg:stabilizer_construction}).

Essentially, the stabilizer construction is to randomly generate an obfuscated
code and a secret, so that the corresponding correlation function is
sufficiently away from zero, to enable efficient verification.
Specifically, we reduce this problem to sampling two matrices $\vb{D}$ and
$\vb{F}$, so that $\vb{D}$ is a generator matrix of a random doubly-even code,
and $\vb{F}$ consists of $g$ random columns satisfying the constraints
$\vb{D}^T \vb{F} = \vb{0}$ and $\rank(\vb{F}^T \vb{F}) = g$.
Jointly, columns in $\vb{D}$ and $\vb{F}$ span a linear subspace that contains
the all-ones vector, which must be a codeword because
$\vb{H}_{\vb{s}}\ \vb{s} = \vb{1}$.
We give an efficient algorithm to sample such matrices $\vb{D}$ and $\vb{F}$.

A naive implementation of the circuit $U_{\vb{H}, \theta}$ requires applying gates of the form $e^{i\theta \calX_{\vb{p}}}$, which can be highly non-local.
To resolve this issue, we show that $U_{\vb{H}, \theta}$ can be compiled into an equivalent circuit consisting of CNOT circuits interleaving with rounds of $e^{i\theta X}$ gates only.
This implies that implementing the IQP circuits of our protocol requires $O(n^2/\log{n})$ CNOT gates and $O(n/\log{n})$ depth.

To explore the classical security, we consider a general class of classical
attacks based on secret extraction.
Given $(\vb{H}, \vb{s}) \in \calH_{n, m, g}$, extracting the secret $\vb{s}$
from $\vb{H}$ leads to finding the hidden code $\calC_{\vb{s}}$ from a larger
obfuscated code.
Such a hidden substructure problem seems hard for a classical computer, and we
formulate the following conjecture.

\begin{conjecture}[Hidden Structured Code (HSC) Problem]\label{conj:hsc-intro}
  For certain appropriate choices of $n, m, g$, there exists an efficiently
  samplable distribution over instances $(\vb{H}, \vb{s})$ from the family
  $\calH_{n, m, g}$, so that no polynomial-time classical algorithm can find the
  secret $\vb{s}$ given $n, m$ and $\vb{H}$ as input, with high probability over
  the distribution on $\calH_{n, m, g}$.
\end{conjecture}

To support this conjecture, we extend Kahanamoku-Meyer's attack to target
general IQP circuits with $\theta = \pi/8$, and we call this attack the
Linearity Attack.
This generalized attack uses linear algebraic techniques to search for a
candidate set of secrets, and performs classical sampling according to this
candidate set.
By choosing appropriate parameters, random instances drawn by our stabilizer
scheme turns out to invalidate the Linearity Attack, since the search for the
candidate set takes exponential time.
As a result, the stabilizer scheme is secure against the Linearity Attack. 
Moreover, our analysis suggests that choosing a different set of parameters for
the QRC-based construction can fix the recent loophole in the original
Shepherd-Bremner scheme.
This refutes the Claim~3.1 in Ref.~\cite{kahanamoku-meyer_forging_2023}, which states that the QRC-based construction can be efficiently broken classically in general.
{In response to the new classical attacks by Gross and Hangleiter~\cite{gross_secret_2023}, we take their analysis as a guideline and suggest a possible construction under the framework that we developed.}

\subsection{Related works}\label{subsec:related_works}

The first explicit construction recipe of $(\vb{H}, \vb{s})$ for the case
$\theta = \pi/8$ is given by Shepherd and
Bremner~\cite{shepherd_temporally_2009}.
In the their construction, $\vb{H}_{\vb{s}}$ is constructed from a specific
error-correcting code, the quadratic-residue code
(QRC)~\cite{macwilliams1977book}, which guarantees that the correlation function
is always $1/\sqrt{2}$, a value sufficiently away from zero as desired.
Formally, let $\calH^{\QRC}_{n, m, q} = \{(\vb{H}, \vb{s})\}$ be a family of
pairs of an IQP matrix $\vb{H} \in \FF_2^{m\times n}$ and a secret $\vb{s}$ so
that $\vb{H}_{\vb{s}}$ generates a QRC of length $q$ (up to row permutations)
and $\vb{H}$ is of full column rank.
What the Shepherd-Bremner construction achieves is to randomly sample instances
from $\calH^{\QRC}_{n, m, q}$, where $n = (q+3)/2$.

However, it turns out that this set of parameters can only give easy instances.
In Ref.~\cite{kahanamoku-meyer_forging_2023}, Kahanamoku-Meyer gave a
secret-extraction attack (KM attack) against the Shepherd-Bremner construction.
With his attack, a classical prover can find the secret $\vb{s}$ efficiently
with high probability.
Once the secret is found, the prover can easily pass the test by generating
appropriately biased data in the direction of the secret, without the need of
actually simulating the IQP circuits.
In Ref.~\cite{yung_anti-forging_2020}, Yung and Cheng proposed to circumvent the
attack by starting with a small randomized IQP circuit and using the obfuscation
technique in the Shepherd-Bremner scheme to hide that small IQP
circuit~\cite{shepherd_temporally_2009}.
The verifier cannot directly use a fully randomized IQP circuit because the
correlation function will be close to zero for most choices of secrets in that
case, due to the anti-concentration property of IQP
circuits~\cite{bremner_average-case_2016}.
Small correlation functions make it difficult for the verifier to distinguish
between an honest quantum prover and a cheating classical prover outputting
random bit strings.
This poses a challenge, to balance the security given by randomized
constructions with the scale of the correlation functions that enables easy
verification.
This challenge is not fully resolved by the heuristic construction in
Ref.~\cite{yung_anti-forging_2020}.

In addition, Shepherd studied IQP circuits with tools of binary matroids and
Tutte polynomials, and derived some related results to this
work~\cite{shepherd_binary_2010}.
Specifically, the amplitude of the IQP circuit $\mel{0^n}{e^{i\theta H}}{0^n}$
is expressed in terms of the normalized Tutte polynomial, and its computational
complexity is studied in various cases.
When $\theta = \pi/4$, the magnitude of the related Tutte polynomial can be
efficiently evaluated using Vertigan's algorithm~\cite{vertigan_bicycle_1998},
which is similar to the Gottesman-Knill algorithm~\cite{Gottesman98}.
This idea was further explored by Mann~\cite{mann_simulating_2021}, which
related computing the amplitude to the bicycle dimension and the Brown's
invariant using results of Ref.~\cite{pendavingh_evaluation_2014}.
But when $\theta = \pi/8$ (and any other values except for the multiple of
$\pi/4$), computing the amplitude is $\#P$-hard in the worst case.
Moreover, Ref.~\cite{shepherd_binary_2010} also derived similar relation to
Eq.~\eqref{eq:correlation_function}, in the language of the normalized Tutte
polynomial.
Therefore, it was proved that the correlation function is efficiently classical
computable when $\theta = \pi/8$, and suggests that this could be used to
perform hypothesis test for access to quantum computers, although no new
construction was proposed in Ref.~\cite{shepherd_binary_2010}.

\section{Preliminaries}\label{sec:preliminaries}

\subsection{Notations}

We mainly work on the field $\FF_2$.
We use bold capital letters such as $\vb{H}$ to denote a matrix and bold
lower-case letters such as $\vb{s}$ to denote a vector.
If not stated otherwise, a vector is referred to as a column vector, and a row
vector will be added the transpose symbol, like $\vb{p}^T$.
The (Hamming) weight of a vector $\vb{x}$ is denoted as $|\vb{x}|$.
The inner product between two vectors $\vb{x}$ and $\vb{s}$ is denoted as
$\vb{x} \cdot \vb{s}$; sometimes we will also use $\vb{H} \cdot \vb{s}$ to
denote the matrix multiplication.
We use $\col(\vb{H})$ and $\row(\vb{H})$ to denote the set of columns and rows of a
matrix $\vb{H}$, respectively.
We use $c(\vb{\vb{H}})$ and $r(\vb{H})$ to denote the number of columns and the
number of rows of a matrix $\vb{H}$, respectively.
{The column space of a matrix $\vb{H}$ is denoted by $\range(\vb{H})$.}
The rank of a matrix $\vb{H}$ is denoted as $\rank(\vb{H})$.
{If not otherwise stated, full rank means full column rank.}
We use $\ker(\vb{H})$ to denote the kernel space of $\vb{H}$, i.e., the space of
vectors $\vb{v}$ such that $\vb{Hv} = \vb{0}$.
We call two square matrices $\vb{A}$ and $\vb{B}$ congruent if there exists an
invertible matrix $\vb{Q}$ satisfying $\vb{A} = \vb{Q}^T \vb{B} \vb{Q}$, denoted
as $\vb{A} \sim_c \vb{B}$.
We call such an transformation \emph{congruent transformation}.

The all-ones vector will be denoted as $\vb{1}$, with its dimension inspected
from the context; the similar rule applies to the all-zeros vector (or matrix)
$\vb{0}$.
We define $[n] := \{1, 2, \ldots, n\}$.
{Given $S \subseteq [n]$, define $\vb{1}_S$ to be a vector with 1's at the indices in $S$ and 0's elsewhere.}
The $n \times n$ identity matrix is denoted as $\vb{I}_n$.
For a vector $\vb{x}$, we define its support as $\supp(\vb{x}) := \{j: x_j = 1\}$.
If not stated otherwise, a full-rank matrix is referred to a matrix with full
column rank.

We denote the linear subspace spanned by a set of vectors
$\{\vb{c}_1, \ldots, \vb{c}_k\}$ as
$\langle \vb{c}_1, \ldots, \vb{c}_k \rangle$.
Given linear subspaces $V = \langle \vb{c}_1, \ldots, \vb{c}_l \rangle$ and
$U = \langle \vb{c}_1, \ldots, \vb{c}_k \rangle$ with $k < l$, we denote the
complement subspace of $U$ in $V$ with respect to the basis
$\{ \vb{c}, \ldots, \vb{c}_l \}$ by $(V \slash U)_{\vb{c}_1, \ldots, \vb{c}_l}$;
namely,
$(V \slash U)_{\vb{c}_1, \ldots, \vb{c}_l} := \langle \vb{c}_{k+1}, \ldots, \vb{c}_l \rangle$.
Usually, we are not interested in a specific basis, so we use $V \slash U$ to
denote a random complement subspace of $U$ in $V$, i.e.,
$V \slash U \gets_{\mathcal{R}} \{ \langle \vb{c}_{k+1}, \ldots, \vb{c}_l \rangle : V = \langle \vb{c}_1, \ldots, \vb{c}_l \rangle, U = \langle \vb{c}_1, \ldots, \vb{c}_k \rangle \}$,
where $\gets_{\mathcal{R}}$ denotes a random instance from a set.
We let $V \backslash U := \{ \vb{v}: \vb{v} \in V, \vb{v} \not\in U\}$ be the
ordinary complement of two sets.

\subsection{Stabilizer formalism}\label{subsec:stab_formalism}

\paragraph{Overlap of two stabilizer states.}

Given two stabilizer states $\ket{\psi}$ and $\ket{\phi}$, let
$\Stab(\ket{\psi})$ and $\Stab(\ket{\phi})$ be their stabilizer groups,
respectively, which are subgroups of the $n$-qubit Pauli group.
Let $\{P_1, \ldots, P_n\}$ be the generators of $\Stab(\ket{\psi})$ and
$\{Q_1, \ldots, Q_n\}$ be those of $\Stab(\ket{\phi})$.
Note that the set of generators is not unique.
Then, the overlap $|\ip{\psi}{\phi}|$ is determined by their stabilizer
groups~\cite{aaronson_improved_2004}.

\begin{proposition}[\cite{aaronson_improved_2004}]\label{prop:overlap}
  Let $\ket{\psi}$ and $\ket{\phi}$ be two stabilizer states.
  Then, $\ip{\psi}{\phi} = 0$ if their stabilizer groups contain the same Pauli
  operator of the opposite sign.
  Otherwise, $\abs{\ip{\psi}{\phi}} = 2^{-g/2}$, where $g$ is the minimum number
  of different generators over all possible choices.
\end{proposition}
For completeness, we provide an alternative proof in Appendix~\ref{app:overlap}.
In particular, this implies that
$\ev{\calZ_{\vb{s}}} = \mel{0^n}{U_{\vb{H}_{\vb{s}}, \pi/4}}{0^n}$ has magnitude
either 0 or $2^{-g/2}$, where $n-g$ is the maximum number of independent
Pauli-$Z$ products in the stabilizer group of $U_{\vb{H}_{\vb{s}}, \pi/4} \ket{0^n}$.

\paragraph{Tableau representation.}
A stabilizer state or circuit can be represented by a stabilizer tableau, which
is an $n$-by-$(2n+1)$ binary matrix.
The idea is to use $2n+1$ bits to represent each generator of the stabilizer
group.
First, a single-qubit Pauli operator can be represented by $(x, z)$; $(0, 0)$
corresponds to $I$, $(1, 0)$ corresponds to $X$, $(0, 1)$ corresponds to $Z$ and
$(1, 1)$ corresponds to $Y$.
For stabilizer generators, the phase can only be $\pm 1$ since the stabilizer
group does not contain $-I$.
So, one can use an extra bit $r$ to represent the phase; $r = 0$ is for $+1$
while $r = 1$ is for $-1$.
Then, an $n$-qubit stabilizer generator can be represented by $2n+1$ bits,
\begin{align}
    (x_1, \ldots, x_n, z_1, \ldots, z_n, r) \; .
\end{align}
For example, the vector for $-X_1 Z_2$ is $(1, 0, 0, 1, 1)$.
Any stabilizer state can be specified by $n$ stabilizer generators, which
commute with each other.
Therefore, the state is associated with the following tableau,
\begin{align}
  \begin{pmatrix}
    x_{11} & \cdots & x_{1n} & z_{11} & \cdots & z_{1n} & r_1 \\
    \vdots & \ddots & \vdots & \vdots & \ddots & \vdots & \vdots \\
    x_{n1} & \cdots & x_{nn} & z_{n1} & \cdots & z_{nn} & r_n \\
  \end{pmatrix} \; ,
\end{align}
whose rows define the stabilizer generators.
The first $n$ columns are called the $X$ part, the $(n+1)$-th to $2n$-th columns
are called the $Z$ part, and the last column are called the phase column of the
stabilizer tableau.
As an example, the $\ket{0^n}$ state is stabilized by
$\langle Z_1, \ldots, Z_n \rangle$, and its stabilizer tableau is given by,
\begin{align}\label{eq:standard_tab_0}
    \begin{pmatrix}
        0 & \cdots & 0 & 1 & \cdots & 0 & 0 \\
        \vdots & \ddots & \vdots & \vdots & \ddots & \vdots & \vdots \\
        0 & \cdots & 0 & 0 & \cdots & 1 & 0 \\
    \end{pmatrix} \; .
\end{align}
We will call it the standard stabilizer tableau of $\ket{0^n}$. 

\subsection{Coding theory}\label{subsec:coding_theory}

We present some results regarding coding theory here, with the proof presented
in Appendix~\ref{app:coding_theory}.
We only consider linear codes over $\FF_2$ in this paper. 
A linear code, or simply a code $\calC$ of length $m$ is a linear subspace of
$\FF_2^m$.
One can use a generator matrix $\vb{H}$ to represent a code, with its columns
spanning the codespace $\calC$.
The dual code is defined as
$\calC^{\perp} := \{ \vb{v} \in \FF_2^m: \vb{v} \cdot \vb{w} = 0 \text{ for
} \vb{w} \in \calC\}$.
The dual code of a linear code is also a linear code.
It is not hard to see that $\calC^{\perp} = \ker(\vb{H}^T)$, which implies
$\dim(\calC) + \dim(\calC^{\perp}) = m$.
A code $\calC$ is weakly self-dual if $\calC \subseteq \calC^{\perp}$ and
(strictly) self-dual if $\calC = \calC^{\perp}$, in which case
$\dim(\calC) = m/2$.

A code $\calC$ is an even code if all codewords have even Hamming weight and a
doubly-even code if all codewords have Hamming weight a multiple of $4$.
It is not hard to show that a doubly-even code is a weakly self-dual code.
Moreover, we have the following proposition.
\begin{proposition}\label{prop:all_one_vector}
  The all-ones vector is a codeword of $\calC$ if and only if its dual code
  $\calC^{\perp}$ is an even code.
\end{proposition}
We define the notion of (un)biased even codes, which will be useful in the
stabilizer characterization of IQP circuits
(\cref{sec:stabilizer_characterization_IQP}).
\begin{definition}
  A code $\calC$ is called a \emph{biased even code} if it is an even code where
  the number of codewords with Hamming weight $0$ and $2$ modulo $4$ are not
  equal.
  It is called an \emph{unbiased even code} otherwise.
\end{definition}

Let the (maximum) self-dual subspace of $\calC$ be
$\calD := \calC \bigcap \calC^{\perp}$, which is itself a weakly self-dual code.
Note that $\calD$ must be an even code, since all codewords are orthogonal to
themselves and hence have even Hamming weight.
We have the following lemma.

\begin{lemma}\label{lemma:self_dual_even_code}
  A weakly self-dual even code is either a doubly-even code or an unbiased even
  code.
  For the former case, all columns of its generator matrix have weight 0 modulo
  4 and are orthogonal to each other.
  For the latter case, there is at least one column in the generator matrix with
  weight 2 modulo 4.
\end{lemma}

One can apply a basis change to the generator matrix $\vb{H}$, resulting in
$\vb{HQ}$, where $\vb{Q}$ is an invertible matrix.
This will not change the code $\calC$.
Define the Gram matrix of the generator matrix by $\vb{G} := \vb{H}^T \vb{H}$.
A basis change on $\vb{H}$ transforms $\vb{G}$ into $\vb{Q}^T \vb{G} \vb{Q}$,
which is a congruent transformation.
The rank of Gram matrix is also an invariant under basis change.
It can be related to the code $\calC$ in the following way.

\begin{proposition}\label{prop:Gram_matrix_and_code}
  Given a generator matrix $\vb{H}$, let its Gram matrix be
  $\vb{G} = \vb{H}^T \vb{H}$ and the generated code be $\calC$.
  Let $\calD = \calC \bigcap \calC^{\perp}$, where $\calC^{\perp}$ is the dual
  code of $\calC$.
  Then, $\rank(\vb{G}) = \dim(\calC) - \dim(\calD)$.
\end{proposition}

\subsection{Shepherd-Bremner construction}\label{subsec:SB_construction}

In the Shepherd-Bremner construction, the part $\vb{H}_{\vb{s}}$ is constructed
from the quadratic-residue code.
The quadratic residue code is a cyclic code.
Its cyclic generator has 1 in the $j$-th position if $j$ is a non-zero quadratic
residue modulo $q$.
The size parameter $q$ of the QRC is a prime number and $q+1$ is required to be
a multiple of eight~\cite{shepherd_temporally_2009}.
For $q = 7$, the cyclic generator reads ${(1, 1, 0, 1, 0, 0, 0)}^T$, because
$j = 1, 2, 4$ are quadratic residues modulo 7.
The basis for the codespace of QRC is generated by rotating the cyclic
generator, which is the last $4$ columns of the following matrix,
\begin{align}\label{eq:QRC_main_part}
  \vb{H}_{\vb{s}}^{\rm QRC} =
  \begin{pmatrix}
    1 & 1 & 0 & 0 & 0 \\
    1 & 1 & 1 & 0 & 0 \\
    1 & 0 & 1 & 1 & 0 \\
    1 & 1 & 0 & 1 & 1 \\
    1 & 0 & 1 & 0 & 1 \\
    1 & 0 & 0 & 1 & 0 \\
    1 & 0 & 0 & 0 & 1
  \end{pmatrix} \; .
\end{align}
The first column is added so that the secret is easy to find, i.e.,
$\vb{s} = {(1,0,0,0,0)}^T$.

After obtaining the initial $\vb{H}_{\vb{s}}^{\rm QRC}$, the verifier needs to
hide the secret and make the IQP circuit look random, while leaving the value of
the correlation function unchanged.
In the Shepherd-Bremner construction, the verifier will first add redundant rows
$\vb{R}_{\vb{s}}$, which are rows that are orthogonal to $\vb{s}$, to obtain the
full IQP matrix
\begin{align}
  \vb{H} = \begin{pmatrix}
             \vb{H}_{\vb{s}}^{\QRC} \\
             \vb{R}_{\vb{s}}
           \end{pmatrix} \; .
\end{align}
Its corresponding Hamiltonian $R_{\vb{s}}$ commutes with $\calZ_{\vb{s}}$ and
hence will not affect the correlation function.
After initializing $\vb{H}$ and $\vb{s}$, the verifier needs to apply
obfuscation to hide the secret.
The obfuscation is achieved by randomly permuting rows in $\vb{H}$ and
performing column operations to $\vb{H}$ and changing $\vb{s}$ accordingly.
\begin{definition}[Obfuscation]\label{def:obfuscation}
  Given an instance $(\vb{H}, \vb{s})$, the obfuscation is defined as the
  transformation
  \begin{align}
    \label{eq:obfuscation}
    \vb{H} &\gets \vb{PHQ} & \vb{s} \gets \vb{Q}^{-1} \vb{s} \; ,
 \end{align}
 where $\vb{P}$ is a random row-permutation matrix and $\vb{Q}$ is a random
 invertible matrix.
\end{definition}
Note that row permutations will not change the value of the correlation
function, since the gates in IQP circuits commute with each other.
As for the column operations, it can be shown that if the secret $\vb{s}$ is
transformed accordingly, to maintain the inner-product relation with the rows in
$\vb{H}$, then the value of the correlation function remains
unchanged~\cite{shepherd_temporally_2009,yung_anti-forging_2020}.

In the Shepherd-Bremner scheme~\cite{shepherd_temporally_2009}, the measure of
success is given by the probability bias
$\calP_{\s \perp} := \sum_{\x \cdot \s = 0 } p(\x)$, the probability of
receiving bit strings that are orthogonal to the secret $\s$, where $p(\x)$ is
the output probability of the IQP circuit.
This measure is equivalent to the correlation function, since
$\calP_{\s \perp} = \frac{1}{2} (\ev{\calZ_{\s}} + 1)$~\cite{shepherd_binary_2010,
  chen_experimental_2021}.
Due to the properties of QRC, $\ev{\calZ_{\s}}$ always equals $1/\sqrt{2}$ (in
terms of probability bias, 0.854).

\section{Stabilizer characterization of IQP circuits}\label{sec:stabilizer_characterization_IQP}

In this section, we establish the connection between IQP circuits, stabilizer formalism and
coding theory, which turns out to be useful in constructing the IQP circuits for
the verification protocol.
For $\theta = \pi/8$, we show that the stabilizer tableau of the Clifford
operation $e^{i2\theta \Ham(\vb{H}_{\vb{s}})}$ has a nice structure that allows us to
determine the value of
$\ev{\calZ_{\vb{s}}} = \mel{0^n}{U_{\vb{H}_{\vb{s}}, 2\theta}}{0^n}$ efficiently.
As an application, we analyze the Shepherd-Bremner construction with this
framework.

We first give the form of the stabilizer tableau of $U_{\vb{H}, \pi/4} \ket{0^n}$.

\begin{theorem}\label{thm:IQP_tableau}
  Given a binary matrix $\vb{H} = (\vb{c}_1, \ldots, \vb{c}_n)$, the stabilizer tableau of the state
  $U_{\vb{H}, \pi/4} \ket{0^n}$ can be expressed as,
  \begin{align}
    \label{eq:IQP_tableau}
    \left(
    \begin{array}{c | c | c}
      \begin{matrix}
        \vb{c}_1 \cdot \vb{c}_1 & \cdots & \vb{c}_1 \cdot \vb{c}_n \\
        \vdots & \ddots & \vdots \\
        \vb{c}_n \cdot \vb{c}_1 & \cdots & \vb{c}_n \cdot \vb{c}_n
      \end{matrix}
      &
      \begin{matrix}
        1 & \cdots & 0 \\
        \vdots & \ddots & \vdots \\
        0 & \cdots & 1
      \end{matrix}
      &
      \begin{matrix}
        r_1 \\
        \vdots \\
        r_n
      \end{matrix}
    \end{array}
    \right) \; .
 \end{align}
 Here, if one uses $00, 01, 10, 11$ to represent
 $|\vb{c}_j| = 0, 1, 2, 3 \pmod{4}$, then $r_j$ is equal to the first bit.
\end{theorem}

This theorem can be proved by starting from the standard tableau of $\ket{0^n}$,
and keeping track of the stabilizer tableau after applying each terms of
$U_{\vb{H}, \pi/4}$ (i.e., each row of $\vb{H}$).
The complete proof is delayed to Appendix~\ref{app:IQP_stabilizer}.
We will call Eq.~\eqref{eq:IQP_tableau} the IQP (stabilizer) tableau and it is
of the form $(\vb{G}, \vb{I}_n, \vb{r})$.
We apply the above theorem to $\vb{H}_{\vb{s}}$, in which case the $X$ part is
$\vb{G}_{\vb{s}} = \vb{H}_{\vb{s}}^T \vb{H}_{\vb{s}}$.

Next, we relate the correlation function to the code generated by
$\vb{H}_{\vb{s}}$, denoted as $\calC_{\vb{s}}$.
Note that $\vb{H}_{\vb{s}} \vb{s} = \vb{1}$ means that the all-ones vector is a
codeword of $\calC_{\vb{s}}$.
From Proposition~\ref{prop:all_one_vector}, this means that the dual code
$\calC_{\vb{s}}^{\perp}$ is an even code and the intersection
$\calD_{\vb{s}} := \calC_{\vb{s}} \bigcap \calC_{\vb{s}}^{\perp}$ is a weakly
self-dual even code.
Then, $\calD_{\vb{s}}$ will be either a doubly-even code or an unbiased even
code, according to Lemma~\ref{lemma:self_dual_even_code}.

\begin{theorem}\label{thm:cor_func_rank} % [Rank and overlap]
  Given an IQP matrix $\vb{H}_{\vb{s}}$ and a vector $\vb{s}$, so that
  $\vb{H}_{\vb{s}}\ \vb{s} = \vb{1}$.
  Denote the code generated by columns of $\vb{H}_{\vb{s}}$ by $\calC_{\vb{s}}$
  and its dual code by $\calC_{\vb{s}}^{\perp}$.
  Let $\calD_{\vb{s}} := \calC_{\vb{s}} \bigcap \calC_{\vb{s}}^{\perp}$.
  Then, the magnitude of the correlation function
  $\ev{\calZ_{\vb{s}}} = \mel{0^n}{U_{\vb{H}_{\vb{s}}, \pi/4}}{0^n}$ is $2^{-g/2}$ if
  $\calD_{\vb{s}}$ is a doubly-even code and zero if $\calD_{\vb{s}}$ is an
  unbiased even code.
  Here, $g := \dim(\calC_{\vb{s}}) - \dim(\calD_{\vb{s}})$ is also the rank of
  the Gram matrix $\vb{G}_{\vb{s}} = \vb{H}^T_{\vb{s}} \vb{H}_{\vb{s}}$.
\end{theorem}

We leave the proof in Appendix~\ref{app:IQP_stabilizer}.
Interestingly, from a group-theoretic perspective, the rank of the Gram matrix
$g$ is also the minimum number of different generators over all possible choices
of the stabilizer groups between $\ket{0^n}$ and
$U_{\vb{H}_{\vb{s}}, \pi/4} \ket{0^n}$ (Proposition~\ref{prop:overlap}).
Furthermore, we note that this result integrates several results in Ref.~\cite{shepherd_binary_2010} concisely, with a particular focus on coding theory, so that it aligns better with our objective of constructing IQP circuits for the verification protocol.
Ref.~\cite{shepherd_binary_2010} studies the IQP circuits with $\theta = \pi/4$ with a reworking of
Vertigan's algorithm for evaluating the magnitude of the Tutte polynomial of a
binary matroid at the point $(-i, i)$~\cite{vertigan_bicycle_1998}.
There, the amplitude $\mel{\vb{x}}{U_{\vb{H}, \pi/4}}{0^n}$ is considered any IQP matrix $\vb{H}$, where the all-ones vector may not
be a codeword of the code generated by the binary matrix $\vb{H}$.
Such an amplitude has been further studied in Ref.~\cite{mann_simulating_2021},
which gives the expression of the phase of the amplitude by applying results of
Ref.~\cite{pendavingh_evaluation_2014}.
In the language of binary matroids, the dual intersection $\calD_{\vb{s}}$ is the bicycle space of the matroid represented by $\vb{H}_{\vb{s}}$ and its dimension $\dim(\calD_{\vb{s}})$ is also known as the bicycle dimension~\cite{vertigan_bicycle_1998,mann_simulating_2021}.
Finally, we note that although computing the magnitude suffices for our later
construction, the sign of the correlation function can also be computed
efficiently, as shown in Ref.~\cite{mann_simulating_2021}.
In addition, when $g = O(\log{n})$, the correlation function has an inverse
polynomial scaling.
In this case, one can use the random sampling algorithm in
Ref.~\cite{yung_anti-forging_2020} to determine the sign efficiently.

To show the usefulness of the stabilizer characterization, we apply these two
theorems to analyze the Shepherd-Bremner construction.
Combined with the properties of QRC, we have the following corollary (with proof
presented in Appendix~\ref{app:IQP_stabilizer}).

\begin{corollary}\label{cry:QRC_generator}
  Let $q$ be a prime such that 8 divides $q+1$.
  Let $\vb{H}_{\vb{s}}^{\QRC}$ be a matrix whose first column is $\vb{1}$ (of
  length $q$), and whose remaining columns are the basis of the
  quadratic-residue code of length $q$, formed by the cyclic generator (i.e., in
  the form of Eq.~\eqref{eq:QRC_main_part}).
  Then, letting $H_{\vb{s}} = \Ham(\vb{H}_{\vb{s}}^{\QRC})$, the stabilizer tableau of
  $e^{i\pi H_{\vb{s}}/4} \ket{0^n}$ can be expressed as
  the following form,
  \begin{align}
    \left(
    \begin{array}{c|c|c}
      \begin{matrix}
        1 & \cdots & 1 \\
        \vdots & \ddots & \vdots \\
        1 & \cdots & 1
      \end{matrix}
      &
      \begin{matrix}
        1 & \cdots & 0 \\
        \vdots & \ddots & \vdots \\
        1 & \cdots & 1 \\
      \end{matrix}
      &
      \begin{matrix}
        1 \\
        \vdots \\
        1
      \end{matrix}
    \end{array}
    \right) \; .
 \end{align}
 As a result, the corresponding stabilizer group is generated by $-Y_1 X_2 \cdots X_n$, $-X_1 Y_2 X_3 \cdots X_n$, $\ldots$, and $-X_1 X_2 \cdots X_{n-1} Y_n$, where $n = (q+3)/2$.
 Moreover, the correlation function
 $\ev{\calZ_{\vb{s}}} = \mel{0^n}{e^{i\pi H_{\vb{s}}/4}}{0^n}$ has a magnitude $1/\sqrt{2}$.
\end{corollary}

\section{Stabilizer construction}\label{sec:stab_construction}

In this section, we present the stabilizer construction, which is a systematic
way to construct IQP circuits with $\theta = \pi/8$ for verification.
In fact, the goal is to generate a pair $(\vb{H}, \vb{s})$, such that they
satisfy certain conditions, which stem from \cref{thm:cor_func_rank}.
We first define the family of pairs that we would like to sample
from.
\begin{definition}\label{def:IQP_family}
  Let $\calH_{n, m, g} = \{(\vb{H}, \vb{s})\}$ be a family of pairs of an IQP
  matrix $\vb{H} \in \FF_2^{m\times n}$ and a secret $\vb{s} \in \FF_2^n$
  satisfying the following conditions.
  (1) $\calD_{\vb{s}} = \calC_{\vb{s}} \bigcap \calC_{\vb{s}}^{\perp}$ is a
  doubly-even code, where $\calC_{\vb{s}}$ is the code generated by columns of
  $\vb{H}_{\vb{s}}$ and $\calC_{\vb{s}}^{\perp}$ is its dual code; (2)
  $\rank(\vb{H}_{\vb{s}}^T \vb{H}_{\vb{s}}) = g$; (3) $\rank(\vb{H}) = n$.
\end{definition}
In this definition, the size of the IQP circuits are determined by $n$ and $m$,
which correspond to the number of qubits and gates, respectively.
Additionally, condition (1) is to guarantee that the correlation function
$\ev{\calZ_{\vb{s}}}$ corresponding to instances of $\calH_{n, m, g}$ is
nonzero, and condition (2) states that its magnitude is given by $2^{-g/2}$.
Therefore, the family $\calH_{n, m, g}$ includes all instances of IQP circuits
of a certain size that have correlation function $\pm 2^{-g/2}$ with respect to
some secret $\vb{s}$.
Note that the rank of the Gram matrix $\vb{H}_{\vb{s}}^T \vb{H}_{\vb{s}}$ should
be $g = O(\log{n})$ for the protocol to be practical.
The reason for considering IQP matrices $\vb{H}$ with full column rank will be
made clear when we discuss the classical security of the IQP-based verification
protocol (\cref{subsubsec:classical_sampling}).

\begin{metaalgorithm}[H]
  \begin{algspec}
    \textbf{Parameters: n, m, g} \\
    \textbf{Output: $(\vb{H}, \vb{s}) \in \calH_{n, m, g}$}
    \begin{algorithmic}[1]
      \State{Randomly sample $m_1$ and $d$ with certain constraints}
      \Comment{\cref{app:details_in_stabilizer_construction}} \State{Sample
        $\vb{D} \in \FF_2^{m_1 \times d}$ and $\vb{F} \in \FF_2^{m_1 \times g}$
        satisfying certain conditions}
      \Comment{\cref{app:details_in_stabilizer_construction}} \State{Initialize
        $\vb{H}_{\vb{s}} \gets (\vb{F}, \vb{D}, \vb{0}_{m_1\times (n-r)})$,
        where $r = g + d$} \State{Sample a secret $\vb{s}$ from the solutions of
        $\vb{H}_{\vb{s}}\ \vb{s} = \vb{1}$}
      \State{$\vb{H} \gets \begin{pmatrix} \vb{H}_{\vb{s}} \\ \vb{R}_{\vb{s}} \end{pmatrix}$,
        where $\vb{R}_{\vb{s}}$ is a random matrix with $m - m_1$ rows
        satisfying $\vb{R}_{\vb{s}} \vb{s} = \vb{0}$ and $\rank(\vb{H}) = n$}
      \State{Perform obfuscation as in \cref{def:obfuscation}}
    \end{algorithmic}
    \end{algspec}
    \caption{Stabilizer construction}\label{alg:stabilizer_construction}
\end{metaalgorithm}

Moreover, we give an efficient classical sampling algorithm to sample instances
from $\calH_{n, m, g}$, which is the stabilizer construction
(\cref{alg:stabilizer_construction}).
\begin{theorem}\label{thm:stabilizer_construction}
  There exists an efficient classical sampling algorithm that sample from
  $\calH_{n, m, g}$, given the parameters $n, m$ and $g$.
\end{theorem}
For the algorithmic purpose, we set two additional parameters, $m_1$ and $d$,
which are the number of rows in $\vb{H}_{\vb{s}}$ and the dimension of
$\calD_{\vb{s}}$, respectively.
These are random integers satisfying certain natural constraints (see
Appendix~\ref{app:details_in_stabilizer_construction}).
The rank of $\vb{H}_{\vb{s}}$ is then equal to $r = g + d$.
The stabilizer construction works by sampling $\vb{H}_{\vb{s}}$ and
$\vb{R}_{\vb{s}}$ in certain `standard forms', up to row permutations and column operations ({see \cref{fig:iqp_unobfuscated}}).
Note that the `standard forms' of $\vb{H}_{\vb{s}}$ and $\vb{R}_{\vb{s}}$ are
not necessarily unique.

We first discuss $\vb{R}_{\vb{s}}$.
To ensure that $\rank(\vb{H}) = n$, observe that in any $\vb{H}$ of full column
rank, the redundant rows $\vb{R}_{\vb{s}}$ can always be transformed by row
permutations into a form, where the first $n - r$ rows form a basis of $\FF_2^n$
together with the rows in $\vb{H}_{\vb{s}}$.
Therefore, up to row permutations, the first $n - r$ rows of $\vb{R}_{\vb{s}}$
are sampled to be random independent rows that are orthogonal to $\vb{s}$ and
lie outside the row space of $\vb{H}_{\vb{s}}$.
The remaining rows in $\vb{R}_{\vb{s}}$ are random rows orthogonal to $\vb{s}$.

Next, we discuss sampling $(\vb{H}_{\vb{s}}, \vb{s})$, which is the core of the
stabilizer construction.
Essentially, we want to randomly generate a (possibly redundant) generator
matrix $\vb{H}_{\vb{s}}$ of a code $\calC_{\vb{s}}$, so that its dimension is
$r$, its intersection $\calD_{\vb{s}}$ with the dual code is a doubly-even code
with dimension $d = r - g$ and the all-ones vector is a codeword.
The last condition guarantees that a secret $\vb{s}$ can always be found.
Note that, we allow $\rank(\vb{H}_{\vb{s}}) < n$.
That is, we allow $\vb{H}_{\vb{s}}$ to be a ``redundant'' generator matrix of
$\calC_{\vb{s}}$, instead of a full-rank one.
This is called adding column redundancy to the full-rank generator matrix of
$\calC_{\vb{s}}$, because after the obfuscation process, there will be redundant
linear combinations in the columns of $\vb{H}_{\vb{s}}$.
We give a more formal discussion of column redundancy in
\cref{app:column_redundancy}.

For such a generator matrix $\vb{H}_{\vb{s}}$, there is an invertible matrix
$\vb{Q}$ to perform a basis change so that
\begin{align}
  \label{eq:canonical_H_s}
  \vb{H}_{\vb{s}} \vb{Q} = (\vb{F}, \vb{D}, \vb{0}_{m_1 \times (n-r)}) \; ,
\end{align}
where $\vb{D} \in \FF_2^{m_1 \times d}$ is a generator matrix of the doubly-even
code $\calD_{\vb{s}}$, and columns in $\vb{F} \in \FF_2^{m_1 \times g}$ span
$\calC_{\vb{s}} \slash \calD_{\vb{s}}$.
In addition, it can be shown that
$\rank(\vb{F}^T \vb{F}) = \rank(\vb{Q}^T \vb{H}_{\vb{s}}^T \vb{H}_{\vb{s}} \vb{Q}) = \rank(\vb{H}_{\vb{s}}^T \vb{H}_{\vb{s}}) = g$.
Moreover, although there might be no unique standard form of $\vb{H}_{\vb{s}}$,
the Gram matrix has a unique standard form.
First note that row permutations have no effect on the Gram matrix, since
$\vb{P}^T \vb{P} = \vb{I}$ for a permutation matrix $\vb{P}$.
So we focus on column operations.
As shown in Ref.~\cite{kim_two_2008}, there exists an invertible matrix
$\vb{Q}$, so that
$\vb{Q}^T \vb{H}_{\vb{s}}^T \vb{H}_{\vb{s}} \vb{Q} = \mathrm{diag} \left( \vb{I}_g, \vb{0} \right)$
or $\mathrm{diag}\left( \bigoplus\limits_{i = 1}^{g/2} \vb{J}, \vb{0} \right)$,
depending on whether at least one diagonal element of $\vb{H}_{\vb{s}}^T \vb{H}_{\vb{s}}$ is 1 or not,
where $\vb{J} := \begin{pmatrix} 0 & 1 \\ 1 & 0 \end{pmatrix}$.
However, for the construction purpose, we need to ensure that the all-ones
vector is a codeword of $\calC_{\vb{s}}$.
Therefore, in \cref{app:details_in_stabilizer_construction}, we give a slightly
different standard form of $\vb{H}_{\vb{s}}^T \vb{H}_{\vb{s}}$, which can be
achieved by $\vb{H}_{\vb{s}}$ in the form of $(\vb{F}, \vb{D}, \vb{0})$.

\begin{figure}[t]
  \centering
  \includegraphics[width = 0.4\textwidth]{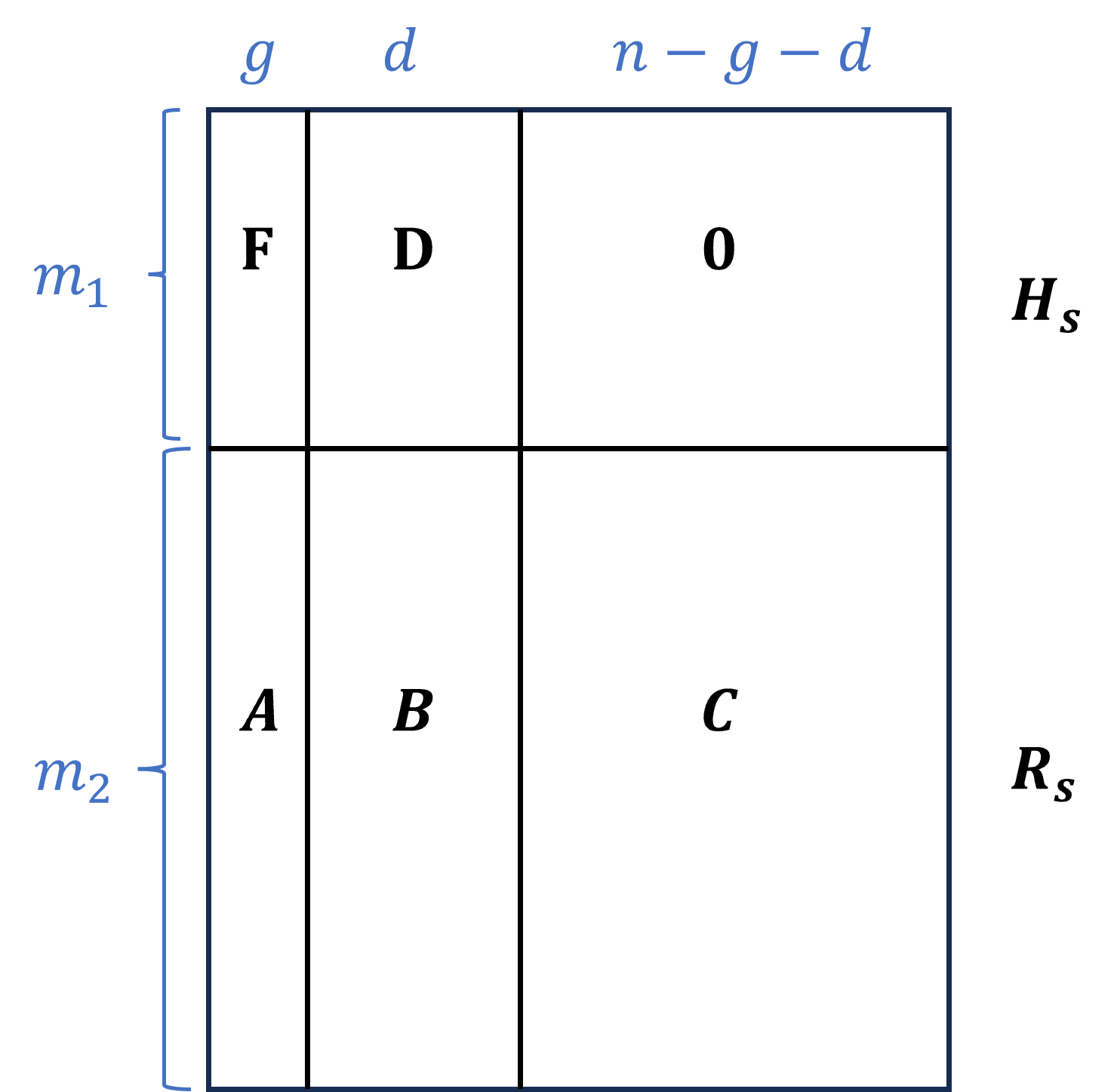}
  \caption{The IQP matrix $\vb{H}$ in the unobfuscated form. 
  Here, $m_2 =m - m_1$ is the number of rows in $\vb{R}_{\vb{s}}$.}
  \label{fig:iqp_unobfuscated}
\end{figure}

In summary, sampling $(\vb{H}_{\vb{s}}, \vb{s})$ is reduced to generating an
$\vb{H}_{\vb{s}} = (\vb{F}, \vb{D}, \vb{0})$ so that the Gram matrix
$\vb{H}_{\vb{s}}^T \vb{H}_{\vb{s}}$ is in the standard form presented in
\cref{app:details_in_stabilizer_construction}.
Then, a secret $\vb{s}$ is sampled from the solutions of
$\vb{H}_{\vb{s}} \ \vb{s} = \vb{1}$.
Sampling such an $\vb{H}_{\vb{s}}$ is further reduced to sampling $\vb{D}$ and
$\vb{F}$, so that $\vb{D}$ is a generator matrix for a random doubly-even code
and $\vb{F}$ is a random matrix satisfying $\vb{D}^T \vb{F} = \vb{0}$,
$\rank(\vb{F}^T \vb{F}) = g$ and that $\vb{1}$ is in the column space of
$(\vb{F}, \vb{D})$.
We claim that sampling such $\vb{D}$ and $\vb{F}$ can be done efficiently, with
details deferred to Appendix~\ref{app:details_in_stabilizer_construction}.
To conclude, before the obfuscation process as defined in \cref{def:obfuscation}, the IQP matrix $\vb{H}$ will be in the form of Fig.~\ref{fig:iqp_unobfuscated}.
Since $\vb{H}$ is full rank, we have that $\vb{C}$ is also full rank, i.e., $\rank(\vb{C}) = n-g-d$.

\section{Implementation}
\label{sec:implementation}

Before proceeding to classical security, we discuss the implementation of $U_{\vb{H}, \theta}$ with $\theta = \pi/8$ on a real quantum device.
First, after receiving the matrix $\vb{H}$, the prover can always transform it into 
\begin{align}
  \vb{PHQ} = \begin{pmatrix}
    \vb{I}_n \\ \overline{\vb{H}}
  \end{pmatrix} \ ,
\end{align}
with certain row permutation matrix $\vb{P}$ and invertible matrix $\vb{Q}$, where $\overline{\vb{H}}$ is a $(m-n) \times n$ matrix.
Then, he runs the IQP circuit $U_{\vb{PHQ}, \theta}$ instead, obtains samples $\vb{x}$ and returns the postprocessed samples $\vb{Q}^{-T} \vb{x}$ to the verifier.
Without loss of generality, one can also assume that the IQP matrix sent by the verifier is of the above form.

Next, we show how to compile $U_{\vb{H}, \theta}$ into an equivalent circuit consisting of CNOT circuits interleaving with a round of $e^{i\theta X}$ gates only.
Since $U_{\vb{PHQ}, \theta} = \prod_{j=1}^n e^{i\theta X_j} U_{\overline{\vb{H}}, \theta}$, we only need to apply the compilation to $U_{\overline{\vb{H}}, \theta}$.
For rows in $\overline{\vb{H}}$, we first find a maximal set of independent rows, say $\{ \vb{p}_1, \cdots, \vb{p}_k\}$.
Since these rows are independent, there exists a Clifford circuit $C_1$ such that $\calX_{\vb{p}_j} = C_1 X_j C_1^{\dagger}$ for $j \in [k]$, which implies that
\begin{align}
  \prod_{j = 1}^k e^{i\theta \calX_{\vb{p}_j}} = C_1 \prod_{j=1}^k e^{i\theta X_j} C_1^{\dagger} \ .
\end{align}
One can realize $C_1$ with CNOT gates only, since it only transforms Pauli-$X$ operators into Pauli-$X$ products by conjugation.
Iterating this process for the remaining rows in $\overline{\vb{H}}$ and supposing that the maximum iteration is $t$, then
\begin{align}
  U_{\vb{PHQ}, \theta} &= \prod_{j=1}^n e^{i\theta X_j} \prod_{\ell=1}^t \left( C_{\ell} \prod_{j=1}^{k_{\ell}} e^{i\theta X_j} C_{\ell}^{\dagger} \right) \\
  &= \prod_{j=1}^n e^{i\theta X_j} \prod_{\ell=1}^t \left( C'_{\ell} \prod_{j=1}^{k_{\ell}} e^{i\theta X_j} \right) \ ,
\end{align}
where we combine $C_{\ell-1}^{\dagger}$ with $C_{\ell}$.
For security reason, $n < m < 2n$ (discussed in the next section) and one can expect $t$ to be a small constant.
According to Ref.~\cite{aaronson_improved_2004}, $C'_{\ell}$ can be constructed with $O(n^2/\log{n})$ CNOT gates, which means that the gate complexity of $U_{\vb{PHQ}, \theta}$ is also $O(n^2/\log{n})$.
Using results from Ref.~\cite{jiang_optimal_2022}, the circuit depth can be optimized to $O(n/\log{n})$.

\section{Classical attacks and security}\label{sec:classical_attack_security}

In this section, we examine the classical security of our protocol, i.e., the
possibility that an efficient classical prover can pass the test.
A straightforward classical attack is to simulate the IQP circuit sent by the
verifier.
We do not expect this to be efficient, since there is generally no structure to
be exploited by a classical simulation algorithm.
For example, due to the obfuscation as in Eq.~\eqref{eq:obfuscation}, the
geometry of the IQP circuit can be arbitrary, which implies that the treewidth
in a tensor network algorithm cannot be easily
reduced~\cite{markov_simulating_2008}.

Here, we focus on another class of classical attacks based on extracting
secrets.
Given an IQP matrix $\vb{H}$, once the hidden secret $\vb{s}$ is found, a
classical prover can first calculate the correlation function
$\ev{\calZ_{\vb{s}}}$ efficiently.
Then, he generates a sample $\vb{x}$ which is orthogonal to $\vb{s}$ with
probability $(1 + \ev{\calZ_{\vb{s}}})/2$ and not orthogonal to $\vb{s}$ with
probability $(1 - \ev{\calZ_{\vb{s}}})/2$.
The generated samples will have the correct correlation with the secret $\vb{s}$
and hence pass the test.
Kahanamoku-Meyer's attack algorithm for the Shepherd-Bremner construction is an
instance of this class~\cite{kahanamoku-meyer_forging_2023}.

But generally, this attack may not be efficient.
From a code perspective, the stabilizer construction is to sample a random code
satisfying certain constraints, and hide it by adding redundancy and performing
obfuscation.
Finding the secret allows one to find the hidden subcode, which should be a hard
problem in general.
In particular, we formulate the following conjecture.
\begin{conjecture}[Hidden Structured Code (HSC) Problem, Restatement of \cref{conj:hsc-intro}]\label{conj:HSC}
  For certain appropriate choices of $n, m, g$, there exists an efficiently
  samplable distribution over instances $(\vb{H}, \vb{s})$ from the family
  $\calH_{n, m, g}$, so that no polynomial-time classical algorithm can find the
  secret $\vb{s}$ given $n, m$ and $\vb{H}$ as input, with high probability over
  the distribution on $\calH_{n, m, g}$.
\end{conjecture}
Naturally, sampling instances with uniform distribution from $\calH_{n, m, g}$
is more favorable, since it does not put any bias on specific instances.
For the underlying distribution induced by the stabilizer construction
(\cref{alg:stabilizer_construction}), it seems that it is uniform or close to
uniform, as the output instances are random instances satisfying certain natural
constraints imposed by the structure of the family $\calH_{n, m, g}$.
Though, we do not have a rigorous proof for this claim.
Moreover, a similar conjecture was given in Ref.~\cite{shepherd_temporally_2009}
for the family $\calH_{n, m, q}^{\QRC}$, where the problem is to decide whether
a given $\vb{H}$ is from the family $\calH_{n, m, q}^{\QRC}$ or not.
They conjectured that such a problem is $\mathsf{NP}$-complete.
Here, to better align with the classical attack, we consider the problem of
finding the secret $\vb{s}$ instead.

To support \cref{conj:HSC}, we first generalize Kahanamoku-Meyer's attack
algorithm to target any IQP-based verification protocols with $\theta = \pi/8$.
We show that this generalized attack, named the Linearity Attack, fails to break
our construction.
Furthermore, our analysis reveals that the loophole of the original
Shepherd-Bremner construction stems from an improper choice of parameters.
The Shepherd-Bremner construction can be improved by the column redundancy
technique, which enables random sampling from the family
$\calH_{n, m, q}^{\QRC}$ with any possible parameters and thereby fixes the
loophole.

\subsection{Linearity Attack}\label{subsec:linearity_attack}

Classical attacks based on secret extraction aim to mimic the quantum behavior
on certain candidate set $S$.
Observe that given an IQP circuit represented by the binary matrix $\vb{H}$, a
quantum prover can output a sample $\vb{x}$, which has the correlation function
$\ev{\calZ_{\vb{s}}}$ in the direction of $\vb{s}$ for every $\vb{s}$, even if
it is not the secret of the verifier.
If a classical prover can also generate samples that have the correct
correlation with every $\vb{s}$, then he has the power to classically sample
from an IQP circuit, which is implausible~\cite{bremner_classical_2011,
  bremner_average-case_2016}.
However, he has the knowledge that the verifier will only check one secret.
Therefore, a general attack strategy for him is to first reduce the set of
candidate secrets from ${\{0, 1\}}^n$ to a (polynomial-sized) subset $S$, and then
generate samples that have the correct correlation with every vector in the
candidate set.

Here, we discuss Linearity Attack, which is an instance of classical attacks
based on secret extraction and generalizes the attack algorithm in
Ref.~\cite{kahanamoku-meyer_forging_2023}.
It consists of two steps.
First, it uses linear algebraic techniques to construct a candidate set $S$.
Then, the prover calculates the correlation function for every vector in $S$,
and outputs samples that have the correct correlation with those vectors.

\subsubsection{Secret extraction}\label{subsubsec:secret_extraction}

\begin{metaalgorithm}[H]
  \centering
    \begin{algspec}
      \begin{algorithmic}[1]
        \Procedure{ExtractSecret}{$\vb{H}$} %
        \State{Initialize $S \gets \emptyset$.} \Comment{candidate set} %
        \Repeat%
        \State{Uniformly randomly pick $\vb{d} \in \FF_2^n$.}
        \State{Construct $\vb{H}_{\vb{d}}$ and
          $\vb{G}_{\vb{d}} = \vb{H}_{\vb{d}}^T \vb{H}_{\vb{d}}$} %
        \For{each vector $\vb{s}_i \in \ker(\vb{G}_{\vb{d}})$} %
        \If{$\vb{s}_i$ passes certain property check} \Comment{To be specified}
        \State{Add $\vb{s}_i$ to $S$.}
        \EndIf%
        \EndFor%
        \Until{some stopping criterion is met.}
        \State{\textbf{return} $S$} %
        \EndProcedure
      \end{algorithmic}
    \end{algspec}
    \caption{The \textsc{ExtractSecret}($\vb{H}$) procedure of Linearity
      Attack.}\label{alg:extract_secret}
\end{metaalgorithm}

\paragraph{Overview.}
The secret extraction procedure in the Linearity Attack is presented in
Meta-Algorithm~\ref{alg:extract_secret}, which is a generalized version of the
procedure described in Ref.~\cite{kahanamoku-meyer_forging_2023}.
The algorithm begins by randomly selecting a vector $\vb{d}$ and eliminating
rows in $\vb{H}$ that are orthogonal to $\vb{d}$, resulting in
$\vb{H}_{\vb{d}}$.
Subsequently, the algorithm searches for vectors that satisfy certain property
check in $\ker(\vb{G}_{\vb{d}})$, where
$\vb{G}_{\vb{d}} = \vb{H}_{\vb{d}}^T \vb{H}_{\vb{d}}$ represents the Gram matrix
associated with $\vb{d}$.
In what follows, we discuss some technical details and defer the analysis to
Section~\ref{subsec:analysis}.

\paragraph{Secret extraction in Kahanamoku-Meyer's attack.}
Meta-Algorithm~\ref{alg:extract_secret} differs slightly from the approach
described in Ref.~\cite{kahanamoku-meyer_forging_2023}.
In the original algorithm, the classical prover begins by constructing a matrix
$\vb{M} \in \FF_2^{l\times n}$ through linear combinations of rows in $\vb{H}$.
Specifically, after sampling the vector $\vb{d}$, the classical prover proceeds
to sample $l$ random vectors $\vb{e}_1, \ldots, \vb{e}_l$.
Then, the $j$-th row of $\vb{M}$ is defined by,
\begin{align}
  \vb{m}^T_j := \sum_{\substack{\vb{p}^T \in \row(\vb{H}) \\
  \vb{p} \cdot \vb{d} = \vb{p} \cdot \vb{e}_j = 1}} \vb{p}^T \; .
\end{align}
After that, the original algorithm searches for the vectors that can pass
certain property check in $\ker(\vb{M})$ instead.

Our secret extraction algorithm is a generalization and simplification to the
original approach.
In Appendix~\ref{app:secret_extraction_KM}, we show that rows in $\vb{M}$ belong
to the row space of $\vb{G}_{\vb{d}}$.
Therefore, to minimize the size of $\ker(\vb{M})$, one can simply set
$\vb{M} = \vb{G}_{\vb{d}}$, eliminating the need to sample the vectors
$\vb{e}_1, \ldots, \vb{e}_l$.

\paragraph{Property check.}
Next, we discuss the property checks designed to determine whether a vector in
$\ker(\vb{G}_{\vb{d}})$ can serve as a potential secret or not.
In the context of the Shepherd-Bremner construction targeted in
Ref.~\cite{kahanamoku-meyer_forging_2023}, the property check is to check
whether $\vb{s}_i$ in $\ker(\vb{M})$ corresponds to a quadratic-residue code or
not.
To accomplish this, the prover constructs $\vb{H}_{\vb{s}_i}$ for the vector
$\vb{s}_i$ and performs what we refer to as the QRC check, examining whether
$\vb{H}_{\vb{s}_i}$ generates a quadratic-residue code (with possible row
reordering).
However, determining whether a generator matrix generates a quadratic-residue
code is a nontrivial task.
Consequently, the algorithm in Ref.~\cite{kahanamoku-meyer_forging_2023}
attempts to achieve this by assessing the weight of the codewords in the code
generated by $\vb{H}_{\vb{s}_i}$.
In a quadratic-residue code, the weight of the codewords will be either 0 or 3
(mod 4).
But still, there will be exponentially many codewords, and checking the weights
of the basis vectors is not sufficient to ensure that all codewords have weight
either 0 or 3 (mod 4).
So in practice, the prover can only check a small number of the codewords.

For instances derived from the stabilizer construction, the prover will have
less information about the code $\calC_{\vb{s}}$; he only has the knowledge that
this code has a large doubly-even subcode, as quantified by the rank of
$\vb{G}_{\vb{s}}$.
Therefore, the property check for Meta-Algorithm~\ref{alg:extract_secret}
involves checking whether the rank of $\vb{H}_{\vb{s}_i}^T \vb{H}_{\vb{s}_i}$
falls below certain threshold and whether self-dual intersection
$\calD_{\vb{s}_i}$ is doubly-even.
However, determining an appropriate threshold presents a challenge for the
classical prover, who can generally only make guesses.
If the chosen threshold is smaller than the rank of $\vb{G}_{\vb{s}}$, then the
secret extraction algorithm will miss the real secret, even if it lies within
$\ker(\vb{G}_{\vb{d}})$.

\paragraph{Stopping criteria.}
Lastly, various stopping criteria can be employed in the secret extraction
procedure.
One approach is to halt the procedure once a vector successfully passes the
property check, as adopted in Ref.~\cite{kahanamoku-meyer_forging_2023}.
Alternatively, the procedure can be stopped after a specific number of
repetitions or checks.
In our implementation, we utilize a combination of these two criteria.
If no vectors are able to pass the property check before the stopping criterion
is reached, an empty candidate set $S$ is returned, indicating a failed attack.
Conversely, if the candidate set $S$ is non-empty, the attack proceeds to the
classical sampling step to generate classical samples.

\subsubsection{Classical sampling}\label{subsubsec:classical_sampling}

Classical sampling based on multiple candidate secrets is nontrivial.
Mathematically, the problem is formulated as follows.
\begin{problem}\label{prob:classical_sampling}
  Given an IQP circuit $C$ and a candidate set
  $S = \{ \vb{s}_1, \ldots, \vb{s}_t \}$, outputs a sample $\vb{x}$ so that
  \begin{align}
    \EE[{(-1)}^{\vb{x} \cdot \vb{s}_i}] = \ev{\calZ_{\vb{s}_i}} \; ,
 \end{align}
 for $i = 1, \ldots, t$, where $\EE[\cdot]$ is over the randomness of the
 algorithm.
\end{problem}
Note that $\EE[{(-1)}^{\vb{x} \cdot \vb{s}_i}]$ is the expectation value of
Eq.~\eqref{eq:Z_s_estimation}.
We may allow a polynomially-bounded additive error in the problem formulation,
considering the inevitable shot noise due to finite samples.
The complexity of this problem depends on various situations.
To the best of our knowledge, we are not aware of an efficient classical
algorithm that solves this problem in general.
In Appendix~\ref{app:more_attack_classical_sampling}, we present two sampling
algorithms that will work in some special cases.
A sufficient condition for these two sampling algorithms to work is that the
candidate set is an independent subset of ${\{0, 1\}}^n$.

\paragraph{Naive sampling algorithm.}
In this work, we mainly focus on the case $|S| = 1$, in which case the problem
is easy to solve, yet remains worth discussing.
A naive sampling algorithm is as follows.
To generate samples with the correct correlation on $\vb{s}$, one just needs to
output samples that are orthogonal to the candidate vector $\vb{s}'$ with
probability $\beta_{\vb{s}'} = (\ev{\calZ_{\vb{s}'}} + 1)/2$ and otherwise with
probability $1 - \beta_{\vb{s}'}$.
One can prove that if the candidate secret from the \textsc{ExtractSecret}
procedure is the real secret $\vb{s}$, then the generated samples using this
strategy will have the correlation function approximately $\ev{\calZ_{\vb{s}}}$
with the real secret.
Otherwise, the correlation function with the real secret will be zero.
We have the following lemma (see Appendix~\ref{app:equivalent_secrets} for the
proof).

\begin{lemma}\label{lemma:zero_correlation}
  Given a matrix $\vb{H}$ and two vectors $\vb{s} \neq \vb{s}'$, let
  $\ev{\calZ_{\vb{s}}}$ and $\ev{\calZ_{\vb{s}'}}$ be their corresponding
  correlation functions, as defined in Eq.~\eqref{eq:correlation_function}.
  If a sample $\vb{x}$ is generated to be a vector orthogonal to $\vb{s}'$ with
  probability $\beta_{\vb{s}'} = (\ev{\calZ_{\vb{s}'}} + 1)/2$ and otherwise
  with probability $1 - \beta_{\vb{s}'}$, then
  $\EE [{(-1)}^{\vb{x} \cdot \vb{s}}] = 0$.
\end{lemma}

The above lemma holds even if $\vb{H} \vb{s} = \vb{H} \vb{s}'$, in which case
$\vb{s}$ and $\vb{s}'$ are said to be \emph{equivalent secrets}.
Equivalent secrets have the same non-orthogonal and redundant part, and the
correlation functions $\ev{\calZ_{\vb{s}}}$ and $\ev{\calZ_{\vb{s}'}}$ are the
same.
It is clear that the number of equivalent secrets is given by
$2^{n - \rank(\vb{H})}$, which will be 1 if $\vb{H}$ is of full column rank.
When there are multiple equivalent secrets, it could be the case that the vector
$\vb{s}'$ is returned by the secret extraction procedure, because it can also
pass the property check, even if it is not the real secret itself.
In this case, our previous classical sampling algorithm can only give samples
with zero correlation function on the real secret $\vb{s}$, according to
Lemma~\ref{lemma:zero_correlation}.

\paragraph{Sampling according to $\vb{H}$.}
To address this issue, we propose a second classical sampling algorithm.
Observe that linear combination of rows in $\vb{R}_{\vb{s}}$ gives vectors that
are orthogonal to $\vb{s}$ and summation of an odd number of rows in
$\vb{H}_{\vb{s}}$ gives vectors that are not orthogonal to $\vb{s}$.
We denote the former set of vectors $\SS_0(\vb{s})$ and the latter
$\SS_1(\vb{s})$.
The identification of these sets relies on determining the submatrices
$\vb{H}_{\vb{s}}$ and $\vb{R}_{\vb{s}}$.
To achieve this, it suffices to find a vector $\vb{s}'$ that is equivalent to
the real secret $\vb{s}$.
Therefore, upon receiving the candidate secret $\vb{s}'$ from the secret
extraction procedure, the classical prover proceeds by computing
$\ev{\calZ_{\vb{s}'}}$ and $\beta_{\vb{s}'}$, followed by identifying
$\SS_0(\vb{s}')$ and $\SS_1(\vb{s}')$.
A sample $\vb{x}$ is drawn from $\SS_0(\vb{s}')$ with probability
$\beta_{\vb{s}'}$ and from $\SS_1(\vb{s}')$ with probability
$1 - \beta_{\vb{s}'}$.
If the vector $\vb{s}'$ is equivalent to $\vb{s}$, then this sampling algorithm
will generate samples with the correct correlation function with respect to the
real secret $\vb{s}$, as opposed to the naive sampling algorithm.

This also explains why we consider IQP matrices of full column rank in the
stabilizer construction.
If the classical prover is given an IQP matrix $\vb{H}$ that is not full-rank,
he can always apply an invertible matrix $\vb{Q}$ so that
$\vb{HQ} = (\vb{H}', \vb{0})$, where $\vb{H}'$ is of full column rank.
Then, he runs the secret extraction algorithm on $\vb{H}'$.
Once a candidate secret is found, he can use it to identify the corresponding
$\SS_0$ and $\SS_1$ from the original matrix $\vb{H}$, as well as computing the
correlation function.
Finally, if the identification matches that of the real secret, then using the
second classical sampling algorithm will allow him to pass the test.

\subsection{Analysis}\label{subsec:analysis}

Here, we present analysis on the secret extraction of Linearity Attack.

\paragraph{Probability of sampling a good $\vb{d}$.}
First, we have the following proposition.
\begin{proposition}\label{prop:correct_d}
  Given an IQP matrix $\vb{H}$ and two vectors $\vb{d}$ and $\vb{s}$, we have
  $\vb{G}_{\vb{s}} \vb{d} = \vb{G}_{\vb{d}}\, \vb{s}$, where
  $\vb{G}_{\vb{s}} = \vb{H}_{\vb{s}}^T \vb{H}_{\vb{s}}$ and
  $\vb{G}_{\vb{d}} = \vb{H}_{\vb{d}}^T \vb{H}_{\vb{d}}$.
  Therefore, $\vb{s}$ lies in $\ker(\vb{G}_{\vb{d}})$ if and only if
  $\vb{G}_{\vb{s}} \vb{d} = \vb{0}$, which happens with probability $2^{-g}$
  over all choices of $\vb{d}$, where $g = \rank(\vb{G}_{\vb{s}})$ is the rank
  of $\vb{G}_{\vb{s}}$.
\end{proposition}

The proof is given in Appendix~\ref{app:prob_good_d}.
This proposition tells us that if the random $\vb{d}$ does not satisfy
$\vb{G}_{\vb{s}} \vb{d} = \vb{0}$, then the verifier's secret $\vb{s}$ will not
lie in $\ker(\vb{G}_{\vb{d}})$.
In this case, Meta-Algorithm~\ref{alg:extract_secret} will not be able to find
the correct secret from the kernel of $\vb{\vb{G}_{\vb{d}}}$, and it has to be
started over with a new $\vb{d}$.

If the correlation function with respect to the real secret has inverse
polynomial scaling, i.e., $2^{-g/2} = \Omega(1/\poly(n))$, then the probability
of sampling a good $\vb{d}$ is also large, which is
$2^{-g} = \Omega(1/\poly(n))$.
This might appear advantageous for the attacker.
But note that a classical attack cannot determine whether the sampled $\vb{d}$
is good or not before he can find the real secret.
In fact, he even cannot \emph{definitively} determine whether a vector
$\vb{s}_i$ in $\ker(\vb{G}_{\vb{d}})$ that passes the property check is the real
secret or not.

\begin{figure*}[t]
  \centering %
  \includegraphics[width = 0.85\textwidth]{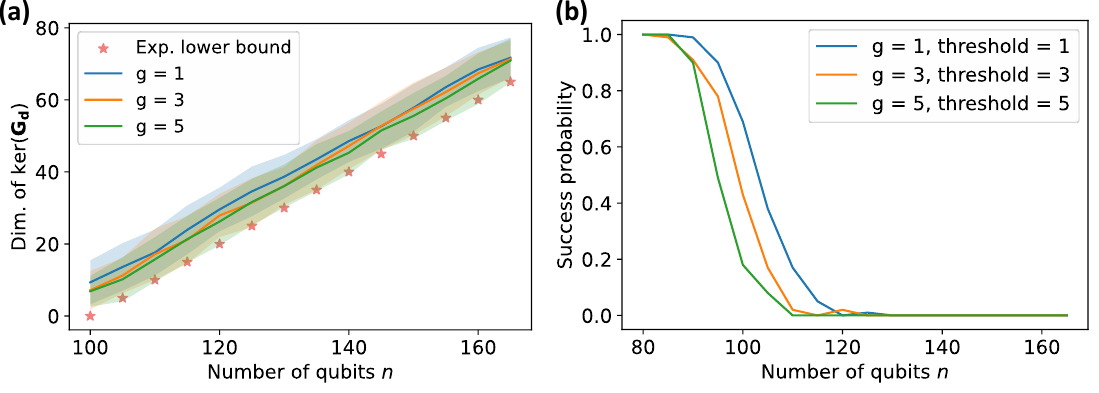}
  \caption{\textbf{(a)} The dimension of $\ker(\vb{G}_{\vb{d}})$ for
    $g = 1, 3, 5$ and the number of rows $m = 200$.
    The asterisks indicate the expected lower bound $n - m/2$.
    \textbf{(b)} The success probability of the attack.
    Here, we set the threshold for the rank in the property check to be the same
    as $g$.}\label{fig:stab_fig}
\end{figure*}

\paragraph{Size of $\ker(\vb{G}_{\vb{d}})$.}
The next question is, how large is the size of $\ker(\vb{G}_{\vb{d}})$.
This is important because the steps before the property check takes $O(n^3)$
time, which comes from the Gaussian elimination used to solve the linear system
to find the kernel of $\vb{G}_{\vb{d}}$.
However, for the property check, the prover will potentially need to check every
vectors in $\ker(\vb{G}_{\vb{d}})$, which takes time proportional to its size.
It is important to note that checking the basis vectors of
$\ker(\vb{G}_{\vb{d}})$ is not sufficient to find the real secret $\vb{s}$,
because the linearity structure is not preserved under taking the Gram matrix.
Even if $\vb{s} \in \ker(\vb{G}_{\vb{d}})$, the basis vectors of the kernel
space can all have high ranks for their associated Gram matrices.
Below, we give an expected lower bound for the size of $\ker(\vb{G}_{\vb{d}})$,
with the proof presented in Appendix~\ref{app:kernel_size}.
\begin{theorem}\label{thm:kernel_size}
  Given $(\vb{H}, \vb{s}) \in \calH_{n, m, g}$, randomly sample a vector
  $\vb{d}$.
  Then, the size of $\ker(\vb{G}_{\vb{d}})$ is greater than $2^{n - m/2}$ in
  expectation over the choice of $\vb{d}$.
\end{theorem}

Therefore, the size of $\ker(\vb{G}_{\vb{d}})$ is increased exponentially by
increasing $n$.
The increase of $n$ can be achieved by adding column redundancy, i.e., adding
more all-zeros columns in \cref{eq:canonical_H_s}.
But in the stabilizer construction, the column redundancy cannot be arbitrarily
large.
Recall that to make the IQP matrix $\vb{H}$ full rank, one needs to add at least
$n - r$ redundant rows, where $r = \rank(\vb{H}_{\vb{s}})$.
If $\vb{H}$ is not full rank, then as we discussed in
Section~\ref{subsubsec:classical_sampling}, the classical prover can always
perform column operations to effectively reduce the number of columns $n$, and
hence reduce the dimension of $\ker(\vb{G}_{\vb{d}})$.

\paragraph{Suggested parameter regime.}
Based on the above analysis, it is important to choose a good parameter regime
to invalidate the Linearity Attack.
Suppose the expected security parameter is $\lambda$, meaning that the expected
time complexity of a classical prover is $\Omega(2^{\lambda})$.
Then, generally we require $n - m/2 \geq \lambda$ for $\ker(\vb{G}_{\vb{d}})$ to
be sufficiently large, and the number of redundant rows $m - m_1 \geq n - r$ for
$\vb{H}$ to be full-rank, where $m_1$ is the number of rows in
$\vb{H}_{\vb{s}}$.
Specifically, for the stabilizer construction, given $n$ and $g$, we randomly
choose the parameter $r \geq g$.
Then, we require that the number of rows in $\vb{H}_{\vb{s}}$ and $\vb{H}$
satisfies
\begin{align}\label{eq:parameter_regime}
  m_1 &\leq n - 2 \lambda + r & m_1 + n - r \leq m &\leq 2(n - \lambda) \; ,
\end{align}
respectively.
In addition, since $m$ is the number of gates in the IQP circuit, we will
require sufficiently large $n$ and $m = \Omega(n)$ to invalidate classical
simulation.

\paragraph{Numerical simulation.}
In \cref{fig:stab_fig}~(a), we plot the dimension of $\ker(\vb{G}_{\vb{d}})$ for
$g = 1, 3, 5$ and $m = 200$.
For each number of columns $n$, we sample $100$ instances from $\calH_{n, m, g}$
with the stabilizer construction (\cref{alg:stabilizer_construction}).
Then, a random $\vb{d}$ is sampled and we calculate the dimension of
$\ker(\vb{G}_{\vb{d}})$.
The asterisks are the expected lower bound $n - m/2$, as shown in
\cref{thm:kernel_size}.
The numerical experiment demonstrates good agreement with the theoretical
prediction.
In \cref{fig:stab_fig}~(b), we present the numerical results for the success
probability of the attack.
Although to invalidate the attack, the maximum number of property checks should
be $2^{50}$ or larger, we set it to be $2^{15}$ for a proof of principle in the
numerical experiment.
For each number of columns $n$, we sample $100$ random instances from
$\calH_{n, m, g}$, where $m = 200$.
Then, the Linearity Attack is applied to each instance and the success
probability is defined as the fraction of successfully attacked instances, which
is the instance that the attacker can classically generate samples to spoof the
test.
As one can see, the success probability decreases to zero as $n$ exceeds
$m/2+15 = 115$, as expected.

\paragraph{Challenge.}
In addition, we have posted a challenge problem as well as the source code for
generation and verification on
GitHub~\cite{challenge},
to motivate further study.
The challenge problem is given by the $\vb{H}$ matrix of a random instance from
$\calH_{n, m, g}$ with $n = 300$ and $m = 360$; the $g$ parameter is hidden
because in practice, the prover can only guess a value.
One needs to generate samples with the correct correlation function in the
direction of the hidden secret to win the challenge.

\subsection{A fix of the Shepherd-Bremner construction}\label{subsec:fix_shepherd_bremner}

Finally, we would like to remark why the attack in
Ref.~\cite{kahanamoku-meyer_forging_2023} can break the Shepherd-Bremner
construction and how we can fix it by adding column redundancy.
Let $\calH^{\QRC}_{n, m, q} = \{(\vb{H}, \vb{s})\}$ be a family of pairs of an
IQP matrix $\vb{H} \in \FF_2^{m\times n}$ and a secret $\vb{s}$ so that
$\vb{H}_{\vb{s}}$ generates a QRC of length $q$ (up to row permutations) and
$\vb{H}$ is of full column rank.
What the construction recipe of Ref.~\cite{shepherd_temporally_2009} does is to
randomly sample instances from $\calH^{\QRC}_{n, m, q}$, where $n = (q+3)/2$ and
$m \geq q$, leaving a loophole for the recent classical
attack~\cite{kahanamoku-meyer_forging_2023}.
To see why the parameter regime is as above, we first note that the length of
QRC is $q$, implying that the number of rows in $\vb{H}_{\vb{s}}$ is $q$ and
hence $m \geq q$.
Moreover, the dimension of a length-$q$ QRC is $(q+1)/2$, which implies that the
rank of $\vb{H}_{\vb{s}}$ is $(q+1)/2$.
But an all-ones column was added in the construction (see
Eq.~\eqref{eq:QRC_main_part}), which is a codeword of QRC, leading to
$n = (q+3)/2$.

In the Shepherd-Bremner construction, the rank of Gram matrix $\vb{G}_{\vb{s}}$
associated with the real secret $\vb{s}$ is 1 according to
Corollary~\ref{cry:QRC_generator}.
Therefore, the probability of choosing a good $\vb{d}$ is $1/2$ (as also shown
in Theorem~3.1 of Ref.~\cite{kahanamoku-meyer_forging_2023}).
However, since the number of columns and the number of rows in $\vb{H}$ is
$n = (q+3)/2$ and $m \geq q$, respectively, the size of $\ker(\vb{G}_{\vb{d}})$
is generally small.
As a result, the prover can efficiently explore the entire
$\ker(\vb{G}_{\vb{d}})$, and if no vector passes the property check, the prover
can simply regenerate $\vb{d}$ and repeat the secret extraction procedure.
The numerical results in Ref.~\cite{kahanamoku-meyer_forging_2023} indicated
that the size of $\ker(\vb{G}_{\vb{d}})$ is indeed constant when applied to the
Shepherd-Bremner construction, which suggests that an efficient classical prover
can pass the test and hence break the original construction.
Specifically, for the challenge instance posted in
Ref.~\cite{shepherd_temporally_2009}, $m$ is taken to be $2q$.
Then, according to \cref{thm:kernel_size}, the dimension of
$\ker(\vb{G}_{\vb{d}})$ is expected to be constant, making it susceptible to the
attack.

To address this issue, the original Shepherd-Bremner construction can be
enhanced by introducing additional column redundancy to extend the number of
columns $n$, which can achieve random sampling from
families $\calH^{\QRC}_{n, m, q}$ with any $n \geq (q+1)/2$ (\cref{app:column_redundancy}).
This hides the dimension information of the hidden QRC.
Combined with other obfuscation techniques in the Shepherd-Bremner construction,
this achieves random sampling from $\calH^{\QRC}_{n, m, q}$ with any possible
parameters.

Below, we propose a parameter regime that can invalidate the attack in
Ref.~\cite{kahanamoku-meyer_forging_2023}.
Given the length $q$ of the QRC, we have $r = (q + 1)/2$ and
$m_1 = q$~\cite{macwilliams1977book}.
So, the first formula in Eq.~\eqref{eq:parameter_regime} gives
$n \geq (q - 1)/2 + 2\lambda$ and the second formula gives the range of the
number of redundant rows $n - (q+1)/2 \leq m_2 \leq 2n - 2\lambda - q$.
In this way, the size of $\ker(\vb{G}_{\vb{d}})$ will be larger than
$2^{\lambda}$ in general, offering a viable solution to fortify the
Shepherd-Bremner construction against the attack.
Note that the column redundancy technique was used in Ref.~\cite{yung_anti-forging_2020} to
scramble a small random IQP circuit into a large one, to maintain the value of
the correlation function, although its connection to the classical security was
not explored.
Moreover, a multi-secret version was explored in Ref.~\cite{snoyman_proof_2020},
which was shown to be more vulnerable to the classical attack instead.

\begin{figure*}[t]
  \centering %
  \includegraphics[width = 0.85\textwidth]{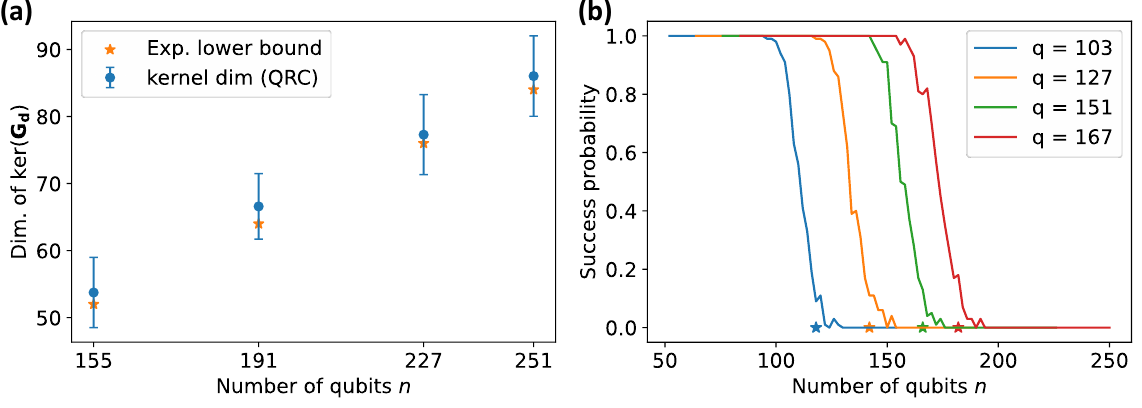}
  \caption{\textbf{(a)} The dimension of $\ker(\vb{G}_{\vb{d}})$ for
    $q = 103, 127, 151, 167$.
    Here, the number of rows and columns are $m = 2q$ and $n = r + q$, where
    $r = (q+1)/2$ is the dimension of QRC.
    \textbf{(b)} The success probability of the attack.
    The asterisks denote the points $(q+15, 0)$.}\label{fig:qrc_fig}
\end{figure*}

We perform numerical experiment to support our previous analysis.
When $m = 2q$, $n$ can be as large as $r+q$ and the expected kernel dimension of
$\vb{G}_{\vb{d}}$ is $r$.
In \cref{fig:qrc_fig}~(a), we plot the kernel dimensions under the setting
$n = r+q$ and $m = 2q$, with $q = 103, 127, 151$ and $167$.
For each parameter set, 100 instances are sampled from $\calH_{n, m, q}^{\QRC}$,
and then a random $\vb{d}$ is sampled for each instance and we evaluate the
dimension of $\ker(\vb{G}_{\vb{d}})$.
We also plot the expected lower bound $n - m/2$ for a comparison. 
In \cref{fig:qrc_fig}~(b), we plot the success probability versus the number of
columns (qubits) $n$.
Here, $m$ is set to be $2q$ and $n$ is increased from $r = (q+1)/2$ to $r+q$. 
For each value of $n$, 100 random instances from $\calH_{n, m, q}^{\QRC}$ are
sampled, and the success probability is the fraction of successful attacks among
them.
We set the security parameter to be $15$ for a proof of principle, meaning that
the maximum number of QRC checks is set to be $2^{15}$.
The success probabilities drop down to zero when $n > q+15$, as expected.
Our analysis and numerical results demonstrate that Claim~3.1 in Ref.~\cite{kahanamoku-meyer_forging_2023}, which originally states that the QRC-based construction can be broken efficiently by the KM attack, turns out to be false under appropriate choices of parameters.

\subsection{More classical attacks}
\label{subsec:more_classical_attacks}

After the first version of this work was posted on arXiv, Gross and Hangleiter proposed several new classical algorithms for secret extraction.
In the following, we will discuss the Radical Attack, the Hamming's Razor Attack and the Double Meyer Attack introduced in~\cite{gross_secret_2023}.
The first two algorithms successfully recovered the secret on our challenge instance, while the Double Meyer Attack is claimed to operate in quasi-polynomial time, though it did not succeed on our challenge instance.
We can choose appropriate parameter regimes to invalidate all these attacks, except for the Hamming's razor, which we have to modify our original construction to circumvent it (see \cref{app:more_classical_attacks}).

For simplicity, we consider the IQP matrix $\vb{H}$ in the unobfuscated form, as shown in Fig.~\ref{fig:iqp_unobfuscated}.
The discussion below can be generalized to the obfuscated form of $\vb{H}$.

\paragraph{Radical Attack.}
The main motivation behind the Radical Attack is that $\rank(\vb{H}^T \vb{H})$ for our challenge instance is large, significantly deviating from the behavior of a random matrix.
Taking the IQP matrix $\vb{H}$ as input, the procedure of the Radical Attack is as follows~\cite{gross_secret_2023}:
\begin{enumerate}
  \item Compute the kernel of $\vb{G} = \vb{H}^T \vb{H}$. Let $\vb{K}$ be the column-generating matrix of $\ker(\vb{G})$.
  \item Compute the joint support $S$ of $\vb{H} \vb{v}$ for $\vb{v}$ in the columns of $\vb{K}$.
  \item Solve $\vb{H} \vb{s} = \vb{1}_{S}$ for $\vb{s}$.
\end{enumerate}

Essentially, the Radical Attack tries to recover the row indices of $\vb{H}_{\vb{s}}$, which is the joint support $S$ in the above procedure.
Specifically, it tries to find vectors $\vb{v} \in \FF_2^n$, satisfying 
\begin{align}\label{eq:radical_attack_condition}
    \vb{H}_{\vb{s}}^T \vb{H}_{\vb{s}} \vb{v} &= \vb{0} \ , & \vb{R}_{\vb{s}} \vb{v} = \vb{0} \ .
\end{align}
In the unobfuscated picture, we have $\vb{H} \vb{v} = \begin{pmatrix} \vb{c}_1 \\ \vb{0} \end{pmatrix}$, where $\vb{c}_1 \in \range(\vb{D}) = \calD_{\vb{s}}$.
Finding sufficiently many such vectors $\vb{v}$ allows the attacker to recover the row indices of $\vb{H}_{\vb{s}}$, and hence the secret $\vb{s}$.
Moreover, such vectors indeed lie in $\ker(\vb{G})$, since $\vb{H}^T \vb{H} \vb{v} = \vb{H}_{\vb{s}}^T \vb{H}_{\vb{s}} \vb{v} + \vb{R}_{\vb{s}}^T \vb{R}_{\vb{s}} \vb{v} = \vb{0}$.
However, not all vectors in $\ker(\vb{G})$ satisfies \cref{eq:radical_attack_condition}.
So, in the actual implementation, a filtering process is applied to only select $\vb{v} \in \ker(\vb{G})$ that satisfies $\vb{H} \vb{v}$ being a doubly-even vector~\cite{gross_secret_2023}.
In this way, the filtered vectors $\vb{v}$ satisfy the conditions in \cref{eq:radical_attack_condition} with high probability.

For the challenge instance, $m_1$ is randomly chosen from values that satisfy the parameter constraints indicated in \cref{app:details_in_stabilizer_construction}.
It turns out that such random choice leads to a large $\ker(\vb{H}^T \vb{H})$, making the Radical Attack a success to recover the secret.

However, the Radical Attack highly depends on the structure of $\vb{H}$, which makes it easy to invalidate.
The desired vectors $\vb{v}$ are from $\vb{0}_g \oplus \FF_2^{n-g}$ in the unobfuscated picture, since it satisfies \cref{eq:radical_attack_condition}.
So, a straightforward way to invalidate the Radical Attack is to make the matrix $(\vb{B}, \vb{C})$ in the unobfuscated form (\cref{fig:iqp_unobfuscated}) full rank, in which case the two conditions in \cref{eq:radical_attack_condition} cannot be satisfied simultaneously.
Indeed, if $(\vb{B}, \vb{C})$ is of full rank, then vectors satisfying $\vb{R}_{\vb{s}} \vb{v} = \vb{0}$ must have nonzero entries in the first $g$ positions, which implies that $\vb{H}_{\vb{s}} \vb{v} \not\in \calD_{\vb{s}}$ and hence $\vb{H}_{\vb{s}}^T \vb{H}_{\vb{s}} \vb{v} \neq \vb{0}$.
Making $(\vb{B}, \vb{C})$ full rank imposes another parameter constraint: $m_2 \geq n - g$.
This is a stronger constraint than our previous one in \cref{app:details_in_stabilizer_construction}, i.e., $m_2 \geq n - g - d$, whose action is to make the IQP matrix $\vb{H}$ full rank.

\paragraph{Variants of Linearity Attack.}
There are two variants of Linearity Attack proposed in Ref.~\cite{gross_secret_2023}, namely, the Lazy Linearity Attack and the Double Meyer Attack.
Neither algorithm succeeds in recovering the secret in the challenge instance.
Here, we discuss the second one and leave the discussion on the first one to \cref{app:more_classical_attacks}.

Recall that in the Linearity Attack, the attacker samples a random vector $\vb{d}$ and the secret vector lies in $\ker(\vb{H}_{\vb{d}}^T \vb{H}_{\vb{d}})$ with probability $2^{-g}$ (\cref{prop:correct_d}), which is generally high since $g = O(\log{n})$ is required for the verification protocol to be sample efficient.
Therefore, to invalidate the Linearity Attack, we need to make the size of the search space $\ker(\vb{H}_{\vb{d}}^T \vb{H}_{\vb{d}})$ large enough.
According to \cref{thm:kernel_size}, the expected dimension of this kernel space is $\lambda_1 := n - m/2$.

Now, the idea of the Double Meyer Attack is that, instead of choosing only one random vector, the attacker can choose $k$ random vectors $\vb{d}_1, \cdots, \vb{d}_k$ independently and search in the intersection of the corresponding $k$ kernel spaces.
The probability that the real secret lies in the intersection space decays to $2^{-gk}$.
But Ref.~\cite{gross_secret_2023} argues that the dimension of the intersection space may shrink exponentially as $k$ increases.
Specifically, they expect the dimension to be $2^{-k} (n - m/2)$.
In that case, $k = O(\log{n})$ random vectors would suffice to make the intersection space of constant size.
Therefore, the running time of the Double Meyer Attack may be of order $O(2^{gk}) = 2^{O(\log^2{n})}$, which suggests that the Double Meyer Attack is a quasi-polynomial time algorithm.

However, we would like to point out that more in-depth analysis of the Double Meyer Attack is needed to support the claim of the quasi-polynomial time complexity.
Moreover, in practice, it is possible to choose appropriate parameters to fail the Double Meyer Attack (see \cref{app:more_classical_attacks}).

\paragraph{Hamming's razor.}
The Hamming's razor has a similar flavor to the Radical Attack, except that it tries to recover the row indices of $\vb{R}_{\vb{s}}$~\cite{gross_secret_2023}.
The procedure of the Hamming's razor is as follows:
\begin{enumerate}
  \item Initialize $S \gets \emptyset$.
  \item Randomly remove $p$ fraction of rows from $\vb{H}$. Denote the remaining submatrix as $\vb{H}'$.
  \item Solve for the kernel of $\vb{H}'$ and let $\vb{K}$ be the column-generating matrix of $\ker(\vb{H}')$.
  \item For each column $\vb{v}$ in $\vb{K}$, add the elements of $\supp(\vb{Hv})$ to $S$.
  \item Repeat steps 2-4 many times.
  \item Solve $\vb{Hs} = \vb{1}_{S^C}$ for $\vb{s}$, where $S^C := [m]\backslash S$.
\end{enumerate}

In its essence, the Hamming's razor tries to find vectors $\vb{v} \in \vb{0}_{g+d} \oplus \FF_2^{n-g-d}$ in the unobfuscated picture, so that $\vb{H} \vb{v} = \begin{pmatrix} \vb{0} \\ \vb{c}_2 \end{pmatrix}$ with $\vb{c}_2 \in \range(\vb{C})$.
Finding sufficiently many such vectors allows the attacker to recover the row indices of $\vb{R}_{\vb{s}}$, and hence the secret $\vb{s}$.

The question is, whether and when the procedure of the Hamming's razor finds the desired vectors $\vb{v} \in \vb{0}_{g+d} \oplus \FF_2^{n-g-d}$.
It turns out that, in the parameter regime that fails the Linearity Attack and its variants, Hamming's razor succeeds in recovering the secret almost surely.
We provide a detailed analysis in \cref{app:more_classical_attacks}.
Intuitively, in our original construction, the first $g+d$ columns of $\vb{H}$ remains of full rank even after removing a large fraction of rows, but the last $n-g-d$ columns could become linearly dependent due to the large block of all-zeros matrix.
Therefore, vectors in the kernel of $\vb{H}'$ could have nonzero entries only in the last $n-g-d$ positions by choosing a suitable fraction $p$.

Based on this analysis, a possible workaround to invalidate the Hamming's razor is to impose the sparsity constraint on the first $g+d$ columns of $\vb{H}$.
We leave the detailed construction to \cref{app:more_classical_attacks}.
Under the modified construction and imposing all previously identified parameter constraints, we generate instances that could fail the Hamming's razor as well as all other known attacks with the parameter set $n = 700, m = 1200$ and $g = 10$~\cite{challenge}.

\section{Discussion}\label{sec:discussion}

In this work, we give the stabilizer scheme for the IQP-based protocols
for verifiable quantum advantage, which focuses on the case $\theta = \pi/8$ in
the IQP circuits.
With the connection between IQP circuits, stabilizer formalism and coding
theory, we study the properties of correlation functions and IQP circuits.
Based on these properties, we give an efficient procedure to sample generator
matrices of random codes satisfying certain conditions, which lies at the core
of our stabilizer scheme.
Then, one needs to hide and obfuscate this generator matrix into a larger
matrix.
We propose a new obfuscation method called column redundancy, which uses the
redundant generator matrix to hide the information of the dimension of the
hidden code.

To explore the classical security of our protocol, we consider a family of
attacks based on extracting secrets.
We conjecture that such attacks cannot be efficient classically for random
instances generated by our stabilizer scheme.
To support this conjecture, we extend the recent attack algorithm on the
QRC-based construction to the general case for $\theta = \pi/8$, which we call
the Linearity Attack.
Our analysis shows that this attack fails to find the secret in polynomial time
by choosing instances from a good parameter regime.
Notably, our column redundancy technique also fixes the loophole in the original Shepherd-Bremner construction {due to Ref.~\cite{kahanamoku-meyer_forging_2023}}.
{In response to the recent new classical attacks by Gross and Hangleiter~\cite{gross_secret_2023}, we propose modification to our original construction and the possible parameter regimes that could resist all currently known classical attacks.}
Our work paves the way for cryptographic verification of quantum computation
advantage in the NISQ era.

There are several open problems for future research.
The most important one is to rigorously prove the security of the IQP-based
verification protocols.
In Conjecture~\ref{conj:HSC}, we state that classical attacks based
on secret extraction is on average hard.
It would be favorable to prove the random self-reducibility of the problem, so
that the hardness conjecture can be relaxed to the worst-case scenario.
For example, recently a worst-to-average-case reduction was found for computing the probabilities of IQP circuits and it would be interesting to see if the techniques of Ref.~\cite{movassagh_hardness_2023} could be leveraged to gain insight into the validity of \cref{conj:HSC}.
Before one can rigorously prove the hardness of classical attacks, one might
gain intuition by considering other possible classical attacks.
{In particular, more in-depth analysis of the recent new classical attacks~\cite{gross_secret_2023} would be helpful to better identify the strengths and weaknesses of our protocol, especially when it comes to the Double Meyer Attack and the Hamming's razor.
On the construction side, more structure is included in our workaround to invalidate the known attacks, which is heuristic and might open opportunities for new attacks.
It would be favorable to consider fundamentally different and more natural distributions on the family $\calH_{n, m, g}$ of instances $(\vb{H}, \vb{s})$; we leave it for future study.}
In terms of implementing the protocol in practice, generating instances
according to a given architecture and noise analysis are also important open problems.

Finally, we would like to point out that IQP circuits with $\theta = \pi/8$ lie in the third level of Clifford hierarchy~\cite{gottesman_demonstrating_1999}.
All quantum circuits in this level feature classical efficient simulation of Pauli observables by definition, which is the reason why $\ev{\calZ_{\vb{s}}}$ can be efficiently computed for IQP circuits.
{A large class of circuits in this level is in the form $C = L_1 D L_2$, where $L_1$ and $L_2$ are Clifford gates and $D$ consists of $T$, $CS$ and $CCZ$ gates~\cite{beigi_c3_2010,cui_diagonal_2017}.
The IQP circuits considered in this work are a special case with $L_1 = L_2 = H^{\otimes n}$.}
It is interesting to explore how to use other circuits in the third level of Clifford hierarchy for demonstrating verifiable quantum advantage.
{The code-based techniques that we establish in this work seem to easily generalize to the third level of Clifford hierarchy.}
We believe that the mathematical structure of the stabilizer scheme provides a
promising avenue for the use of certain cryptographic techniques to improve
the security of {verification protocols based on obfuscated IQP circuits and beyond}, and to construct instances that can be
readily implemented with current technology.

\begin{acknowledgments}
We are grateful to David Gross and Dominik Hangleiter for explaining and discussing their new classical attacks.
We thank Ryan Snoyman for sharing his honors thesis, where he also considered
the same problem and made some insightful observations.
We also thank Scott Aaronson, Earl Campbell, Ryan Mann, Mauro Morales and Man-Hong Yung for
helpful discussions.
BC acknowledges the support by the Sydney Quantum Academy and 
the support by the National Research Foundation, Singapore, and A*STAR under its CQT Bridging Grant and its Quantum Engineering Programme under grant NRF2021-QEP2-02-P05.
MJB acknowledges the support of Google.
MJB acknowledges support by the ARC Centre of Excellence for Quantum Computation
and Communication Technology (CQC2T), project number CE170100012.
ZJ acknowledges the support by National Natural Science Foundation of China
(Grant No. 12347104), National Key Research and Development Program of China
(Grant No. 2023YFA1009403), Beijing Natural Science Foundation (Grant No.
Z220002), and a startup funding from Tsinghua University.
\end{acknowledgments}

\bibliography{refs}

\onecolumngrid

\appendix

\section{Derivation of the overlap of two stabilizer states}
\label{app:overlap}

Here, we would like to prove Proposition~\ref{prop:overlap}, which was discussed in Ref.~\cite{aaronson_improved_2004}.
Recall that given two stabilizer states $\ket{\psi}$ and $\ket{\phi}$, we have:
\begin{enumerate}
    \item $\ip{\psi}{\phi} = 0$ if their stabilizer groups contain the same Pauli operator of the opposite sign. That is, if there exists a Pauli operator $P$, such that $P \in \Stab(\ket{\psi})$ and $-P \in \Stab(\ket{\phi})$, then $\ip{\psi}{\phi} = 0$.
    \item Suppose the condition of the first case does not hold. Let $g$ be the minimum number of different generators (i.e., the number of $i$ s.t. $P_i \neq Q_i$). Then, $\abs{\ip{\psi}{\phi}} = 2^{-g/2}$. 
\end{enumerate}

\begin{proof}[Proof of Proposition~\ref{prop:overlap}]
First, suppose $\Stab(\ket{\psi}) = \langle P_1, \cdots, P_n \rangle$ and $\Stab(\ket{\phi}) = \langle Q_1, \cdots, Q_n \rangle$.
Then, we can write the states as,
\begin{align}
    \dyad{\psi} &= \left( \frac{I + P_1}{2} \right) \cdots \left( \frac{I + P_n}{2} \right) \\
    \dyad{\phi} &= \left( \frac{I + Q_1}{2} \right) \cdots \left( \frac{I + Q_n}{2} \right) \ .
\end{align}
The square of the overlap is then given by, 
\begin{align}
    |\ip{\psi}{\phi}|^2 = \Tr(\dyad{\psi} \dyad{\phi}) = \Tr(\frac{I+P_1}{2} \cdots \frac{I+P_n}{2} \frac{I+Q_1}{2} \cdots \frac{I+Q_n}{2}) \ .
\end{align}
\begin{enumerate}
    \item Suppose $Q_1 = -P_n$. Then, we have,
    \begin{align}
        \frac{I+P_n}{2} \frac{I-P_n}{2} = 0 \ .
    \end{align}
    Thus, $\ip{\psi}{\phi} = 0$ in this case. 

    \item Suppose that $P_i = Q_i$ for $i > g$ and that the group $\langle P_1, \cdots, P_g \rangle$ is not equal to $\langle Q_1, \cdots, Q_g \rangle$. 
    By commutation, we can group the same generators, which gives,
    \begin{align}
        \frac{I+P_i}{2} \frac{I+Q_i}{2} = \frac{I+P_i}{2} \frac{I+P_i}{2} = \frac{I+P_i}{2} \ ,
    \end{align}
    for $i > g$.
    This will eliminate the terms related to $Q_{g+1}, \cdots, Q_n$.
    Then, 
    \begin{align}
        |\ip{\psi}{\phi}|^2 &= \Tr(\frac{I+P_1}{2} \cdots \frac{I+P_n}{2} \frac{I+Q_1}{2} \cdots \frac{I+Q_g}{2}) \\
        &= \frac{1}{2^g} \mel{\psi}{(I+Q_1) \cdots (I+Q_g)}{\psi} \ .
    \end{align}
    For every term $Q_i Q_j \cdots Q_k \neq I$ in the expansion, there exists a Pauli operator $P \in \Stab(\ket{\psi})$ that anticommutes with it; otherwise, the term will be in the stabilizer group of $\ket{\psi}$. For such an operator $Q$, we have $\mel{\psi}{Q}{\psi} = 0$. Indeed, notice that
    \begin{align}
        \mel{\psi}{Q}{\psi} = \mel{\psi}{Q P}{\psi} = -\mel{\psi}{P Q}{\psi} = -\mel{\psi}{Q}{\psi} \ ,
    \end{align}
    which implies $\mel{\psi}{Q}{\psi} = 0$. Finally, we have,
    \begin{align}
        |\ip{\psi}{\phi}|^2 = \frac{1}{2^g} \mel{\psi}{I}{\psi} = \frac{1}{2^g} \ ,
    \end{align}
    and $\abs{\ip{\psi}{\phi}} = 2^{-g/2}$.
\end{enumerate}
\end{proof}

\section{Proofs for coding theory}
\label{app:coding_theory}

We first prove the following proposition.

\begin{proposition}
    The all-ones vector is a codeword of $\calC$ if and only if its dual code $\calC^{\perp}$ is an even code.
\end{proposition}

\begin{proof}
    Suppose $\vb{1} \in \calC$. Then for every $\vb{c} \in \calC^{\perp}$, we have $\vb{c} \cdot \vb{1} = 0$, which means that $|\vb{c}|$ is even and hence $\calC^{\perp}$ is an even code.
    Conversely, suppose $\calC^{\perp}$ is an even code. Then, all codewords will be orthogonal to the all-ones vector, and thus it is in $\calC$.
\end{proof}

To prove that a weakly self-dual even code is either a doubly-even code or an unbiased even code, we will need the following lemma.
\begin{lemma}\label{lemma:sum_of_codewords}
    Given two vectors $\vb{c}_1, \vb{c}_2 \in \FF_2^m$ with even parity and $\vb{c}_1 \cdot \vb{c}_2 = 0$, let $\vb{c}_3 = \vb{c}_1 + \vb{c}_2$.
    Then, $|\vb{c}_3| = 0 \pmod{4}$ if $|\vb{c}_1| = |\vb{c}_2| \pmod{4}$ and $|\vb{c}_3| = 2 \pmod{4}$ if $|\vb{c}_1| \neq |\vb{c}_2| \pmod{4}$.
\end{lemma}

\begin{proof}
    Let $|\vb{c}_1| = a + 4 k_1$ and $|\vb{c}_2| = b + 4 k_2$, where $0 \leq a, b < 4$. 
    Let the size of joint support of $\vb{c}_1$ and $\vb{c}_2$ be $k_{12}$.
    Then, $\vb{c}_1 \cdot \vb{c}_2 = k_{12} = 0 \pmod{2}$, which means that $k_{12}$ is an even number. 
    So, 
    \begin{align}
        |\vb{c}_3| = a + b - 2 k_{12} + 4 (k_1 + k_2) = a + b \pmod{4} \ .
    \end{align}
    \begin{itemize}
        \item If $|\vb{c}_1| = |\vb{c}_2| \pmod{4}$, we have $a = b = 0$ or $2$. In either case, $|\vb{c}_3| = 0 \pmod{4}$.
        \item If $|\vb{c}_1| \neq |\vb{c}_2| \pmod{4}$, we have $a = 0$ and $b = 2$ or $a = 2$ and $b = 0$. In either case, $|\vb{c}_3| = 2 \pmod{4}$.
    \end{itemize}
\end{proof}

One can adapt the proof of this lemma to show that a doubly-even code is a weakly self-dual code.
To see this, suppose $\calC$ is a doubly-even code, and $\vb{c}_1, \vb{c}_2 \in \calC$.
Then, we have $|\vb{c}_1| = 4 k_1$ and $|\vb{c}_2| = 4 k_2$.
Suppose $\vb{c}_3 = \vb{c}_1 + \vb{c}_2$, which gives $|\vb{c}_3| = 4 (k_1 + k_2) - 2 k_{12}$.
Since $\vb{c}_3$ is also a codeword of a doubly-even code, we have $|\vb{c}_3| = 0 \pmod{4}$, which implies that $k_{12}$ is even and thus $\vb{c}_1 \cdot \vb{c}_2 = 0$.

Now, we are ready to prove Lemma~\ref{lemma:self_dual_even_code}.

\begin{proof}
    Let $\calD$ be a weakly self-dual even code spanned by $\{\vb{c}_1, \cdots, \vb{c}_d\}$.
    Then, $\vb{c}_i$'s are all even-parity and orthogonal to each other.
    Any codeword of $\calD$ can be written as $\vb{c} = a_1 \vb{c}_1 + \cdots + a_d \vb{c}_d$.
    According to Lemma~\ref{lemma:sum_of_codewords}, in the linear combination of $\vb{c}$, if there is an odd number of $\vb{c}_i$'s with weight 2 modulo 4, then $\vb{c}$ will have weight 2 mod 4, and otherwise, $\vb{c}$ will have weight 0 mod 4.
    Therefore, if all $\vb{c}_i$'s have weight 0 modulo 4, then $\calD$ is doubly-even.
    If there exist $\vb{c}_i$'s with weight 2 modulo 4, then $\calD$ is an unbiased even code.
\end{proof}

Given a generator matrix $\vb{H}$, the rank of its Gram matrix $\vb{G} = \vb{H}^T \vb{H}$ is an invariant under a basis change.
That is, $\rank(\vb{Q}^T \vb{G} \vb{Q}) = \rank(\vb{G})$ for $\vb{Q}$ invertible.
It may be tentative to consider this as a direct consequence of Sylvester's law of inertia, but this is not the case since we are working in $\FF_2$.
Nevertheless, this can be proven as follows.
First, the column space of $\vb{GQ}$ is a subspace of $\vb{G}$, which implies $\rank(\vb{GQ}) \leq \rank(\vb{G})$.
On the other hand, the column space of $\vb{G}$ is a subspace of $\vb{GQ}$, because $\vb{GQ} \vb{Q}^{-1} = \vb{G}$, which implies $\rank(\vb{G}) \leq \rank(\vb{GQ})$.
Therefore, we have $\rank(\vb{G}) = \rank(\vb{GQ})$.
Applying this reasoning again gives $\rank(\vb{G}) = \rank(\vb{GQ}) = \rank(\vb{Q}^T \vb{G} \vb{G})$.

The rank of the Gram matrix measures how close a code $\calC$ is to being a self-dual code.
In particular, $\rank(\vb{G}) = \dim(\calC) - \dim(\calD)$, where $\calD := \calC \bigcap \calC^{\perp}$, which is the Proposition~\ref{prop:Gram_matrix_and_code} in the main text.

\begin{proof}[Proof of Proposition~\ref{prop:Gram_matrix_and_code}]
    Suppose $\vb{H} \in \FF_2^{m\times n}$ and let $g = \rank(\vb{G})$, where $\vb{G} = \vb{H}^T \vb{H}$.
    Let $r = \dim(\calC)$ and $d = \dim(\calD)$, where $d \leq r \leq n$.
    We first prove for the case $r = n$.
    In this case, every codeword in $\vb{c} \in \calC$ can be expressed as $\vb{c} = \vb{Ha}$ for a unique $\vb{a} \in \FF_2^n$ and the correspondence is one-to-one.
    Then, we claim that $\vb{Ha} \in \calD$ is equivalent to $\vb{a} \in \ker(\vb{G})$ and thus $d$ is equal to the dimension of $\ker(\vb{G})$, which is $d = r - g$.
    Indeed, if $\vb{Ha} \in \calD$, we have $\vb{Ga} = \vb{H}^T \vb{H} \vb{a} = \vb{0}$, which means that $\vb{a} \in \ker(\vb{G})$.
    Conversely, if $\vb{a} \in \ker(\vb{G})$, we have $\vb{H}^T \vb{H} \vb{a} = \vb{0}$, which means that $\vb{Ha} \in \calC^{\perp}$.
    Since $\vb{Ha} \in \calC$, this implies $\vb{Ha} \in \calD$.

    Now, we consider the case $r < n$.
    In this case, there always exists an invertible matrix $\vb{Q}$ such that $\vb{HQ} = (\vb{H}', \vb{0}_{m\times (n-r)})$ and $\vb{H}' \in \FF_2^{m\times r}$ is a generator matrix of $\calC$ that is of full column rank.
    Moreover, 
    \begin{align}
        \rank(\vb{G}) = \rank(\vb{Q}^T \vb{G} \vb{Q}) = \rank(\vb{H}'^T \vb{H}') \ .
    \end{align}
    Then, applying the previous reasoning to $\vb{H}'$ yields that $d = r - \rank(\vb{H}'^T \vb{H}') = r - g$.
\end{proof}

\section{IQP circuits and stabilizer formalism}
\label{app:IQP_stabilizer}

\subsection{IQP stabilizer tableau}

Theorem~\ref{thm:IQP_tableau} can be proven in the following way.
First, we start with the standard tableau of $\ket{0^n}$, which is 
\begin{align}
    \begin{pmatrix}
        0 & \ldots & 0 & 1 & \ldots & 0 & 0 \\
        \vdots & \ddots & \vdots & \vdots & \ddots & \vdots & \vdots \\
        0 & \ldots & 0 & 0 & \ldots & 1 & 0 \\
    \end{pmatrix} \ ,
\end{align}
corresponding to the stabilizer generators $\{Z_1, \cdots, Z_n\}$.
Then, we apply the local terms in $e^{i\pi H /4}$ one by one, and keep track of the change of the stabilizer tableau.
We have the following lemma which gives the form of $Z_j$ conjugated by $e^{i\pi H/4}$.

\begin{lemma}[Evolution of $Z_j$]
    \label{lemma:evolution_of_stabilizer}
    Let $\mathbf{H} = (\vb{c}_1, \vb{c}_2, \cdots, \vb{c}_n)$ be a binary matrix and let $H = \Ham(\vb{H})$. Then, after the conjugation of $e^{i\pi H /4}$, we have,
    \begin{align}
        e^{i\pi H /4} Z_j e^{-i\pi H /4} = i^{|\vb{c}_j|} \prod_{k=1}^n X_k^{\vb{c}_j \cdot \vb{c}_k} Z_j \ ,
    \end{align}
    where $|\vb{c}_j|$ is the Hamming weight of $\vb{c}_j$.
\end{lemma}

For example, let 
\begin{align}
    \mathbf{H} = 
    \begin{pmatrix}
        1 & 1 & 0 & 0 \\
        0 & 1 & 0 & 1 \\
        1 & 0 & 0 & 1
    \end{pmatrix} \ .
\end{align}
Then, after the conjugation of $e^{i\pi H /4}$, we have
\begin{align}
    Z_1 \to (-1) (X_1 X_2) (X_1 X_4) Z_1 = - Z_1 X_2 X_4 \ .
\end{align}

\begin{proof}
    First, note that $e^{i\pi H /4} = \prod_{\vb{p}^T \in \row(\mathbf{H})} e^{i\pi \calX_{\vb{p}} /4}$, where $\calX_{\vb{p}} := X^{p_1} \otimes \cdots \otimes X^{p_n}$. 
    For each row $\vb{p}^T$, if $p_j = 1$, then
    \begin{align}
        e^{i\pi \calX_{\vb{p}} /4} Z_j e^{-i\pi \calX_{\vb{p}} /4} = e^{i\pi \calX_{\vb{p}} /2} Z_j = i \calX_{\vb{p}} Z_j \ ;
    \end{align}
    and if $p_j = 0$, $Z_j$ will remain unchanged. 
    We suppose $p_j = 1$ for later illustration.
    Then, we apply the operator corresponding to another row $\vb{p}'^T$, which gives,
    \begin{align}
        i e^{i\pi \calX_{\vb{p}'} /4} \calX_{\vb{p}} Z_j e^{-i\pi \calX_{\vb{p}'} /4} = i \calX_{\vb{p}} e^{i\pi \calX_{\vb{p}'} /4} Z_j e^{-i\pi \calX_{\vb{p}'} /4} \ .
    \end{align}
    If $p'_j = 1$, we have that the post-evolution stabilizer is given by $i^2 \calX_{\vb{p}} \calX_{\vb{p}'} Z_j$.
    In general, let $\mathbf{H}_{j}$ be the submatrix of $\mathbf{H}$ that consists of all rows whose $j$-th entry is 1.
    Then, after the conjugation of $e^{i\pi H /4}$, we have 
    \begin{align}
        e^{i\pi H /4} Z_j e^{-i\pi H /4} = i^{|\vb{c}_j|} \prod_{\vb{p}^T \in \row(\mathbf{H}_{j}) } \calX_{\vb{p}} Z_j \ .
    \end{align}
    For the Pauli $X$'s in the above, whether there is the $X_k$ component depends on the number of 1's in both the $j$-th and $k$-th column of $\vb{H}$. 
    Indeed, the exponent of $X_k$ is equal to $\vb{c}_j \cdot \vb{c}_k$. This completes the proof.
\end{proof}

Next, we are ready to prove Theorem~\ref{thm:IQP_tableau}.
\begin{proof}[Proof of Theorem~\ref{thm:IQP_tableau}]
    Since $Z_j$ is the $j$-th stabilizer generator fo $\ket{0^n}$, Lemma~\ref{lemma:evolution_of_stabilizer} actually gives the $j$-th stabilizer generator of $U_{\vb{H}, \pi/4} \ket{0^n}$.
    We can also write it in the following form,
    \begin{align}\label{eq:IQP_stabilizer_generator}
        (-1)^{r_j} \prod_{k=1}^n i^{\vb{c}_j \cdot \vb{c}_j} X_k^{\vb{c}_j \cdot \vb{c}_k} Z_j \ \ ,
    \end{align}
    where $2 r_j + \vb{c}_j \cdot \vb{c}_j = |\vb{c}_j| \pmod{4}$ (note that the inner product is taken over $\FF_2$).
    Therefore, if one uses $00, 01, 10, 11$ to represent $|\vb{c}_j| = 0, 1, 2, 3 \pmod{4}$, then $r_j$ is equal to the first bit.
    Finally, from this form of stabilizer generators, we can write down the stabilizer tableau of $\ket{\psi}$ as 
    \begin{align}
        \begin{pmatrix}
            \vb{c}_1 \cdot \vb{c}_1 & \ldots & \vb{c}_1 \cdot \vb{c}_n & 1 & \ldots & 0 & r_1 \\
            \vdots & \ddots & \vdots & \vdots & \ddots & \vdots & \vdots \\
            \vb{c}_n \cdot \vb{c}_1 & \ldots & \vb{c}_n \cdot \vb{c}_n & 0 & \ldots & 1 & r_n \\
        \end{pmatrix} \ .
    \end{align}
\end{proof}

\subsection{Correlation function}

Here, we prove Theorem~\ref{thm:cor_func_rank} which connects the correlation function $\ev{\calZ_{\vb{s}}}$, the code $\calC_{\vb{s}}$ generated by $\vb{H}_{\vb{s}}$ and the rank of the Gram matrix $\vb{G} = \vb{H}_{\vb{s}}^T \vb{H}_{\vb{s}}$.

\begin{proof}[Proof of Theorem~\ref{thm:cor_func_rank}]
    First, $\vb{H}_{\vb{s}}\, \vb{s} = \vb{1}$ implies that $\calC_{\vb{s}}^{\perp}$ is an even code and so is $\calD_{\vb{s}}$.
    When $\theta = \pi/8$, we have
    \begin{align}
        \ev{\calZ_{\vb{s}}} &= \mel{0^n}{U_{\vb{H}_{\vb{s}}, 2\theta}}{0^n} \\
        &= \mel{0^n}{\prod_{\vb{p}^T \in \row(\vb{H}_{\vb{s}})} e^{i2\theta \calX_{\vb{p}}}}{0^n} \\
        &= \mel{0^n}{\prod_{\vb{p}^T \in \row(\vb{H}_{\vb{s}})} \frac{1}{\sqrt{2}} (I + i\calX_{\vb{p}})}{0^n} \\
        &= \frac{1}{\sqrt{2^m}} \sum_{\vb{a} \in \{0, 1\}^m} i^{|\vb{a}|} \mel{0^n}{\calX_{\vb{a}^T \vb{H}_{\vb{s}}}}{0^n} \\
        &= \frac{1}{\sqrt{2^m}} \sum_{\vb{a}: \vb{a}^T \vb{H}_{\vb{s}} = \vb{0}} i^{|\vb{a}|} \\
        &= \frac{1}{\sqrt{2^m}} \sum_{\vb{a} \in \calC_{\vb{s}}^{\perp}} i^{|\vb{a}|} \ ,
    \end{align}
    where $m$ is the number of rows in $\vb{H}_{\vb{s}}$.
    Since $\calC_{\vb{s}}^{\perp}$ is an even code, we can write,
    \begin{align}
        \label{eq:correlation_dual_code}
        \ev{\calZ_{\vb{s}}} = \frac{1}{\sqrt{2^m}} \left( \sum_{\substack{ \vb{a} \in \calC_{\vb{s}}^{\perp} \\ |\vb{a}| = 0 \mod{4}}} 1 - \sum_{\substack{ \vb{a} \in \calC_{\vb{s}}^{\perp} \\ |\vb{a}| = 2 \mod{4}}} 1  \right) \ .
    \end{align}
    Let $d = \dim(\calD_{\vb{s}})$, $r = \dim(\calC_{\vb{s}})$ and $g = r - d$. 

    One can always find an invertible matrix $\vb{Q}$, such that in $\vb{H}_{\vb{s}} \vb{Q}$, the first $g$ columns are in $\calC_{\vb{s}} \backslash \calD_{\vb{s}}$, the $g$-th to the $r$-th columns form a basis of $\calD_{\vb{s}}$ and the remaining columns are all-zeros. 
    This transformation will not change the value of the correlation function according to Eq.~\eqref{eq:correlation_dual_code}, because it preserves the code $\calC_{\vb{s}}$ and hence the dual code $\calC_{\vb{s}}^{\perp}$.
    Under this transformation, the stabilizer tableau related to $\vb{H}_{\vb{s}} \vb{Q}$ is given by $(\vb{G}', \vb{I}_n, \vb{r}')$. 
    Here, only the top-left $g\times g$ submatrix of $\vb{G}'$ can be nonzero, and all other entries are zero. 
    According to Proposition~\ref{prop:Gram_matrix_and_code}, the rank of $\vb{G}$ is also $g$, which means that the $g\times g$ submatrix is full rank.

    As for the phase column $\vb{r}'$, if the basis of $\calD_{\vb{s}}$ have weight 0 modulo 4, then only the first $g$ entries of $\vb{r}'$ can be nonzero, and all other entries are zero, according to Theorem~\ref{thm:IQP_tableau}. 
    In this case, $\calD_{\vb{s}}$ is a doubly-even code.
    For this set of generators represented by the transformed tableau, the number of non-$Z$ generators is $g$, corresponding to the first $g$ rows of $(\vb{G}', \vb{I}_n, \vb{r}')$.
    This is the minimum number over all possible choices, since the top-left submatrix of $\vb{G}'$ is already full rank.
    Thus, the correlation function is nonzero and has a magnitude $2^{-g/2}$, according to Proposition~\ref{prop:overlap}.

    On the other hand, if in $\vb{H}_{\vb{s}} \vb{Q}$, some basis of $\calD_{\vb{s}}$ have weight 2 modulo 4, then the corresponding entries in $\vb{r}'$ are 1, which gives $Z$-products with minus sign in the stabilizer group of $\ket{\psi_{\vb{s}}} := U_{\vb{H}_{\vb{s}}, \pi/4} \ket{0^n}$.
    This means that $\ket{\psi_{\vb{s}}}$ has zero overlap with $\ket{0^n}$ and hence the correlation function is zero.
    In this case, it can be shown that $\calD_{\vb{s}}$ is an unbiased even code using Lemma~\ref{lemma:self_dual_even_code}.
\end{proof}

\subsection{Stabilizer characterization applied to the Shepherd-Bremner construction}

Here, we prove Corollary~\ref{cry:QRC_generator}, which follows from Theorem~\ref{thm:IQP_tableau} and the properties of QRC.

\begin{proof}
    The rank of QRC is $(q+1)/2$, which means that there are $n = (q+3)/2$ columns in $\vb{H}_{\vb{s}}^{\QRC}$. To prove this corollary, it suffices to prove the following,
    \begin{align}
        |\vb{c}_j| &= 3 \pmod{4} & \vb{c}_j \cdot \vb{c}_k &= 1 \pmod{2} \ ,
    \end{align}
    according to Theorem~\ref{thm:IQP_tableau}. 

    First, the number of non-zero quadratic residues modulo $q$ is $(q-1)/2$. Since $q+1$ is a multiple of 8, we have $|\vb{c}_j| = (q-1)/2 = 3 \pmod{4}$ for $j \neq 1$. For $j = 1$, $|\vb{c}_1| = q = 3 \pmod{4}$. 

    As for the second formula, the cases (a) $j = k$, (b) $j =1$ but $k \neq 1$ and (c) $j \neq 1$ but $k = 1$ follow the proof of the first formula. So, we focus on proving it for $j \neq k \neq 1$. 
    Define the extended QRC by appending an extra parity bit to the codeword of QRC, which equals the Hamming weight of the codeword modulo 2.
    From classical coding theory, the extended QRC is self-dual~\cite{macwilliams1977book}. That is, every two codewords of the extended QRC is orthogonal to each other.  
    For $\vb{c}_j$, the added parity bit is 1, since these columns are odd-parity. 
    Then, the fact that the extended codewords are orthogonal to each other implies that $\vb{c}_j \cdot \vb{c}_k = 1 \pmod{2}$. This proves the form of the stabilizer tableau, which represents the generators $\{ -Y_1 X_2 \cdots X_n, -X_1 Y_2 \cdots X_n, \cdots, -X_1 X_2 \cdots Y_n \}$.
    
    Multiplying the first generator to the remaining $n-1$ generators gives the same stabilizer group with a different set of generators $\langle -Y_1 X_2 \cdots X_n, Z_1 Z_2, Z_1 Z_3 \cdots, Z_1 Z_n \rangle$.
    In this representation, the $Z$-type stabilizer generators have a positive phase and the number of non-$Z$ generator is $g = 1$.
    According to Proposition~\ref{prop:overlap}, the correlation function has a magnitude $1/\sqrt{2}$ (a.k.a. 0.854 probability bias) with respect to the secret, regardless of the size parameter $q$.
\end{proof}

\section{Details in stabilizer construction}
\label{app:details_in_stabilizer_construction}

We first give the following parameter constraints of the stabilizer construction, which are necessary conditions for instances from the family $\calH_{n, m, g}$.
\begin{proposition}[Parameter constraints]
    \label{prop:parameter_constraints}
    Given $(\vb{H}, \vb{s}) \in \calH_{n, m, g}$, let $\calD_{\vb{s}} = \calC_{\vb{s}} \bigcap \calC_{\vb{s}}^{\perp}$, where $\calC_{\vb{s}}$ is the code generated by $\vb{H}_{\vb{s}}$ and $\calC_{\vb{s}}^{\perp}$ is the dual code. 
    Let $m_1$ be the number of rows in $\vb{H}_{\vb{s}}$ and $d = \dim(\calD_{\vb{s}})$, which means $\dim(\calC_{\vb{s}}) = g + d$. 
    Then, we have
    \begin{itemize}   
        \item $g + d \leq n$;
        \item $0 < m_1 \leq m$; 
        \item $n - g - d \leq m - m_1$;
        \item $g + 2d \leq m_1$;
        \item $m_1 = g \mod{2}$.
    \end{itemize}
\end{proposition}

\begin{proof}
    The first constraint is from $\rank(\vb{H}_{\vb{s}}) \leq n$.
    The second one is trivial. 
    The third is due to the fact that $\vb{H}$ is of full column rank, which means that the number of redundant rows should be $m - m_1 \geq n - \rank(\vb{H}_{\vb{s}}) = n - g - d$.
    The fourth is because $\dim(\calC_{\vb{s}}) + \dim(\calC_{\vb{s}}^{\perp}) = m_1$ and $\dim(\calD_{\vb{s}}) \leq \dim(\calC_{\vb{s}}^{\perp})$.
    The fifth one is from \cref{thm:standard_form_H} proved later.
\end{proof}

As stated in the main text, given the parameters $m_1, n, d$ and $g$, the stabilizer construction is reduced to sampling a random $\vb{H}_{\vb{s}}$ and $\vb{s}$, so that $\vb{H}_{\vb{s}} = (\vb{F}, \vb{D}, \vb{0}_{m_1 \times (n-r)})$, and $\vb{H}_{\vb{s}} \vb{s} = \vb{1}$, where $r = d + g$.
Here, $\vb{D} \in \FF_2^{m_1 \times d}$ is a generator matrix of a random doubly-even code $\calD_{\vb{s}} = \calC_{\vb{s}} \bigcap \calC_{\vb{s}}^{\perp}$, and $\vb{F} \in \FF_2^{m_1 \times g}$ span $\calC_{\vb{s}} \slash \calD_{\vb{s}}$ satisfying $\vb{D}^T \vb{F} = \vb{0}$ and $\rank(\vb{F}^T \vb{F}) = g$.

In more details, $\vb{D}$ and $\vb{F}$ shall satisfy the following conditions.
\begin{proposition}[Conditions of $\vb{D}$ and $\vb{F}$]
\label{prop:conditions_D_F}
    Given $(\vb{H}, \vb{s}) \in \calH_{n, m, g}$, let $\vb{H}_{\vb{s}}$ be the rows of $\vb{H}$ that are not orthogonal to $\vb{s}$. Then, there exists an invertible $\vb{Q}$, so that $\vb{H}_{\vb{s}} \vb{Q} = (\vb{F}, \vb{D}, \vb{0})$ and
    \begin{itemize}
        \item $\vb{D}$ consists of $d = r - g$ independent vectors with weight 0 modulo 4, which are orthogonal to each other, with $r = \rank(\vb{H}_{\vb{s}})$.
        \item $\vb{F}$ consists of $g$ independent columns from $\ker(\vb{D}^T)$ which lie outside the column space of $\vb{D}$.
        \item $\vb{F}^T \vb{F}$ is a random full-rank $g$-by-$g$ symmetric matrix.
        \item The all-ones vector $\vb{1}$ either explicitly appears as the first column of $\vb{D}$ or $\vb{F}$, or it can be written as the sum of the first two columns of $\vb{F}$.
    \end{itemize}
\end{proposition}

\begin{proof}
    The matrix $\vb{D}$ is taken as the generator matrix of the dual intersection $\calD_{\vb{s}}$, which is a doubly-even code.
    The form of $\vb{D}$ follows from Proposition~\ref{lemma:self_dual_even_code}.
    The second condition is because $\calD_{\vb{s}} \subset \calC_{\vb{s}}^{\perp}$, which implies $\vb{D}^T \vb{F} = \vb{0}$. So, columns of $\vb{F}$ lie in $\ker(\vb{D}^T)$.
    The third condition is because $\rank(\vb{F}^T \vb{F}) = \rank(\vb{Q}^T \vb{H}_{\vb{s}}^T \vb{H}_{\vb{s}} \vb{Q}) = \rank(\vb{H}_{\vb{s}}^T \vb{H}_{\vb{s}}) = g$.
    As for the last condition, if $\vb{1} \in \calC_{\vb{s}}$, one can always perform basis change so that $\vb{1}$ explicitly appears in the columns of $\vb{H}_{\vb{s}}$.
    More specifically, if $\vb{1} \in \calD_{\vb{s}}$, the column operation $\vb{Q}$ can transform it as the first column of $\vb{D}$. 
    If not, it can be made as the first column of $\vb{F}$.
    The second part of this condition can be achieved by adding the second column of $\vb{F}$ to the first.
\end{proof}

\paragraph{Standard form.}

Although Ref.~\cite{kim_two_2008} gave the congruent standard form of any symmetric matrix over $\FF_2$, for our purpose of ensuring the all-ones vector being a codeword, we consider a slightly different standard form.
Such a standard form can be achieved by $\vb{H}_{\vb{s}} = (\vb{F}, \vb{D}, \vb{0})$ where $\vb{F}$ and $\vb{D}$ satisfy the conditions in \cref{prop:conditions_D_F}.
Then, the construction algorithm only needs to sample $\vb{H}_{\vb{s}}$ so that $\vb{H}_{\vb{s}}^T \vb{H}_{\vb{s}}$ is in the standard form. 
Note that $\vb{H}_{\vb{s}}$ itself may not have a unique standard form under row permutations and column operations.

\begin{theorem}
    \label{thm:standard_form_H}
    Let $(\vb{H}, \vb{s})$ be a random instance from $\calH_{n, m, g}$ and let $\vb{H}_{\vb{s}}$ be the rows of $\vb{H}$ satisfying $\vb{H}_{\vb{s}} \ \vb{s} = \vb{1}$. 
    Then, there exists an invertible matrix $\vb{Q}$, so that $\vb{H}_{\vb{s}} \vb{Q}$ is in the form of $(\vb{F}, \vb{D}, \vb{0})$ with $\vb{D}$ and $\vb{F}$ satisfying the conditions in \cref{prop:conditions_D_F} and
    \begin{align}
    \label{eq:standard_form_H}
        \vb{Q}^T \vb{H}_{\vb{s}}^T \vb{H}_{\vb{s}} \vb{Q} = 
        \begin{pmatrix}
            \vb{I} \\
             & \vb{J} \\
             & & \ddots \\
             & & & \vb{J} \\
             & & & & \vb{0}_{(n-g) \times (n-g)}
        \end{pmatrix}
        \text{ or }
        \begin{pmatrix}
            \vb{J} \\
             & \ddots \\
             & & \vb{J} \\
             & & & \vb{0}_{(n-g) \times (n-g)}
        \end{pmatrix} \ ,
    \end{align}
    where $\vb{I}$ is either 1 or $\vb{I}_2$ and $\vb{J} := \begin{pmatrix} 0 & 1 \\ 1 & 0 \end{pmatrix}$.
    In addition, $m_1 = g \mod{2}$, where $m_1$ is the number of rows in $\vb{H}_{\vb{s}}$.
\end{theorem}

\begin{proof}
    First, according to \cref{prop:conditions_D_F}, $\vb{H}_{\vb{s}}^T \vb{H}_{\vb{s}} \sim_c \vb{F}^T \vb{F} \oplus \vb{D}^T \vb{D} \oplus \vb{0}$.
    For the $\vb{D}$ matrix, we have $\vb{D}^T \vb{D} = \vb{0}_{d\times d}$, which already matches standard form, where $d = \rank(\vb{H}_{\vb{s}}) - g$.
    So, we focus on $\vb{G}' := \vb{F}^T \vb{F}$.
    The matrix $\vb{F}$ satisfies $\vb{D}^T \vb{F} = \vb{0}$ and $\rank(\vb{F}^T \vb{F}) = g$, where $g$ is the number of columns in $\vb{F}$.
    We want to find an invertible matrix $\vb{Q}'$, which leaves the $\vb{D}$ matrix unchanged and only changes the $\vb{F}$ matrix, so that $\vb{G}'$ is equal to the top-left $g\times g$ submatrix in the standard form.
    We first discuss the congruent standard form of general full-rank symmetric matrix.
    \begin{enumerate}
        \item First, suppose that not all diagonal elements of $\vb{G}'$ are zero. 
        In this case, we can assume $\vb{G}'_{11} = 1$, because otherwise, we can always apply a permutation matrix to $\vb{F}$, so that the nonzero diagonal element of $\vb{G}'$ is moved to the $(1, 1)$-location.
        Then, up to congruent transformations,
        \begin{align}
        \vb{G}' = 
        \begin{pmatrix}
            1 & \vb{g}^T \\
            \vb{g} & \vb{G}_1 
        \end{pmatrix} \ .
        \end{align}
        Let,
        \begin{align}\label{eq:case_1_Q}
        \vb{Q}_1 = \begin{pmatrix}
            1 & \vb{g}^T \\ 0 & \vb{I}
        \end{pmatrix} \ .
        \end{align}
        We have,
        \begin{align}\label{eq:case_1}
        \vb{Q}_1^T \vb{G}' \vb{Q}_1 = \begin{pmatrix}
            1 & \vb{0}^T \\ \vb{0} & \vb{g} \vb{g}^T + \vb{G}_1
        \end{pmatrix} \ .
        \end{align}

        \item If $\vb{G}'_{jj} = 0$ for $1 \leq j \leq g$, then without loss of generality, we can assume $\vb{G}_{12}' = 1$; otherwise, we can apply a permutation matrix to swap the the non-zero entry to the $(1, 2)$ and $(2, 1)$ positions.
        In this case, up to congruent transformations,
        \begin{align}
        \vb{G}' = \begin{pmatrix}
            \vb{J} & \vb{G}_2 \\ \vb{G}_2^T & \vb{G}_3
        \end{pmatrix} \ ,
        \end{align}
        where $\vb{J} = \begin{pmatrix} 0 & 1 \\ 1 & 0 \end{pmatrix}$.
        Let
        \begin{align}\label{eq:case_2_Q}
        \vb{Q}_1 = \begin{pmatrix}
            \vb{I}_2 & \vb{J} \vb{G}_2 \\ \vb{0} & \vb{I}_{g - 2}
        \end{pmatrix} \ .
        \end{align}
        Then, 
        \begin{align}
        \vb{Q}_1^T \vb{G}' \vb{Q}_1 = \begin{pmatrix}
            \vb{J} & \vb{0}^T \\ \vb{0} & \vb{G}_2^T \vb{J} \vb{G}_2 + \vb{G}_3
        \end{pmatrix} \ .
        \end{align}
    \end{enumerate}
    Therefore, we have $\vb{F}^T \vb{F} \sim_c \left( \vb{I} \oplus \right) \ \vb{J} \oplus \cdots \oplus \vb{J}$.

    \begin{itemize}
        \item If $m_1$ is odd, then $\vb{1}$ cannot be in $\calD_{\vb{s}}$. 
        In this case, $\vb{1}$ is the first column of $\vb{F}$, according to \cref{prop:conditions_D_F}. 
        Then, $\vb{G}'_{11} = 1$ and applying the transformation of \cref{eq:case_1_Q} leads to a matrix in the form of \cref{eq:case_1}. This implies that all other columns in the new $\vb{F}$ are orthogonal to $\vb{1}$ and hence have even parity.
        Therefore, we have $\vb{G}'_{jj} = 0$ for $j > 1$ and $\vb{G}' \sim_c 1 \oplus \left( \bigoplus\limits_{i = 1}^{(g-1)/2} \vb{J} \right)$.

        \item If $m_1$ is even and $\vb{1} \not\in \calD_{\vb{s}}$, then $\vb{1}$ is also the first column of $\vb{F}$, according to \cref{prop:conditions_D_F}.
        Then, we can assume that $\vb{G}_{12}' = 1$, as what we did in proving the general congruent standard form.
        This implies that the second column of $\vb{F}$ must be odd-parity, and so $\vb{G}_{22}' = 1$.
        As as result, up to congruent transformations, 
        \begin{align}
        \vb{G}' = \begin{pmatrix}
            \vb{J}_1 & \vb{G}_2 \\ \vb{G}_2^T & \vb{G}_3
        \end{pmatrix} \ ,
        \end{align}
        where $\vb{J}_1 = \begin{pmatrix} 0 & 1 \\ 1 & 1 \end{pmatrix}$.
        Let 
        \begin{align}
        \vb{Q}_1 = \begin{pmatrix}
            \vb{I}_2 & \vb{J}_3 \vb{G}_2 \\ \vb{0} & \vb{I}_{g - 2}
        \end{pmatrix} \ ,
        \end{align}
        where $\vb{J}_3 = \begin{pmatrix} 1 & 1 \\ 1 & 0 \end{pmatrix}$.
        Then, 
        \begin{align}
        \vb{Q}_1^T \vb{G}' \vb{Q}_1 = \begin{pmatrix}
            \vb{J}_1 & \vb{0}^T \\ \vb{0} & \vb{G}_2^T \vb{J}_3 \vb{G}_2 + \vb{G}_3
        \end{pmatrix} \ .
        \end{align}
        From the second row (column) of $\vb{Q}_1^T \vb{G}' \vb{Q}_1$, one can see that only the second column of the new $\vb{F}$ is odd-parity and all other columns are even parity, which means that the diagonal elements of $\vb{G}_2^T \vb{J}_3 \vb{G}_2 + \vb{G}_3$ are zero.
        Then, \cref{eq:case_2_Q} is repeatedly applied, so that $\vb{G}' = \vb{J}_1 \oplus \vb{J} \oplus \cdots \oplus \vb{J}$.
        Finally, let $\vb{Q}_2$ be an invertible matrix that adds the second column to the first.
        We have $\vb{Q}_2^T \vb{G}' \vb{Q}_2 = \vb{I}_2 \oplus \vb{J} \oplus \cdots \oplus \vb{J}$.

        \item If $m_1$ is even and $\vb{1} \in \calD_{\vb{s}}$, then $m_1$ must be a multiple of $4$ and $\vb{1}$ is the first column of $\vb{D}$, according to \cref{prop:conditions_D_F}.
        In this case, $\vb{D}^T \vb{F} = \vb{0}$ implies that all columns of $\vb{F}$ will be even-parity, which means that the diagonal elements of $\vb{G}'$ will be zero. 
        Moreover, the diagonal elements will remain zero if the congruent transformation only acts nontrivially on $\vb{G}'$.
        Therefore, $\vb{F}^T \vb{F} \sim_c \bigoplus\limits_{i = 1}^{g/2} \vb{J}$.
    \end{itemize}

    Above, the all-ones vector appears as the first column of $\vb{D}$ or $\vb{F}$ except for the third case, where $\vb{1} = \vb{c}_1 + \vb{c}_2$ can be obtained by adding up the first two columns of $\vb{F}$.
    Finally, in all of the above cases, $m_1 = g \mod{2}$.
\end{proof}

\begin{algorithm}[H]
    \centering
    \begin{algspec}
        \textbf{Parameters:} $m_1$ and $d$ \\
        \textbf{Require:} $d \leq m_1/2$
        \begin{algorithmic}[1]
        \State{$\vb{c}_1 \gets$ a random vector with weight 0 modulo 4}
        \State{$\vb{D} \gets (\vb{c}_1)$}
        \For{$t = 1, \cdots, d-1$}
            \State{$\vb{c}_{t+1} \gets$ a random vector from $\ker(\vb{D}^T) \slash \langle \vb{c}_1, \cdots, \vb{c}_t \rangle$ with weight 0 modulo 4}
            \If{$\vb{c}_{t+1}$ does not exist}
                \State{break}
            \EndIf
            \State{$\vb{D} \gets (\vb{D}, \vb{c}_{t+1})$}
        \EndFor
        \If{$\vb{1}$ lies in the column space of $\vb{D}$}
            \State{Apply column operations so that $\vb{1}$ is the first column of $\vb{D}$}
        \EndIf
        \State{\Return{$\vb{D}$}}
        \end{algorithmic}
    \end{algspec}
    \caption{Algorithm to sample a $\vb{D} = (\vb{c}_1, \vb{c}_2, \cdots, \vb{c}_d)$ so that $\vb{c}_i \cdot \vb{c}_j = 0$ and $|\vb{c}_i| = 0 \mod{4}$.}
    \label{alg:sample_D}
\end{algorithm}

Next, we discuss the sampling of $\vb{D}$ and $\vb{F}$.

\paragraph{Sampling $\vb{D}$.}
Here, the goal is to sample a $\vb{D} = (\vb{c}_1, \cdots, \vb{c}_d)$ with $d \leq m_1/2$, where $m_1 = g \mod{2}$.
Columns in $\vb{D}$ are orthogonal to each other and have weight a multiple of 4, according to \cref{prop:conditions_D_F}.
The algorithm is shown in Algorithm~\ref{alg:sample_D}, which works as follows. 
First, $\vb{c}_1$ can be a random vector with weight 0 modulo 4; $\vb{D}$ is initialized as $\vb{D} = (\vb{c}_1)$.
Then, the second column $\vb{c}_2$ is sampled with the constraint that $\vb{c}_1 \cdot \vb{c}_2 = 0$ and $|\vb{c}_2| = 0 \mod{4}$; $\vb{D}$ is updated to be $\vb{D} = (\vb{c}_1, \vb{c}_2)$.
Next, the third column $\vb{c}_3$ is sampled so that it is orthogonal the first two columns and $|\vb{c}_3| = 0 \mod{4}$.
This process is iterated until all $d$ columns are sampled, or until no vector satisfying the condition can be sampled, in which case a matrix with $d-1$ columns will be returned.

In the $t$-th iteration, the vector $\vb{c}_t$ is sampled from $\ker(\vb{D}^T) \slash \langle \vb{c}_1, \cdots, \vb{c}_{t-1} \rangle$ with $\vb{D} = (\vb{c}_1, \cdots, \vb{c}_{t-1})$.
That is, we want $\vb{c}_t$ to be orthogonal to the first $t - 1$ columns and outside the linear subspace that they span.
This can be achieved as follows.
We first solve for a basis of $\ker(\vb{D}^T)$, and then the first $t-1$ of the basis vectors are set as $\{ \vb{c}_1, \cdots, \vb{c}_{t-1} \}$, with the remaining basis vectors changed accordingly.
The vector $\vb{c}_t$ is sampled to be the random linear combination of the remaining basis vectors.
In this way, the orthogonality and independence of $\vb{c}_t$ with respect to $\vb{c}_1, \cdots, \vb{c}_{t-1}$ are guaranteed.

In addition, to ensure that $|\vb{c}_t| = 0 \mod{4}$, we can first sample an even-parity vector from $\ker(\vb{D}^T) \slash \langle \vb{c}_1, \cdots, \vb{c}_{t-1} \rangle$.
It is well-known that for a linear subspace over $\FF_2$, either all vectors are even-parity or half the vectors are even-parity.
Therefore, the sampling of even-parity vector can be efficiently done and we denote resulted vector as $\vb{a}_1$.
The weight of $\vb{a}_1$ will be either 0 or 2 modulo 4.
If $|\vb{a}_1| = 0 \mod{4}$, then it is set to be $\vb{c}_t$.
Otherwise, we sample a vector $\vb{a}_2$ from $\ker(\vb{D}^T) \slash \langle \vb{c}_1, \cdots, \vb{c}_{t-1}, \vb{a}_1 \rangle$ that is orthogonal to $\vb{a}_1$; 
that is, $\vb{a}_2$ is a random vector from $\ker(\vb{D}^T)$ that is orthogonal to and outside $\langle \vb{c}_1, \cdots, \vb{c}_{t-1}, \vb{a}_1 \rangle$.
Then, if $|\vb{a}_2| = 0 \mod{4}$, it is set to be $\vb{c}_t$ and if not, it follows from Lemma~\ref{lemma:sum_of_codewords} that $\vb{a}_1 + \vb{a}_2$ must have a weight that is a multiple of 4, and thus we assign it as $\vb{c}_{t}$.
With this approach, a $\vb{c}_t$ with weight a multiple of 4 can be guaranteed to be sampled except for the final iteration of two extremal cases.

We now turn to discuss the cases $d = m_1/2$ and $d = (m_1 - 1)/2$, where in the last iteration, the column $\vb{c}_d$ \emph{may} not exist.
In such a case, only a matrix $\vb{D}$ with $d - 1$ columns will be returned.
However, we would like to emphasize that it is the $g$ parameter that affects the value of the correlation function.
For the subspace $\calD_{\vb{s}}$, we only require it to be doubly-even, and its dimension does not matter.
Therefore, we do not require the sampling algorithm succeed every time when applied to these two extremal cases.

When $d = m_1/2$ with $m_1$ even, the resulting $\calD_{\vb{s}}$ forms a doubly-even self-dual code, which implies $g = 0$.
In this scenario, $\vb{1}$ is included in $\calD_{\vb{s}}$ because $\vb{1}$ will be orthogonal to all vectors in $\calD_{\vb{s}}$ and itself.
This implies that $m_1$ must be a multiple of 4.
An example of this case is the extended QRC~\cite{macwilliams1977book}.
On the other hand, when $d = (m_1 - 1)/2$ with $m_1$ odd, we have $g = 1$ along with $\calD_{\vb{s}} = \calC_{\vb{s}}^{\perp}$. 
In this situation, the $\vb{F}$ matrix must be $\vb{1}$.
The QRC serves as an example of this particular case.

The reason why the last iteration of Algorithm~\ref{alg:sample_D} may break on these two cases is as follows.
We only discuss $d = m_1/2$, but the reasoning for $d = (m_1 - 1)/2$ is similar.
When $t = d - 1$, $\vb{D} = (\vb{c}_1, \cdots, \vb{c}_{d-1})$. 
If $\vb{1}$ is not in $\langle \vb{c}_1, \cdots, \vb{c}_{d-1} \rangle$, then the algorithm will assign it as $\vb{c}_d$, and this iteration ends normally.
However, if $\vb{1} \in \langle \vb{c}_1, \cdots, \vb{c}_{d-1} \rangle$, the dimension of $\ker(\vb{D}^T)$ is $m_1 - (d - 1) = m_1/2 + 1$.
So, the subspace that $\vb{a}_1$ is sampled from has dimension $m_1/2 + 1 - t =  2$.
If this subspace has a nonzero vector with weight 0 modulo 4, then this vector will be assigned as $\vb{c}_d$ and the iteration will also end normally.
However, if all vectors in this subspace have weight 2 modulo 4, except for the all-zeros vector, then $\vb{c}_d$ could not be found.
In this case, the iteration breaks.

\begin{algorithm}[H]
    \centering
    \begin{algspec}
        \textbf{Input:} $m_1, g$ and $\vb{D}$ \\
        \textbf{Require:} $g \leq m_1 - 2d$ with $d$ the number of columns in $\vb{D}$; $g = m_1 \mod{2}$
        \begin{algorithmic}[1]
        \State{$\calD \gets$ column space of $\vb{D}$}
        \If{$m_1$ is odd}  \Comment{$\vb{G}' = \mathrm{diag}(1, \vb{J}, \cdots, \vb{J})$ and $g$ is odd}
            \State{$\vb{c}_1 \gets \vb{1}_{m_1}$}
            \State $\vb{F} \gets (\vb{c}_1)$
        \ElsIf{$m_1$ is even and $\vb{1}_{m_1} \not\in \calD$}  \Comment{$\vb{G}' = \mathrm{diag}(\vb{I}_2, \vb{J}, \cdots, \vb{J})$}
            \State{$\vb{c}_2 \gets$ a random odd-parity vector in $\ker(\vb{D}^T) \slash \calD$}
            \State{$\vb{c}_1 \gets \vb{1}_{m_1} + \vb{c}_2$}
            \State $\vb{F} \gets (\vb{c}_1, \vb{c}_2)$
        \Else  \Comment{$\vb{G}' = \mathrm{diag}(\vb{J}, \cdots, \vb{J})$}
            \State{$\vb{c}_1 \gets$ a random vector in $\ker(\vb{D}^T) \slash \calD$}
            \State{$\vb{c}_2 \gets$ a random vector in $\ker(\vb{D}^T) \slash \calD$ that satisfies $\vb{c}_1 \cdot \vb{c}_2 = 1$}
            \State $\vb{F} \gets (\vb{c}_1, \vb{c}_2)$
        \EndIf

        \While{number of columns in $\vb{F} < g$}
            \State $\calC \gets \calD \oplus$ column space of $\vb{F}$
            \State $\vb{a} \gets$ a random vector in $\calC^{\perp} \slash \calD$
            \State $\vb{b} \gets$ a random vector in $\calC^{\perp} \slash \calD$ that satisfies $\vb{a} \cdot \vb{b} = 1$
            \State $\vb{F} \gets (\vb{F}, \vb{a}, \vb{b})$
        \EndWhile
        \State{\Return{$\vb{F}$}}
        \end{algorithmic}
    \end{algspec}
    \caption{Algorithm to sample $\vb{F} = (\vb{c}_1, \cdots, \vb{c}_g)$ so that $\vb{D}^T \vb{F} = \vb{0}$ and $\rank(\vb{F}^T \vb{F}) = g$.}
    \label{alg:sample_F}
\end{algorithm}

\paragraph{Sampling $\vb{F}$.}
Next, we give the algorithm to sample $\vb{F} = (\vb{c}_1, \cdots, \vb{c}_g)$ (Algorithm~\ref{alg:sample_F}).
The algorithm takes $m_1, g$ and $\vb{D}$ as inputs, and outputs a matrix $\vb{F}$ so that $\vb{F}^T \vb{F}$ is a rank-$g$ symmetric matrix in the standard form as in Eq.~\eqref{eq:standard_form_H}, $\vb{D}^T \vb{F} = \vb{0}$, and $\vb{1}$ lies in the span of columns in $\vb{D}$ and $\vb{F}$.
This implies that all columns of $\vb{F}$ should be sampled from $\ker(\vb{D}^T) \slash \calD_{\vb{s}}$, with the additional orthogonality constraints imposed by $\vb{F}^T \vb{F}$.

There are three cases for $\vb{F}^T \vb{F}$.
First, if $m_1$ is odd, then $\vb{1}$ cannot lie in $\calD_{\vb{s}}$. 
According to \cref{prop:conditions_D_F} and \cref{thm:standard_form_H}, $\vb{1}$ can be set as the first column of $\vb{F}$ and $\vb{F}^T \vb{F} = \mathrm{diag}(1, \vb{J}, \cdots, \vb{J})$.
Second, if $m_1$ is even but $\vb{1}$ is not in $\calD_{\vb{s}}$, then $\vb{F}^T \vb{F} = \mathrm{diag}(\vb{I}_2, \vb{J}, \cdots, \vb{J})$, according to \cref{thm:standard_form_H}.
In this case, $\vb{c}_1$ and $\vb{c}_2$ are odd-parity vectors, and $\vb{c}_1 + \vb{c}_2 = \vb{1}$.
Third, if $m_1$ is even and $\vb{1}$ lies in $\calD_{\vb{s}}$, then $\vb{F}^T \vb{F} = \mathrm{diag}(\vb{J}, \cdots, \vb{J})$.
In this case, $\vb{c}_1$ is a random vector from $\ker(\vb{D}^T) \slash \calD_{\vb{s}}$ and $\vb{c}_2$ is a random vector from $\ker(\vb{D}^T) \slash \calD_{\vb{s}}$ satisfying $\vb{c}_2 \cdot \vb{c}_1 = 1$.
Note that $\vb{c}_1$ and $\vb{c}_2$ are automatically even-parity, since they are orthogonal to $\vb{1}$.
Moreover, $\vb{c}_2$ must lie outside the space $\calD_{\vb{s}} \oplus \ev{\vb{c}_1}$, due to the constraint $\vb{c}_2 \cdot \vb{c}_1 = 1$.

After the initialization of $\vb{c}_1$ (and $\vb{c}_2$), the algorithm proceeds to sample other columns of $\vb{F}$, if $g > 1$ (or $g > 2$). 
We only illustrate the case when $m_1$ is odd below, but the sampling process for an even $m_1$ follows a similar pattern.
For $m_1$ odd, $\vb{F}^T \vb{F} = \mathrm{diag}(1, \vb{J}, \cdots, \vb{J})$.
We first initialize $\calC_{\vb{s}} \gets \calD_{\vb{s}} \oplus \ev{\vb{c}_1}$, which is the subspace $\calC$ in \cref{alg:sample_F}.
The block diagonal form of $\vb{F}^T \vb{F}$ implies that $\vb{c}_2$ and $\vb{c}_3$ are vectors from $\ker(\vb{D}^T) \slash \calD_{\vb{s}}$ that are orthogonal to $\vb{c}_1 = \vb{1}$, i.e., $\vb{c}_2, \vb{c}_3 \in \calC_{\vb{s}}^{\perp} \slash \calD_{\vb{s}}$.
For $\vb{c}_2$, it is sampled as a random vector from $\calC_{\vb{s}}^{\perp} \slash \calD_{\vb{s}}$, which is the vector $\vb{a}$ in \cref{alg:sample_F}.
For $\vb{c}_3$, it is sampled as a random vector from $\calC_{\vb{s}}^{\perp} \slash \calD_{\vb{s}}$ satisfying $\vb{c}_2 \cdot \vb{c}_3 = 1$, which is the vector $\vb{b}$ in \cref{alg:sample_F}.
This finishes the sampling of columns corresponding to the first $\vb{J}$ block.
Then, the subspace $\calC_{\vb{s}}$ is updated to include $\vb{c}_2$ and $\vb{c}_3$ into its basis, with its dimension increased by 2.
This process is repeated for other columns, until all $g$ columns are sampled.
Finally, $\vb{F} = (\vb{c}_1, \cdots, \vb{c}_g)$ and it can be verified that $\vb{F}^T \vb{F}$ is indeed equal to the standard form $\mathrm{diag}(1, \vb{J}, \cdots, \vb{J})$.

\section{Column redundancy}\label{app:column_redundancy}

Essentially, adding column redundancy is to replace a full rank generator matrix
of a code with a ``redundant'' generator matrix.
The procedure of adding column redundancy is as follows.
\begin{enumerate}
  \item Given a full-rank $\vb{H}_{\vb{s}}$, (e.g., the last 4 columns in
        Eq.~\eqref{eq:QRC_main_part}) and the secret, we first append all-zeros
        columns to $\vb{H}_{\vb{s}}$ and extend $\vb{s}$ accordingly,
    \begin{align}\label{eq:add_column_redundancy}
      \vb{H}_{\vb{s}}
      & \gets (\vb{H}_{\vb{s}}, \vb{0})
      & \vb{s} \gets \begin{pmatrix} \vb{s} \\ \vb{s}' \end{pmatrix} \; .
    \end{align}

  \item Apply random column operations $\vb{Q}$ to obtain
        $\vb{H}_{\vb{s}} \gets \vb{H}_{\vb{s}} \vb{Q}$ and
        $\vb{s} \gets \vb{Q}^{-1} \vb{s}$.
\end{enumerate}
Here, in the first step $\vb{s}'$ is an arbitrary vector whose length is the
same as the number of all-zeros columns appended to $\vb{H}_{\vb{s}}$.
Since the correlation function only depends on the linear code generated by
$\vb{H}_{\vb{s}}$~\cite{shepherd_temporally_2009,yung_anti-forging_2020} and
adding column redundancy does not change the linear code, the correlation
function with respect to the new secret is unchanged after the above two steps.
We would like to remark that although there are $2^{n_2}$ choices for $\vb{s}'$
of length $n_2$, once we fix a choice, the only constraint to the redundant rows
is to be orthogonal to the specific new secret $\vb{s}$.
Moreover, since the final IQP matrix $\vb{H}$ is of full column rank, only the
real secret $\vb{s}$ will correspond to the code generated by $\vb{H}_{\vb{s}}$.
In the case of QRC-based construction, if one chooses a redundant generator
matrix of QRC by adding column redundancy, then $n$ can be any integer larger
than $(q+1)/2$, the dimension of QRC.

\section{Classical sampling and equivalent secrets}
\label{app:equivalent_secrets}

Here, we show that if a classical prover finds a wrong secret $\vb{s}'$, and generate classical samples using the naive sampling algorithm, then the generated samples will have the wrong correlation function on the real secret.
We first prove the following proposition.

\begin{proposition}\label{prop:even_prob}
    For a nonzero $\vb{s} \neq \vb{1}$, if we randomly sample a vector $\vb{d}$ of even parity, then
    \begin{align}
        \Pr_{|\vb{d}| \text{ even}}(\vb{s} \cdot \vb{d} = 1) = \Pr_{|\vb{d}| \text{ even}}(\vb{s} \cdot \vb{d} = 0) = \frac{1}{2} \ .
    \end{align}
\end{proposition}

\begin{proof}
    The set of even-parity vectors forms a linear subspace. It is well known that for a linear subspace and a vector $\vb{s}$, either half of vectors are orthogonal to $\vb{s}$, or all vectors are orthogonal to $\vb{s}$.
    Since $\vb{s} \neq \vb{1}$, we have the former case.
\end{proof}

Given $\vb{H} = \begin{pmatrix} \vb{H}_{\vb{s}} \\ \vb{R}_{\vb{s}} \end{pmatrix}$ and $\vb{s}$, we say another vector $\vb{s}'$ is equivalent to $\vb{s}$ if $\vb{H} \vb{s} = \vb{H} \vb{s}'$; that is, they have the same inner-product relations with rows in $\vb{H}$.
The following lemma shows that a random row orthogonal to $\vb{s}'$ will have probability $1/2$ to have inner product 1 with $\vb{s}$,  even if $\vb{H} \vb{s} = \vb{H} \vb{s}'$.

\begin{lemma}\label{lemma:one_redundant_row}
    For $\vb{s}' \neq \vb{s}$, if we uniformly randomly sample a vector $\vb{p}$ orthogonal to $\vb{s}'$, then
    \begin{align}
        \Pr_{\vb{p} \cdot \vb{s}' = 0} (\vb{p} \cdot \vb{s} = \vb{1}) = \frac{1}{2} \ .
    \end{align}
\end{lemma}

\begin{proof}
    Without loss of generality, assume $\vb{s}'$ has ones in the first $k$ entries and zeros elsewhere. 
    We can split $\vb{p} = \begin{pmatrix} \vb{p}_1 \\ \vb{p}_2 \end{pmatrix}$ and $\vb{s} = \begin{pmatrix} \vb{s}_1 \\ \vb{s}_2 \end{pmatrix}$, where $\vb{p}_1$ is a random even-parity string and $\vb{p}_2$ is uniformly random over $\FF_2^{n-k}$.
    Then, $\vb{p} \cdot \vb{s} = \vb{p}_1 \cdot \vb{s}_1 + \vb{p}_2 \cdot \vb{s}_2$.
    \begin{enumerate}
        \item If $\vb{s}_2 = \vb{0}$ and $\vb{s}_1 \neq \vb{1}$, 
        \begin{align}
            \Pr_{\vb{p} \cdot \vb{s}' = 0} (\vb{p} \cdot \vb{s} = \vb{1}) = \Pr_{\vb{p}_1 \text{ even}} (\vb{p}_1 \cdot \vb{s}_1 = \vb{1}) = \frac{1}{2} \ ,
        \end{align}
        according to Proposition~\ref{prop:even_prob}.

        \item If $\vb{s}_2 \neq \vb{0}$ and $\vb{s}_1 = \vb{1}$,
        \begin{align}
            \Pr_{\vb{p} \cdot \vb{s}' = 0} (\vb{p} \cdot \vb{s} = \vb{1}) = \Pr_{\vb{p}_2} (\vb{p}_2 \cdot \vb{s}_2 = \vb{1}) = \frac{1}{2} \ ,
        \end{align}
        because $\vb{p}_2$ is uniformly random.

        \item If $\vb{s}_2 \neq \vb{0}$ and $\vb{s}_1 \neq \vb{1}$,
        \begin{align}
            \Pr_{\vb{p} \cdot \vb{s}' = 0} (\vb{p} \cdot \vb{s} = \vb{1}) &= \Pr_{\vb{p}_1, \vb{p}_2} (\vb{p}_1 \cdot \vb{s}_1 = \vb{1}, \vb{p}_2 \cdot \vb{s}_2 = \vb{0}) + \Pr_{\vb{p}_1, \vb{p}_2} (\vb{p}_1 \cdot \vb{s}_1 = \vb{0}, \vb{p}_2 \cdot \vb{s}_2 = \vb{1}) \\
            &= \frac{1}{2} \cdot \frac{1}{2} + \frac{1}{2} \cdot \frac{1}{2} = \frac{1}{2} \ ,
        \end{align}
        where we used the independence of $\vb{p}_1$ and $\vb{p}_2$.
    \end{enumerate}
\end{proof}

Now, we are ready to prove Lemma~\ref{lemma:zero_correlation}.

\begin{proof}
    With similar derivations to Lemma~\ref{lemma:one_redundant_row}, one can show that
    \begin{align}
        \Pr_{\vb{p} \cdot \vb{s}' = 0} (\vb{p} \cdot \vb{s} = \vb{0}) = \Pr_{\vb{p} \cdot \vb{s}' = 1} (\vb{p} \cdot \vb{s} = \vb{0}) = \Pr_{\vb{p} \cdot \vb{s}' = 1} (\vb{p} \cdot \vb{s} = \vb{1}) = \frac{1}{2} \ .
    \end{align}
    That means, if one samples a random $\vb{p}$ to be orthogonal to $\vb{s}'$ with probability $\beta$, and not orthogonal to $\vb{s}'$ with probability $1 - \beta$, then
    \begin{align}
    \Pr_{\vb{p}} (\vb{p} \cdot \vb{s} = 0) = \beta \Pr_{\vb{p} \cdot \vb{s} = 0}(\vb{p} \cdot \vb{s} = 0) + (1 - \beta) \Pr_{\vb{p} \cdot \vb{s} = 1}(\vb{p} \cdot \vb{s} = 0) = \frac{1}{2} \ .
    \end{align}
    That is, $\vb{p}$ is uncorrelated with $\vb{s}$ and the correlation function is
    \begin{align}
        \EE [(-1)^{\vb{p} \cdot \vb{s}}] = \Pr_{\vb{p}} (\vb{p} \cdot \vb{s} = 0) - \Pr_{\vb{p}} (\vb{p} \cdot \vb{s} = 1) = 0 \ .
    \end{align}
    Therefore, if the secret extraction procedure returns a vector $\vb{s}' \neq \vb{s}$ and the classical prover uses the naive classical sampling algorithm to generate samples, then the samples will produce zero correlation function on the real secret.
\end{proof}

\section{More on Linearity Attack}
\label{app:more_attack}

\begin{metaalgorithm}[H]
    \centering
    \begin{algspec}
        \textbf{Parameter:} number of linear equations $l$
        \begin{algorithmic}[1]
        \Procedure{ExtractSecret}{$\vb{H}$} 
            \State{Uniformly randomly pick $\vb{d} \in \FF_2^n$.}
            \For{$j = 1, 2, \cdots, l$} \Comment{construct the linear-system matrix $\vb{M}$}
                \State{Uniformly randomly pick $\vb{e}_j \in \FF_2^n$.}
                \State{$\vb{m}_j^T \gets \textsc{RowSum}(\vb{H}_{\vb{d}, \vb{e}_j})$.}
            \EndFor{}
            \State{$\vb{M} \gets (\vb{m}_1, \cdots, \vb{m}_l)^T$.}
            \For{each vector $\vb{s}_i \in \ker(\vb{M})$} 
                \If{$\vb{s}_i$ passes the QRC check}  \Comment{discussed in the main text}
                    \State{\Return{$\vb{s}_i$}}
                \EndIf
            \EndFor
        \EndProcedure
        \end{algorithmic}
    \end{algspec}
    \caption{The \textsc{ExtractSecret}($\vb{H}$) subroutine in Ref.~\cite{kahanamoku-meyer_forging_2023}. Here, given $\vb{H}$ and two vectors $\vb{d}$ and $\vb{e}_t$, we define $\vb{H}_{\vb{d}, \vb{e}_j}$ to be a submatrix from $\vb{H}$ by deleting rows orthogonal to either $\vb{d}$ or $\vb{e}_j$.}
    \label{alg:extract_secret_KM}
\end{metaalgorithm}

\subsection{Secret extraction in Kahanamoku-Meyer's attack}
\label{app:secret_extraction_KM}

Meta-Algorithm~\ref{alg:extract_secret_KM} presents the secret extraction procedure in Ref.~\cite{kahanamoku-meyer_forging_2023}.
The procedure will be repeated if no vectors pass the QRC check.
Here, we prove the following proposition.
\begin{proposition}
    \label{prop:M_and_G_d}
    The matrix $\vb{M}$ obtained from Meta-Algorithm~\ref{alg:extract_secret_KM} consists of rows from the row space of $\vb{G}_{\vb{d}} = \vb{H}_{\vb{d}}^T \vb{H}_{\vb{d}}$.
\end{proposition}

From this proposition, it is clear that to minimize the size of $\ker(\vb{M})$, one can choose $\vb{M} = \vb{G}_{\vb{d}}$.
In this way, the sampling of $\vb{e}_j$'s can be removed.

\begin{proof}
    Recall that the $j$-th row of $\vb{M}$ is obtained in the following way.
    First, we eliminate rows in $\vb{H}$ that are orthogonal to $\vb{d}$, which gives $\vb{H}_{\vb{d}}$.
    Then, we eliminate rows in $\vb{H}_{\vb{d}}$ that are orthogonal to $\vb{e}_j$, which gives $\vb{H}_{\vb{d}, \vb{e}_j}$. 
    Finally, we sum up the rows in $\vb{H}_{\vb{d}, \vb{e}_j}$, which gives $\vb{m}_j^T$.
    Equivalently, we have 
    \begin{align}\label{eq:rows_of_M}
        \vb{m}_j^T = (\vb{H}_{\vb{d}}\ \vb{e})^T \vb{H}_{\vb{d}} = \vb{e}^T \vb{G}_{\vb{d}} \ .
    \end{align}
    To see this, first observe that $\vb{H}_{\vb{d}}\ \vb{e}$ has ones in the positions where the corresponding rows are not orthogonal to $\vb{e}_j$.
    Then, $(\vb{H}_{\vb{d}}\ \vb{e})^T \vb{H}_{\vb{d}}$ selects and sums up the rows in $\vb{H}_{\vb{d}, \vb{e}_j}$.

    According to Eq.~\eqref{eq:rows_of_M}, the rows of $\vb{M}$ are linear combinations of rows of $\vb{G}_{\vb{d}}$ and thus are in the row space of $\vb{G}_{\vb{d}}$.
\end{proof}

For completeness, we also give the success probability that the real secret $\vb{s}$ lies in $\ker(\vb{M})$.

\begin{proposition}[Theorem~3.1 in Ref.~\cite{kahanamoku-meyer_forging_2023} restated]
    \label{prop:correct_d_KM}
    Given $(\vb{H}, \vb{s}) \in \calH_{n, m, q}^{\QRC}$, randomly sample a vector $\vb{d} \in \{0, 1\}^n$ and let $\vb{M}$ be the binary matrix obtained from Meta-Algorithm~\ref{alg:extract_secret_KM}. 
    If $\vb{G}_{\vb{s}} \vb{d} = 0$, then we have $\vb{M}\vb{s} = \vb{0}$, which happens with probability $1/2$ over all choices of $\vb{d}$. 
\end{proposition}

\begin{proof}
    First, note that for the $i$-th row of $\vb{M}$,
    \begin{align}
        \vb{m}_i \cdot \vb{s} = \sum_{\substack{\vb{p}^T \in \row(\vb{H}) \\ \vb{p} \cdot \vb{d} = \vb{p} \cdot \vb{e}_i = 1}} \vb{p} \cdot \vb{s} = \sum_{\vb{p} \in \row(\vb{H})} (\vb{p} \cdot \vb{s}) (\vb{p} \cdot \vb{d}) (\vb{p} \cdot \vb{e}_i) \ ,
    \end{align}
    since each term equals 1 if and only if it has inner product 1 with $\vb{d}$, $\vb{e}$ and $\vb{s}$ simultaneously. 
    The above transformation is to take the conditions in the summation up to the summand, and we can take the term $\vb{p} \cdot \vb{s} = 1$ down to the summation. That is, we can write 
    \begin{align}
        \vb{m}_i \cdot \vb{s} = \sum_{\vb{p} \in \row(\vb{H}_{\vb{s}})} (\vb{p} \cdot \vb{d}) (\vb{p} \cdot \vb{e}_i) \ ,
    \end{align}
    which can be seen to be a quantity only depending on $\vb{H}_{\vb{s}}$. 
    Further observe that the above is the inner product between $\vb{H}_{\vb{s}} \vb{d}$ and $\vb{H}_{\vb{s}} \vb{e}_i$, i.e., 
    \begin{align}
        \vb{m}_i \cdot \vb{s} = (\vb{H}_{\vb{s}} \vb{e}_i) \cdot (\vb{H}_{\vb{s}} \vb{d}) = \vb{e}_i^T \vb{G}_{\vb{s}} \vb{d} \ .
    \end{align}
    Therefore, if $\vb{G}_{\vb{s}} \vb{d} = \vb{0}$, then $\vb{m}_i \cdot \vb{s} = 0$ for every $i$, which means $\vb{M} \vb{s} = \vb{0}$.
    That is, the verifier's secret lies in the kernel of $\vb{M}$ if $\vb{d}$ lies in the kernel of $\vb{G}_{\vb{s}}$.
    If $\vb{H}_{\vb{s}}$ generates a QRC, then $\rank(\vb{G}_{\vb{s}}) = 1$.
    Then, the probability that $\vb{d}$ lies in the kernel of $\vb{G}_{\vb{s}}$ is $2^{n-1}/2^n = 1/2$.    
\end{proof}

\subsection{Classical sampling}
\label{app:more_attack_classical_sampling}

Here, we give two classical sampling algorithms that given an IQP circuit $C$ and a candidate set $S = \{ \vb{s}_1, \cdots, \vb{s}_t \}$ with $t \leq n$ as input, output samples that have the correct correlation function on all candidate secrets in the set.
We first consider a simple case here, where the Gram matrix $\vb{G}_{\vb{s}_i} = \vb{H}_{\vb{s}_i}^T \vb{H}_{\vb{s}_i}$ associated with each candidate secret $\vb{s}_i$ has the same rank.
Then, if the samples are from a quantum computer, the probability bias relative to every candidate secret should be the same, denoted as $\beta$.
Given this candidate set, a classical prover can use Algorithm~\ref{alg:sampling_candidate_same_bias} to mimic the quantum behavior.
Here, the matrix $\vb{S}$ is defined to be a $t \times n$ matrix whose $i$-th row is $\vb{s}_i^T$.
The output bit strings will have probability $\beta$ to be orthogonal to all vectors in the candidate set $S$, and probability $1 - \beta$ to have inner product 1 with them. 
Therefore, the generated samples will have correct bias with every vector in the candidate set and hence the correct correlation function.
The condition for Algorithm~\ref{alg:sampling_candidate_same_bias} to work is that the all-ones vector $\vb{1}$ needs to be in the column space of $\vb{S}$. 
Otherwise, the specific solution $\vb{y}'$ cannot be found.
A sufficient condition is that the candidate vectors are linearly independent.
Then, the matrix $\vb{S}$ will have full row rank, and the all-ones vector $\vb{1}$ will be in the column space of $\vb{S}$.

\begin{algorithm}[H]
    \centering
    \begin{algspec}
        \textbf{Parameter:} number of samples $T$. 
        \begin{algorithmic}[1]
        \Procedure{ClassicalSampling}{$\vb{S}, \beta$}
            \State{Solve $\vb{S y} = \vb{1}$ for a specific solution $\vb{y}'$.}
            \State{Find the basis $\{\vb{y}_1, \cdots, \vb{y}_k\}$ of $\ker(\vb{S})$} \Comment{$k$ is the dimension of $\ker(\vb{S})$.}
            \For{$j = 1, 2, \cdots, T$}
                \State{Randomly sample $(c_1, \cdots, c_k) \in \FF^k$.}
                \State{With probability $\beta$, set $\vb{x}_j \gets \sum_{i=1}^k c_i \vb{y}_i$.}
                \State{With probability $1 - \beta$, set $\vb{x}_j \gets \vb{y}' + \sum_{i=1}^k c_i \vb{y}_i$.}
            \EndFor
            \State{\textbf{return} $\vb{x}_1, \cdots, \vb{x}_T$}
        \EndProcedure
        \end{algorithmic}
    \end{algspec}
    \caption{The \textsc{ClassicalSampling} subroutine for the candidate set where all secrets are associated with the same bias.}
    \label{alg:sampling_candidate_same_bias}
\end{algorithm}

Next, we do not require the associated biases $\beta_1, \cdots, \beta_t$ to be the same.
We present a similar sampling algorithm to Algorithm~\ref{alg:sampling_candidate_same_bias} to output samples that mimic what a quantum computer will output.
For the sake of illustration, we assume that the associated biases are all different, denoted as $\{\beta_1, \beta_2, \cdots, \beta_t\}$, but the following discussion can be easily generalized to the case where some $\beta_i$'s are the same.
As before, the attacker does not have extra information to judge which one is the correct secret, even though the correct secret is in the candidate set.
So he would have to generate samples that have bias $\beta_1$ with $\vb{s}_1$, $\beta_2$ with $\vb{s}_2$, and so on.
Below, the algorithm for generating such samples is shown in Algorithm~\ref{alg:sampling_candidate_diff_bias}.
Again, we transform the set $S$ into a $t \times n$ matrix $\vb{S}$.

The correctness of the sampling algorithm can be easily seen via a sanity check.
But one problem is whether the linear system $\vb{Sy} = \vb{b}_j$ has solutions or not.
If nonzero solutions can be found for every linear systems, then $\vb{b}_j$'s are all in the column space of $\vb{S}$, which implies that the rank of $\vb{S}$ is $t$.
Since there are $t$ rows in $\vb{S}$, a necessary and sufficient condition for the sampling algorithm to work is that $\{\vb{s}_1, \cdots, \vb{s}_t\}$ are linearly independent.
This condition can be relaxed if some of the $\beta_j$'s are the same.

\begin{algorithm}[H]
    \centering
    \begin{algspec}
        \textbf{Input:} a binary matrix $\vb{S} \in \FF^{t\times n}$; biases $\beta_1 > \beta_2 > \cdots > \beta_t$. \\
        \textbf{Parameter:} number of samples $T$. \\
        \textbf{Output:} $\vb{x}_1, \cdots, \vb{x}_T \in \FF^n$.
        \begin{algorithmic}[1]
            \State{Find the basis of $\ker(\vb{S})$, denoted as $\{\vb{y}_1, \cdots, \vb{y}_k\}$.}
            \For{$j = 1, 2, \cdots, t$}
                \State{Define $\vb{b}_j$ to be a binary vector whose last $j$ entries are all zero.}
                \State{Solve a specific solution $\vb{y}_j'$ for $\vb{Sy} = \vb{b}_j$.}
            \EndFor
            \For{$i = 1, 2, \cdots, T$}
                \State{Randomly sample $(c_1, \cdots, c_k) \in \FF^k$.}
                \State{With probability $\beta_t$, set $\vb{x}_i \gets \sum_{i=1}^k c_i \vb{y}_i$.}
                \State{With probability $\beta_{t-1} - \beta_{t}$, set $\vb{x}_i \gets \vb{y}'_t + \sum_{i=1}^k c_i \vb{y}_i$.}
                \State{With probability $\beta_{t-2} - \beta_{t-1}$, set $\vb{x}_i \gets \vb{y}'_{t-1} + \sum_{i=1}^k c_i \vb{y}_i$.}
                \State{\vdots}
                \State{With probability $\beta_{1} - \beta_{2}$, set $\vb{x}_i \gets \vb{y}'_{2} + \sum_{i=1}^k c_i \vb{y}_i$.}
                \State{With probability $1 - \beta_{1}$, set $\vb{x}_i \gets \vb{y}'_{1} + \sum_{i=1}^k c_i \vb{y}_i$.}
            \EndFor
            \State{\textbf{return} $\vb{x}_1, \cdots, \vb{x}_T$}
        \end{algorithmic}
    \end{algspec}
    \caption{The \textsc{ClassicalSampling} subroutine for the candidate set where all vectors are associated with different biases.}
    \label{alg:sampling_candidate_diff_bias}
\end{algorithm}

\subsection{The probability of choosing a good $\vb{d}$}
\label{app:prob_good_d}

Here, we prove Proposition~\ref{prop:correct_d}.

\begin{proof}
    First, $\vb{G}_{\vb{s}} \vb{d} = \vb{H}_{\vb{s}}^T \vb{H}_{\vb{s}} \vb{d}$ and $\vb{H}_{\vb{s}} \vb{d}$ is a vector, where the positions of ones gives the indices of the rows in $\vb{H}_{\vb{s}}$ that have inner product 1 with $\vb{d}$.
    Therefore, the ones of $\vb{H}_{\vb{s}} \vb{d}$ correspond to the rows in $\vb{H}$ that have inner product 1 with both $\vb{s}$ and $\vb{d}$.
    Moreover, if the vector $\vb{H}_{\vb{s}} \vb{d}$ is multiplied to $\vb{H}_{\vb{s}}^T$ on the right, then those rows are summed up, i.e., 
    \begin{align}
        \vb{G}_{\vb{s}} \vb{d} = \vb{H}_{\vb{s}}^T \vb{H}_{\vb{s}} \vb{d} = \sum_{\substack{\vb{p}^T \in \row(\vb{H}) \\ \vb{p} \cdot \vb{d} = \vb{p} \cdot \vb{s} = 1}} \vb{p} \ .
    \end{align}
    Similarly, we have
    \begin{align}
        \vb{G}_{\vb{d}} \vb{s} = \vb{H}_{\vb{d}}^T \vb{H}_{\vb{d}} \vb{s} = \sum_{\substack{\vb{p}^T \in \row(\vb{H}) \\ \vb{p} \cdot \vb{d} = \vb{p} \cdot \vb{s} = 1}} \vb{p} \ .
    \end{align}
    Thus, $\vb{G}_{\vb{s}} \vb{d} = \vb{G}_{\vb{d}} \vb{s}$.
    If we want $\vb{s}$ to lie in $\ker(\vb{G}_{\vb{d}})$, then $\vb{d}$ needs to lie in $\ker(\vb{G}_{\vb{s}})$, which happens with probability 
    \begin{align}
        \frac{2^{n-g}}{2^n} = 2^{-g} \ ,
    \end{align}
    for a random $\vb{d}$.
\end{proof}

\subsection{Size of $\ker(\vb{G}_{\vb{d}})$}
\label{app:kernel_size}

Here, we prove Theorem~\ref{thm:kernel_size}.

\begin{proof}
    First, observe that the rows in $\vb{G}_{\vb{d}}$ are formed by linear combination of rows in $\vb{H}_{\vb{d}}$, which means the rows space of $\vb{G}_{\vb{d}}$ is no larger than that of $\vb{H}_{\vb{d}}$, and $\rank(\vb{G}_{\vb{d}}) \leq \rank(\vb{H}_{\vb{d}})$.
    So, the dimension of $\ker(\vb{G}_{\vb{d}})$ is
    \begin{align}
        n - \rank(\vb{G}_{\vb{d}}) \geq n - \rank(\vb{H}_{\vb{d}}) \ .
    \end{align}
    In expectation, the number of rows $r(\vb{H}_{\vb{d}})$ in $\vb{H}_{\vb{d}}$ is
    \begin{align}
        \EE_{\vb{d}} [r(\vb{H}_{\vb{d}})] = \EE_{\vb{d}} \left[ \sum_{\vb{p}^T \in \row(\vb{H})} \vb{p} \cdot \vb{d} \right] 
        = \sum_{\vb{p}^T \in \row(\vb{H})} \EE_{\vb{d}} \left[ \vb{p} \cdot \vb{d} \right] = \frac{m}{2} \ ,
    \end{align}
    since $\EE_{\vb{d}} \left[ \vb{p} \cdot \vb{d} \right] = 1/2$ for every row $\vb{p}^T$.
    Since $\rank(\vb{H}_{\vb{d}}) \leq r(\vb{H}_{\vb{d}})$, we have $\dim(\ker(\vb{G}_{\vb{d}})) \geq n-m/2$ in expectation.
\end{proof}

\section{Discussion on new classical attacks}
\label{app:more_classical_attacks}

{Here, we complement the discussion in the main text about the classical attacks in Ref.~\cite{gross_secret_2023}.
Using those as guidelines, we propose the structural properties of the IQP matrix $\vb{H}$ that could possibly circumvent all currently known classical attacks.}

\paragraph{Linearity Attack and its variants.}

To invalidate the Linearity Attack, we generally require $\lambda_1 := n - m/2$ to be large, as shown in Theorem~\ref{thm:kernel_size}.
Ref.~\cite{gross_secret_2023} also proposed two variants of the Linearity Attack, namely, the Lazy Linearity Attack and the Double Meyer Attack.

The Lazy Linearity Attack exploits the statistical fluctuation of the dimension of $\ker(\vb{G}_{\vb{d}})$.
One may expect that $\dim(\ker(\vb{G}_{\vb{d}}))$ is approximately a Gaussian random variable with mean $\lambda_1 = n - m/2$.
Therefore, the basic idea is to repeatedly sample $\vb{d}$, check the dimension of $\ker(\vb{G}_{\vb{d}})$ and explore the kernel only if its dimension is lower than certain threshold.
In particular, the expected number of iterations is given by~\cite{gross_secret_2023}
\begin{align}
    E \sim \frac{2^g}{C_{\lambda_1, \sqrt{m}/2}(A)} \ ,
\end{align}
where $A$ is the threshold set for the dimension of $\ker(\vb{G}_{\vb{d}})$ and $C_{\mu, \sigma}$ is the CDF for the Gaussian distribution $N(\mu, \sigma^2)$.
Let $\lambda_2 := \log{E}$. 
Then, when $g = 10$, $\lambda_1 = 150$ and $A = 8$, the scaling of $\lambda_2$ versus the number of rows $m$ is plotted below.
If one chooses $m \leq 1200$, then one can expect $\lambda_2 > 60$.

\begin{figure}[h]
    \centering
    \includegraphics[width=0.5\linewidth]{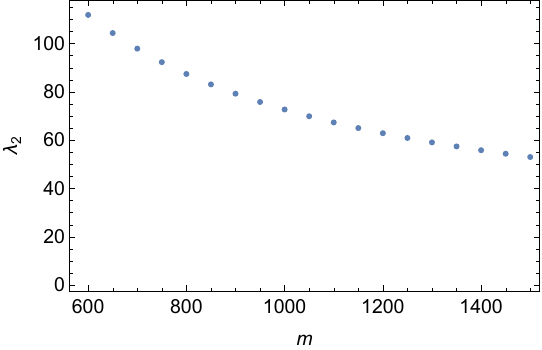}
    \caption{The logarithmic scaling of the expected number of iterations $\lambda_2$ versus the number of rows $m$ in the Lazy Linearity Attack.}
\end{figure}

For the Double Meyer Attack, the basic idea is that if one chooses $k$ random vectors $\vb{d}_1, \cdots, \vb{d}_k$, then intersection of all the kernels $\ker(\vb{G}_{\vb{d}_i})$ may exponentially shrink as $k$ increases.
When $\lambda_1 = 150$, the analysis in Ref.~\cite{gross_secret_2023} suggests that $k = 6$ is sufficient to make the intersection small enough to be explored.
Since the probability that the real secret lies in the intersection is given by $2^{-gk}$, the expected number of iterations is $\sim 2^{gk}$.
However, as pointed out by Ref.~\cite{gross_secret_2023}, the Double Meyer Attack might be infeasible in practice.
For example, choosing $g = 10$ means that the expected number of iterations is $\sim 2^{60}$, which suffices to make the Double Meyer Attack fail in practice.

\paragraph{Analysis of Hamming's razor.}

Without further structure in $\vb{H}$, the Hamming's razor has a good chance to work if one wants to invalidate the Linearity Attack and its variants.
To understand this, let $\vb{H} = (\vb{H}_1, \vb{H}_2)$, where $\vb{H}_1$ consists of the first $g+d$ columns and $\vb{H}_2$ consists of the remaining.
After removing $p$ fraction of rows, solving for the kernel of the remaining submatrix {$\vb{H}' = (\vb{H}'_1, \vb{H}'_2)$} is equivalent to solving for the following two linear systems:
\begin{align}\label{eq:hamming_razor_linear_systems}
    \vb{H}_1' \vb{v}_1 &= \vb{0} & \vb{H}_2' \vb{v}_2 &= \vb{0} \ ,
\end{align}
where $\vb{H}_1'$ is the remaining submatrix from $\vb{H}_1$ and $\vb{H}_2'$ is the remaining submatrix from $\vb{H}_2$.
The attacker wants to find a threshold $p^*$ so that only the second linear system has nontrivial solutions.
In this way, the solution vector {will be of the form} $\vb{v} = (\vb{v}_1^T, \vb{v}_2^T)^T \in \vb{0}_{g+d} \oplus \FF_2^{n-g-d}$, which indicates the indices of the redundant rows since $\vb{H} \vb{v} = \begin{pmatrix} \vb{0} \\ \vb{c}_2 \end{pmatrix}$.

Note that these two linear systems contain the same number of equations.
For the linear system from $\vb{H}_1'$, the number of variables is $g+d$, which is $\sim \frac{m_1}{2}$ when one takes the maximum possible value of $d$.
For the linear system from $\vb{H}_2'$, the number of variables is $n - g - d \sim \frac{m_2}{2} + \lambda_1$, where we used $\lambda_1 = n - m/2$.
So, there are generally more variables in the second linear system, {since we set $\lambda_1$ to be large to invalidate the Linearity Attack and its variants}.
What is worse, the all-zeros block in $\vb{H}_2'$ imposes no constraint.
Thus, when the fraction $p$ is steadily increased, one can expect a threshold $p^*$ so that only the second linear system has a nontrivial solution.

\paragraph{Invalidating the Hamming's razor.}

Here, we propose the following structural constraints on $\vb{H}$ to invalidate the Hamming's razor.
Essentially, we want to make $\vb{H}_1$ sparse, so that the first linear system also has solutions when the second one does, despite that it has relatively more constraints than the second one.

There are two parts to be considered.
For the $\vb{D}$ matrix, we are not aware of methods for sampling sparse generator matrices for doubly-even codes, so we use code concatenation to achieve our purpose.
We discuss the specific construction later.
For the $(\vb{A}, \vb{B})$ matrix, we first set it to be a random $t$-sparse matrix.
In our numerical experiment, we simply set $t$ to be 1.
Then, we postselect on the case where $(\vb{B}, \vb{C})$ is full rank.
Finally, we only need to modify one column of $\vb{A}$ so that $\vb{R}_{\vb{s}} \, \vb{s} = \vb{0}$ is satisfied.

Now, we turn to discuss the construction of the $\vb{D}$ matrix.
Let $\vb{D}_1, \cdots, \vb{D}_k \in \FF_2^{m_0 \times d_0}$ be generator matrices for doubly-even codes.
Here, we require $m_0 k = m_1$ and $d_0 k \geq d$.
Of course, in general, $\vb{D}_i$ can be of different sizes, but for the sake of illustration, we let their sizes be the same.
For a vector $\vb{a} \in \FF_2^{d_0 k}$, we split it into $k$ blocks,
\begin{align}
    \vb{a} = 
    \begin{pmatrix}
        \vb{a}^{(1)} \\ \vdots \\ \vb{a}^{(k)}
    \end{pmatrix} \ ,
\end{align}
where $\vb{a}^{(j)} \in \FF_2^{d_0}$.
Then, we apply $\vb{D}_j$ to the $j$-th block $\vb{a}^{(j)}$, resulting in
\begin{align}
    \begin{pmatrix}
        \vb{D}_1 \vb{a}^{(1)} \\ \vdots \\ \vb{D}_k \vb{a}^{(k)} 
    \end{pmatrix} \ ,
\end{align}
a vector of length $m_1 = m_0 k$, which is also doubly-even.
To define the matrix $\vb{D}$, let $\vb{K}_{\mathrm{in}} \in \FF_2^{d_0 k \times d}$ be a random matrix of full column rank with $d \leq d_0 k$.
Then, for each vector in the columns of $\vb{K}_{\mathrm{in}}$, we split it into $k$ blocks and apply $\vb{D}_j$ to its $j$-th block for $j \in [k]$.
Formally, define the encoding matrix
\begin{align}
    \vb{E}_c = \diag(\vb{D}_1, \cdots, \vb{D}_k) \ .
\end{align}
Then, $\vb{D}$ is constructed to be $\vb{D} = \vb{E}_c \vb{K}_{\mathrm{in}}$.
It is not hard to verify that $\vb{D}$ is a generator matrix for a doubly-even code.
The sparsity of $\vb{D}$ can be controlled by the sparsity of $\vb{K}_{\mathrm{in}}$.
To see this, consider in the extreme case where $\vb{K}_{\mathrm{in}}$ is the identity matrix.
Then, $\vb{D}$ is just a block diagonal matrix.

\paragraph{Numerics.}
We test the above ideas numerically~\cite{challenge}.
We generate instances with the $\vb{D}$ and $(\vb{A}, \vb{B})$ matrices as described above.
The parameter set is $n = 700, m = 1200, g = 10, m_1 = 300, d = 135, m_0 = 20$ and $d_0 = 9$.
The Radical Attack fails to recover the secret since the dimension of $\ker(\vb{H}^T \vb{H})$ is about 1.
Although $\lambda_1 = n - m/2$ for this set of parameters is 100, we observe that the mean value of the dimension of $\ker(\vb{G}_{\vb{d}})$ is about 150.
This implies that the Linearity Attack and its variants should also fail in practice, though we did not put enough computing resources to extrapolate the actual scaling of these attacks (specifically, the Double Meyer Attack).
Finally, for Hamming's razor, both linear systems in \cref{eq:hamming_razor_linear_systems} begin to have solutions when $p \approx 0.35$ and we could not find a threshold $p^*$ so that only the second linear system has solutions.
So, the Hamming's razor also fails in this case.

\end{document}